\documentclass[preprint,9pt]{sigplanconf}
\usepackage{leftidx}
\usepackage{mathtools,amssymb,amsthm}
\usepackage{nicefrac}
\usepackage{infer}          % Inference rules
\usepackage{stmaryrd}   % Semantics symbol
\usepackage{xspace}
\usepackage{graphicx}
\usepackage{scalerel}     % Scales D of declarations
\usepackage[textsize=small]{todonotes}
\usepackage[final,colorlinks=true,linktocpage=true,hyperfootnotes]{hyperref}
\usepackage{xspace}
\usepackage{ctable}        % Table of wp rules

\usepackage{tikz}
\usetikzlibrary{arrows,automata}

\theoremstyle{plain}
\newtheorem{definition}{Definition}[section]
\newtheorem{theorem}{Theorem}[section]
\newtheorem{lemma}{Lemma}[section]

\newtheorem{fact}{Fact}[section]
\theoremstyle{remark}
\newtheorem{example}{Example}

\allowdisplaybreaks[4]

\newcommand{\ie}{\textit{i.e.}\@\xspace}
\newcommand{\wrt}{w.r.t.\@\xspace}
\newcommand{\eg}{\textit{e.g.}\@\xspace}
\newcommand{\cf}{\textit{cf.}\@\xspace}

\newcommand{\Reals}{\ensuremath{\mathbb R}\xspace}
\newcommand{\PosReals}{\ensuremath{\mathbb{R}_{\geq 0}}\xspace}
\newcommand{\Nats}{\ensuremath{\mathbb N}\xspace}
\newcommand{\SD}[1]{\ensuremath{\mathcal{D}\!\left( #1 \right)}\xspace}

\newcommand{\BEX}{\ensuremath{\mathbb{E}_{\leq 1}}\xspace}
\newcommand{\UEX}{\ensuremath{\mathbb{E}}\xspace}

\newcommand{\lfpsymbol}{\mathit{lfp}}
\newcommand{\gfpsymbol}{\mathit{gfp}}
\newcommand{\lfp}[2]{\ensuremath{\lfpsymbol_{#1}\left( #2 \right)}\xspace}
\newcommand{\gfp}[2]{\ensuremath{\gfpsymbol_{#1}\left( #2 \right)}\xspace}

\newcommand{\eqdef}{\triangleq}
\newcommand{\mydot}{\raisebox{-0.5pt}{\scalebox{0.4}{$\, \bullet$}~}}

\newcommand{\CteFun}[1]{\ensuremath{\mathbf{\boldsymbol{#1}}}\xspace}
\newcommand{\tuple}[1]{\langle #1 \rangle}
\newcommand{\EV}[2]{\mathbf{E}_{#1} \left( #2 \right)}
\newcommand{\CharFun}[1]{\ensuremath{ \left[ #1 \right] }\xspace}
\newcommand{\To}{\rightarrow}
\newcommand{\chain}[2]{\ensuremath{{#1}_0 #2 {#1}_1 #2 \cdots}}
\newcommand{\gr}{\varphi}

\newcommand{\ucont}[2]{#1 \overset{\text{\tiny upp-cont}}{\rightarrow} #2}
\newcommand{\lcont}[2]{#1 \overset{\text{\tiny low-cont}}{\rightarrow} #2}
\newcommand{\leftsuper}[2]{{\vphantom{\left\langle #2\right \rangle}}^{#1}{\left\langle #2 \right\rangle}}

\newcommand{\true}{\ensuremath{\mathsf{true}}\xspace}
\newcommand{\false}{\ensuremath{\mathsf{false}}\xspace}

\newcommand{\sem}[1]{\llbracket #1 \rrbracket}
\newcommand{\wpsymbol}{\textnormal{\textsf{wp}}\xspace}
\newcommand{\boldwpsymbol}{\ensuremath{\textnormal{\textsf{\textbf{wp}}}\xspace}}
\renewcommand{\wp}[1]{\ensuremath{\wpsymbol \!\left[ #1 \right] }\xspace}
\newcommand{\wpd}[2]{\wpsymbol \!\left[ #1, #2 \right] \xspace}
\newcommand{\wlpsymbol}{\textnormal{\textsf{wlp}}\xspace}
\newcommand{\wlp}[1]{\ensuremath{\wlpsymbol \!\left[ #1 \right] }\xspace}
\newcommand{\wlpd}[2]{\wlpsymbol \!\left[ #1, #2 \right] \xspace}
\newcommand{\wllpsymbol}{\textnormal{\textsf{w}(\textsf{l})\textsf{p}}\xspace}
\newcommand{\wllp}[1]{\ensuremath{\wllpsymbol \!\left[ #1 \right] }\xspace}
\newcommand{\wllpd}[2]{\wllpsymbol \!\left[ #1, #2 \right] \xspace}
\newcommand{\ewp}[2]{\ensuremath{\wpsymbol \!\left[ #1 \right]_{#2}^{\sharp} }\xspace}
\newcommand{\ewlp}[2]{\ensuremath{\wlpsymbol \!\left[ #1 \right]_{#2}^{\sharp} }\xspace}
\newcommand{\ewllp}[2]{\ensuremath{\wllpsymbol \!\left[ #1 \right]_{#2}^{\sharp} }\xspace}

\newcommand{\subst}[2]{\!\left[{#1}/{#2}\right]}
\newcommand{\decl}{\scaleobj{0.8}{\mathcal{D}}}

\newcommand{\ToExp}[1]{\CharFun{#1}}

\newcommand\restr[2]{{% we make the whole thing an ordinary symbol
  \left.\kern-\nulldelimiterspace % automatically resize the bar with \right
  #1 % the function
  \vphantom{\big|} % pretend it's a little taller at normal size
  \right|_{#2} % this is the delimiter
  }}

\newcommand{\Var}{\ensuremath{\mathcal{V}}\xspace}

\newcommand{\PName}{\ensuremath{\mathit{P}}\xspace}
\newcommand{\State}{\ensuremath{\mathcal{S}}\xspace}
\newcommand{\Expr}{\ensuremath{\mathcal{E}}\xspace}
\newcommand{\Cmd}{\ensuremath{\mathcal{C}}\xspace}
\newcommand{\Skip}{{\sf skip}\xspace}
\newcommand{\Abort}{{\sf abort}\xspace}
\newcommand{\Ass}[2]{\ensuremath{{#1} \mathrel{\coloneqq} {#2}}}
\newcommand{\If}{{\sf if}\xspace}
\newcommand{\Else}{{\sf else}\xspace}
\newcommand{\Cond}[3]{\If \, (#1) \,\{#2\} \,\Else \, \{#3\}}
\newcommand{\Call}[1]{{\sf call} \, #1}
\newcommand{\Calln}[3]{{\sf call}_{#2}^{#3} \, #1}

\newcommand{\While}{{\sf while}\xspace}
\newcommand{\Do}{{\sf do}\xspace}
\newcommand{\WhileDo}[2]{\While \, (#1) \, \Do \, \{#2\}}
\newcommand{\PChoiceSym}[1]{\ensuremath{[#1]}}
\newcommand{\PChoice}[3]{ \{ #1 \} \: \PChoiceSym{#2} \: \{#3\}}

\newcommand{\SEnv}{\ensuremath{\mathsf{SEnv}}\xspace}
\newcommand{\LSEnv}{\ensuremath{\mathsf{LSEnv}}\xspace}
\newcommand{\RtEnv}{\ensuremath{\mathsf{RtEnv}}\xspace}

\newcommand{\prog}[2]{\tuple{#1,#2}}

\newcommand{\GCL}{\textsf{GCL}\xspace}
\newcommand{\pGCL}{\textsf{pGCL}\xspace}
\newcommand{\pRGCL}{\textsf{pRGCL}\xspace}

\newcommand{\lrule}[1]{\textnormal{\small \textsf{[#1]}}\xspace}

\newcommand{\derivsymbol}{\:\Vdash\:}
\newcommand{\deriv}[2]{\ensuremath{#1 \derivsymbol  #2}}
\newcommand{\by}[1]{\text{\small \{#1{}\}}}

\newcommand{\Automaton}{\ensuremath{\mathfrak P}\xspace}

\newcommand{\rewardsymbol}{\ensuremath{\mathsf{rew}}\xspace}
\newcommand{\rew}[1]{\ensuremath{\rewardsymbol\left(#1\right)}\xspace}
\newcommand{\Paths}[1]{\ensuremath{\mathsf{Paths}^{#1}}\xspace}
\newcommand{\PathsP}{\ensuremath{\Paths{\Automaton}}\xspace}
\newcommand{\ExpRew}[2]{\ensuremath{\mathsf{ExpRew}^{#1}\left({#2}\right)}\xspace}
\newcommand{\Prob}[2]{\ensuremath{\mathsf{Prob^{#1}\left({#2}\right)}}\xspace}

\newcommand{\LabUsed}{\ensuremath{\mathsf{Lab}_{\ast}}\xspace}
\newcommand{\Term}{\ensuremath{{\downarrow}}\xspace}
\newcommand{\Sink}{\ensuremath{\mathsf{Term}}\xspace}
\newcommand{\StmtOfLabelSymbol}{\ensuremath{\mathsf{stmt}}\xspace}
\newcommand{\StmtOfLabel}[1]{\ensuremath{\StmtOfLabelSymbol\left(#1\right)}\xspace}
\newcommand{\SuccOneSymbol}{\ensuremath{\mathsf{succ}_{1}}\xspace}
\newcommand{\SuccTwoSymbol}{\ensuremath{\mathsf{succ}_{2}}\xspace}
\newcommand{\SuccOne}[1]{\ensuremath{\SuccOneSymbol\left(#1\right)}\xspace}
\newcommand{\SuccTwo}[1]{\ensuremath{\SuccTwoSymbol\left(#1\right)}\xspace}
\newcommand{\OpState}[2]{\ensuremath{\left\langle {#1},\, {#2} \right\rangle}\xspace}
\newcommand{\OpTrans}[3]{\ensuremath{{#1} \xrightarrow{#2} #3}}
\newcommand{\Init}{\ensuremath{\mathsf{init}}}
\newcommand{\OPRMC}[3]{\ensuremath{\Automaton_{#2}^{#3}\left\llbracket #1 \right\rrbracket}\xspace}

\newcommand{\eetsymbol}{\ensuremath{\textnormal{\textsf{ert}}}}
\newcommand{\boldeetsymbol}{\ensuremath{\textnormal{\textsf{\textbf{ert}}}}}
\newcommand{\eet}[1]{\ensuremath{\eetsymbol\left[{#1}\right]}}
\newcommand{\eetd}[2]{\ensuremath{\eetsymbol\left[{#1},{#2}\right]}}
\newcommand{\ctert}[1]{\ensuremath{\mathbf{#1}}}
\newcommand{\rt}{t}
\newcommand{\Runtimes}{\ensuremath{\mathbb{T}}}
\newcommand{\eeet}[2]{\ensuremath{\eetsymbol\left[{#1}\right]_{#2}^{\sharp}}}
\newcommand{\eeetd}[3]{\ensuremath{\eetsymbol\left[{#1},{#2}\right]_{#3}^{\sharp}}}
\newcommand{\ctertenv}[1]{\ensuremath{\underline{\mathbf{#1}}}}

\begin{document}

\setlength{\pdfpageheight}{\paperheight}
\setlength{\pdfpagewidth}{\paperwidth}

\conferenceinfo{LICS '16}{July 5--8, 2016, New York City, State of New York, USA} 
\copyrightyear{2016} 
\copyrightdata{978-1-nnnn-nnnn-n/yy/mm} 
\doi{nnnnnnn.nnnnnnn}

% Uncomment one of the following two, if you are not going for the 
% traditional copyright transfer agreement.

%\exclusivelicense                % ACM gets exclusive license to publish, 
                                  % you retain copyright

%\permissiontopublish             % ACM gets nonexclusive license to publish
                                  % (paid open-access papers, 
                                  % short abstracts)

\titlebanner{}        % These are ignored unless
\preprintfooter{Reasoning about Recursive Probabilistic Programs}   % 'preprint' option specified.

\title{Reasoning about Recursive Probabilistic Programs\thanks{This work was supported by the Excellence Initiative of the German federal and state government.}}
%\subtitle{Subtitle Text, if any}

\authorinfo%
{Federico Olmedo \and Benjamin Lucien Kaminski\and Joost-Pieter Katoen \and Christoph Matheja }%
{RWTH Aachen University, Germany}
{\{federico.olmedo, benjamin.kaminski, katoen, matheja\}@cs.rwth-aachen.de}%

\maketitle

\begin{abstract}
  This paper presents a \wpsymbol--style calculus for obtaining expectations on
  the outcomes of (mutually) recursive probabilistic programs. We provide
  several proof rules to derive one-- and two--sided bounds for such
  expectations, and show the soundness of our \wpsymbol--calculus with respect
  to a probabilistic pushdown automaton semantics. 
  We also give a \wpsymbol--style calculus for obtaining bounds on the expected runtime of
  recursive programs that can be used to determine the (possibly infinite)
  time until termination of such programs.
\end{abstract}



\section{Introduction}
\label{sec:intro}
Uncertainty is nowadays more and more pervasive in computer science. 
Applications have to process inexact data from, e.g., unreliable sources such as wireless sensors, machine learning methods, or noisy biochemical reactors.
Approximate computing saves resources such as e.g.\ energy by sacrificing ``strict'' correctness for applications like image processing that can tolerate some defects in the output by running them on unreliable hardware, circuits that every now and then (deliberately) produce incorrect results~\cite{DBLP:conf/oopsla/CarbinMR13}. 
\emph{Probabilistic programming}~\cite{Pfeffer:2016} is a key technique for dealing with uncertainty.  
Put in a nutshell, a probabilistic program takes a (prior) probability distribution as input and obtains a (posterior) distribution.
Probabilistic programs are not new at all; they have been investigated by Kozen~\cite{Kozen:81} and others in the early eighties.
In the last years, the interest in these programs has rapidly grown.
In particular, the incentive by the AI community to use probabilistic programs for describing complex Bayesian networks has boosted the field of probabilistic programming~\cite{DBLP:conf/icse/GordonHNR14}.
Probabilistic programs are used in, amongst others, machine learning, systems biology, security, planning and control, quantum computing, and software--defined networks.
Indeed almost all programming languages, either being functional, object--oriented, logical, or imperative, in the meanwhile have a probabilistic variant. 

This paper focuses on \emph{recursive} probabilistic programs.
Recursion in Bayesian networks where a variable associated with a particular domain entity can depend probabilistically on the same variable associated to a different entity, is ``common and natural''~\cite{DBLP:conf/aaai/PfefferK00}.
Recursive probability models occur in gene regulatory networks that describe (possibly recursive) rule--based dependencies between genes.
Finally, programs describing randomized algorithms are often recursive by nature.
``Sherwood" algorithms exploit randomization to increase efficiency by avoiding or reducing the probability of worst--case behavior.
Varying quicksort by selecting the pivot randomly (rather than doing this deterministically) avoids very uneven splits of the input array.
Its worst--case runtime is the same as the average--case runtime of Hoare's deterministic quicksort since the likelihood of obtaining a quadratic worst--case is significantly lowered~\cite[Sec.\ 2.5]{Mitzenmacher:2005}.
A ``Sherwood" variant of binary search splits the input array at a random position, and yields a similar effect---expected runtimes of worst--, average-- and best--case are aligned~\cite[Sec.\ 11.4.4]{McConnell:2008}. 
``Sherwood" techniques are also useful in selection, median finding, and hashing (such as Bloom filters).

The purpose of this paper is to provide a framework for enabling \emph{formal reasoning about recursive probabilistic programs}.
This rigorous reasoning is important to prove the \emph{correctness} of such programs.
This includes statements about the expected outcomes of recursive probabilistic programs, as well as assertions about their termination probability.
These are challenging problems.
For instance, consider the (at first sight simple) recursive program:
\[
\PName_{\mathsf{rec_3}} \, \rhd \; \;
\PChoice{\Skip}{\nicefrac{1}{2}}{\Call{\PName_{\mathsf{rec_3}}};~ \Call{\PName_{\mathsf{rec_3}}};~ \Call{\PName_{\mathsf{rec_3}}}}
\]
which terminates immediately with probability $\nicefrac 1 2$ or invokes itself three times otherwise.
It turns out that this program terminates with (irrational) probability $\tfrac{\sqrt{5}-1}{2}$ ---the reciprocal of the golden ratio.

Correctness proofs of the ``Sherwood" versions of quicksort and binary search do exist but typically rely on mathematical ad--hoc reasoning about expected values. The aim of this paper is to enable such proofs by means of formal verification of the algorithm itself.

Besides correctness, our interest is in analyzing the \emph{expected runtime} of recursive probabilistic programs in a rigorous manner.
This enables obtaining insight in their \emph{efficiency} and moreover provides a method to show whether the expected time until termination is finite or infinite---a crucial difference for probabilistic programs~\cite{DBLP:conf/mfcs/KaminskiK15,luis}.
Again, analyses of expected runtimes of recursive randomized algorithms do exist using standard mathematics~\cite[Sec.\ 2.5]{Mitzenmacher:2005}, probabilistic recurrence relations~\cite{DBLP:journals/jacm/Karp94}, or dedicated techniques for divide--and--conquer algorithms~\cite{DBLP:journals/dam/Dean06}, usually taking for granted---far from trivial---relationships between the underlying random variables. 
Here the aim is to do this from first principles 
by formal verification techniques, directly on the algorithm.

To accomplish these goals, this paper presents two weakest pre--condition--style calculi for reasoning about recursive probabilistic programs.
The first calculus is an extension of McIver and Morgan's calculus~\cite{McIver:2004} for non--recursive programs and enables obtaining expectations on the outcomes of (mutually) recursive probabilistic programs. 
Compared to an existing extension with recursion~\cite{McIver:2001b}, our approach provides a clear separation between syntax and semantics.
We prove the soundness of our \textsf{wp}--calculus with respect to a probabilistic pushdown automaton semantics.
This is complemented by a set of proof rules to derive one-- and two--sided bounds for expected outcomes of recursive programs.
We illustrate the usage of these proof rules by analyzing the termination probability of the example program above.
Subsequently, we provide a variant of our \textsf{wp}--style calculus for obtaining bounds on the expected runtime of probabilistic programs. 
This extends our recent approach~\cite{Kaminski:ETAPS:2016} towards treating recursive programs.
The application of this calculus includes proving positive almost--sure termination, i.e., does a program terminate with probability one in finite expected time?
Our framework enables (in a very succinct way) establishing a (well--known) relationship between the expected runtime of a probabilistic program with its termination behavior: If an ($\Abort$--free) program has finite expected runtime, then it terminates almost--surely.
We provide a set of proof rules for expected runtimes and show the applicability of our approach by proving several correctness properties as well as the expected runtime of the `Sherwood' variant of binary search.

\paragraph{Organization of the paper.}
\autoref{sec:language} presents our probabilistic programming language with recursion.
\autoref{sec:wp-semantics} presents the \wpsymbol--style semantics for reasoning about program correctness.
\autoref{sec:proof-rules} introduces several proof rules for reasoning about the correctness of recursive programs.
\autoref{sec:eet} presents the expected runtime transformer together with proof rules for recursive programs.
\autoref{sec:operational} describes an operational probabilistic pushdown automata semantics and relates it to the \wpsymbol--style semantics.
\autoref{sec:extensions} discusses some extensions of the results presented in the previous sections.
\autoref{sec:casestudy} presents a detailed analysis of the `Sherwood' variant of binary search.
Finally, \autoref{sec:related} discusses related work and \autoref{sec:conclusion} concludes.
Detailed proofs are provided in the appendix, which is added for the convenience of the reviewer, and will not be part of the final version (if accepted).

\section{Programming Model}
\label{sec:language}
To model our probabilistic recursive programs we consider a simple imperative
language \`a la Dijkstra's Guarded Command Language (\GCL)~\cite{Dijkstra} with
two additional features:
First, a (binary) probabilisitic choice
operator to endow our programs with a probabilistic behavior. For instance,
the program
\[
\PChoice{\Ass{x}{x{+}1}}{\nicefrac{1}{3}}{\Ass{x}{x{-}1}}
\]
either increases $x$ with probability $\nicefrac{1}{3}$ or decreases it with
probability $\nicefrac{2}{3} = 1 - \nicefrac{1}{3}$. Second, we allow for
procedure calls. 
For simplicity, our development assumes the presence of only a single procedure, say
$\PName$. We defer the treatment of multiple (possibly mutually recursive)
procedures to \autoref{sec:extensions}.

Formally, a \emph{command} of our language, coined \pRGCL, is defined by the
following grammar:
$$
\begin{array}{r@{\ \,}c@{\ \,}l@{\qquad}l}
\Cmd  
   &::= & \Skip                    & \mbox{no--op}\\
   &\mid& \Ass{\Var}{\Expr}        & \mbox{assignment}\\
   &\mid& \Abort                    & \mbox{abortion}\\
   &\mid& \Cond{\Expr}{\Cmd}{\Cmd} & \mbox{conditional branching}\\
   &\mid& \PChoice{\Cmd}{p}{\Cmd} & \mbox{probabilistic choice}\\
   &\mid& \Call{\PName}  & \mbox{procedure call}\\
   &\mid& \Cmd; \, \Cmd         & \mbox{sequential composition}   
\end{array}
$$
We assume a set \Var of program \emph{variables} and a set \Expr of
\emph{expressions} over program variables. As usual, we assume that program
\emph{states} are variable valuations, \ie mappings from variables to values;
let $\State$ be the set of program states. Finally, we also assume an
interpretation function $\sem{ \Expr }$ for expressions that maps program states
to values.

No--op, assignments, conditionals and sequential composition are
standard. $\PChoice{c_1}{p}{c_2}$ represents a probabilistic choice: 
it behaves as $c_1$ with probability $p$ and
as $c_2$ with probability $1{-}p$. 
%$\NDChoice{\Cmd_1}{\Cmd_2}$ represents a non--deterministic choice between $\Cmd_1$ and $\Cmd_2$. 
Finally $\Call{\PName}$
makes a (possibly recursive) call to procedure $\PName$. 

For our development we assume that procedure \PName manipulates the global
program state and we thus dispense with parameters and $\textsf{return}$
statements for passing information across procedure calls. The declaration of
\PName consists then of its body and we use $P \: \triangleright \: c$ to de\-note
that $c \in \Cmd$ is the body of $\PName$. We say that a command is
\emph{closed} if it contains no procedure calls. 

A \pRGCL \emph{program} is then given by a pair $\prog{c}{\decl}$, where
$c \in \Cmd$ is the ``main'' command and $\decl \colon \{\PName\} \To \Cmd$ is
the declaration of \PName.\footnote{We chose the declaration of $\PName$ to be a
  mapping from a singleton and not the mere body of $P$ because this minimizes
  the changes to accommodate the subsequent treatment to multiple procedures.} In
order not to clutter the notation, when $c$ is closed we simply write $c$ for
program $\prog{c}{\decl}$, for any declaration $\decl$.

\begin{example}
To illustrate the use of our language consider the following declaration of
a (faulty) recursive procedure for computing the factorial of a natural number stored in $x$:
\begin{align*}
\PName_{\mathsf{fact}} \, \rhd \; \;
&\If~(x \leq 0)~\{\Ass{y}{1} \}~\Else\\[-1pt]
&\quad \bigl\{ \: \{\Ass{x}{x{-}1};\, \Call{\PName_{\mathsf{fact}}};\, \Ass{x}{x{+}1}\}~[\nicefrac{5}{6}]\\[-1pt]
&\quad\ \ \: \{\Ass{x}{x{-}2};\, \Call{\PName_{\mathsf{fact}}};\, \Ass{x}{x{+}2}\};\,
\Ass{y}{y \cdot x} \bigr\}
\end{align*}
In each recursive call $x$ is decreased either by one or two, with
probability $\nicefrac{5}{6}$ and $\nicefrac{1}{6}$, respectively. Therefore some factors might be
missing in the computation of the factorial of $x$. \hfill $\triangle$
\end{example}

As a final remark, observe that the language does not support guarded loops in a
native way because they can be simulated. 
Concretely, the usual guarded loop $\WhileDo{E}{c}$ is simulated by the recursive
procedure $\PName_{\mathsf{while}} \, \rhd \;\; \Cond{E}{c;~ \Call{\PName_{\mathsf{while}}}}{\Skip}$.

\section{Weakest Pre--Expectation Semantics}
\label{sec:wp-semantics}
Inspired by~\citet{Kozen:81}, \citet{McIver:2001b}
generalized Dijk\-stra's weakest pre--condition semantics to (a variant of)
\pRGCL.  In particular, they defined the semantics of recursive programs using
fixed point techniques.  In this section we present a different approach where
the behavior of a recursive program is defined as the limit of its finite
approximations (or truncations) and prove it equivalent to their definition
based on fixed points. 
%%  \todo{\tiny Explain how/why our approach is better/useful/easier to reason about/etc. ---Benni \\@Fede: they are given in--place, along the section.}{}

\subsection{Definition}\label{sec:wp-def}
The wp-semantics over \pRGCL generalizes Dijkstra's weakest
precondition semantics over \GCL twofold:
First, instead of being predicates over program states, pre-- and post--conditions are now (non--negative) real--valued
functions over program states. 
Secondly, instead of merely evaluating a (boolean--valued) post--condition in the final state(s) of a program, we now \emph{measure} the expected value of a (real--valued) post--condition \wrt the distribution of final states. 
Formally, if $f \colon \State \To \Reals^{\geq
  0}$ we let
\begin{equation*}\label{eq:wp-def}
\wpd{c}{\decl}\!(f)  \:\eqdef\: \lambda s \mydot \EV{\sem{c,\decl}(s)}{f}~,
\end{equation*}
where $\sem{c, \decl}(s)$ denotes the distribution of final states from executing
$\prog{c}{\decl}$ in initial state $s$ and $\EV{\sem{c, \decl}(s)}{f}$ denotes the
expected value of $f$ \wrt the distribution of final states
$\sem{c, \decl}(s)$. Consider for instance program 
\[
c_{\mathsf{coins}}\, \boldsymbol{\colon}\; \;
\PChoice{\Ass{x}{0}}{\nicefrac{1}{2}}{\Ass{x}{1}}; 
 \PChoice{\Ass{y}{0}}{\nicefrac{1}{3}}{\Ass{y}{1}}
\]
that flips a pair of fair and biased coins. We have
\begin{align*}
\wp{c_{\mathsf{coins}}}\!(f)\,=\;&\lambda s \mydot \tfrac{1}{6} \, f(s\subst{x,\! y}{0,\! 0}) +
\tfrac{1}{3} \, f(s\subst{x,y}{0,1})\\
& \ \,  + \tfrac{1}{6} \, f(s\subst{x,\! y}{1,\! 0}) +
\tfrac{1}{3} \, f(s\subst{x,y}{1,1})~,
\end{align*}
where $s\subst{x_1,\ldots,x_n}{v_1,\ldots,v_n}$ represents the state obtained by
updating in $s$ the value of variables $x_1,\ldots,x_n$ to $v_1,\ldots,v_n$,
respectively. As above, when $c$ is closed, we usually write
$\wp{c}$ instead of $\wpd{c}{\decl}$, as a declaration $\decl$ plays no role. 

Observe that, in particular, if $\CharFun{A}$ denotes the indicator
function of a predicate $A$ over program states, $\wpd{c}{\decl}\!(\CharFun{A})(s)$ gives
the probability of (terminating and) establishing $A$ after executing
$\prog{c}{\decl}$ from state $s$. For instance we can determine the probability
that the above program $c_{\mathsf{coins}}$ establishes $x=y$ from state $s$ through
\[
\wp{c_{\mathsf{coins}}}\!(\CharFun{x{=}y})(s) \:=\: \tfrac{1}{6} \cdot 1 + \tfrac{1}{3} \cdot 0 +
\tfrac{1}{6} \cdot 0 + \tfrac{1}{3} \cdot 1 \:=\: \tfrac{1}{2}~.
\]

Moreover, for a deterministic program $c$ that from state $s$ terminates in state $s'$, $\sem{c,\decl}(s)$ is the Dirac distribution that concentrates all its mass in $s'$ and $\wpd{c}{\decl}\!\big(\!\CharFun{A}\!\big)(s)$ reduces to $1 \cdot \CharFun{A}(s')$, which gives $1$ if $s' \models A$ and $0$ otherwise. 
This yields the classical weakest pre--condition semantics of ordinary sequential programs.

%\todo[inline]{$\downarrow\downarrow\downarrow$ This needs to go away. $\downarrow\downarrow\downarrow$}
%\autoref{eq:wp-def} defines the transformer $\wp{\:\cdot\:}$ only for the fully probabilistic fragment of \pRGCL. If $\prog{c}{\decl}$ contains
%non-deterministic choices, it admits multiple distributions of final states. To
%give semantics to full \pRGCL we take a demonic view of non--determinism,
%where we assume an adversary trying to minimize the expected value of
%post--conditions. Therefore, for arbitrary (in particular, non-deterministic)
%programs in $\pRGCL$ we extend \autoref{eq:wp-def} by
%%
%\[ 
%\wpd{c}{\decl}(f) \:\eqdef\: \lambda s \mydot \inf \: \bigl\{\EV{\mu'}{f}
%\mid \mu' \in \sem{c,\decl}(s) \bigr\}~.
%\]
%%
%\todo[inline]{$\uparrow\uparrow\uparrow$ This needs to go away. $\uparrow\uparrow\uparrow$ ---Benni}
To reason about partial program correctness, \pRGCL also admits a liberal version of the transformer $\wp{\:\cdot\:}$, namely $\wlp{\:\cdot\:}$. 
In the same vein as for ordinary sequential programs, $\wpd{c}{\decl}\!(\CharFun{A})(s)$ gives the probability that program $\prog{c}{\decl}$ terminates and establishes event $A$ from state $s$, while $\wlpd{c}{\decl}\!(\CharFun{A})(s)$ gives the probability that $\prog{c}{\decl}$ terminates and establishes $A$, or diverges.

Formally, the transformer \wpsymbol operates on unbounded, so--called
\emph{expectations} in $\UEX \eqdef \left\{f ~\middle|~ f\colon\State \To [0,\, \infty]\right\}$,
while the transformer \wlpsymbol operates on bounded expectations in
$\BEX \eqdef \{f \mid f\colon\State \To [0,1] \}$.  Our expectation transformers
have thus type $\wp{\:\cdot\:} \colon \UEX \To \UEX $ and
$\wlp{\:\cdot\:} \colon \BEX \To \BEX$.\footnote{The transformer \textsf{wlp} is well--typed
  because $\wlpd{c}{\decl}\!(f)(s) \leq \sup_{s'} f(s')$ for every state
  $s$.} In the probabilistic setting pre-- and post--conditions are thus
referred to as \emph{pre--} and \emph{post--expectations}.

  \paragraph{Notation.} We use boldface for constant
 expectations, \eg $\CteFun{1}$ denotes the constant expectation
 $\lambda s \mydot 1$. Given an arithmetical expression $E$ over program
 variables we write $E$ for the expectation that in states $s$ returns
 $\sem{E}(s)$. Given a Boolean expression $G$ over program variables let
 $\ToExp{G}$ denote the $\{0,1\}$--valued expectation that on state $s$
 returns $1$ if $\sem{G}(s)=\true$ and $0$ if $\sem{G}(s)=\false$. Finally, given
 variable $x$, expression $E$ and expectation $f$ we use $f\subst{x}{E}$ to
 denote the expectation that on state $s$ returns
 $f(s\subst{x}{\sem{E}(s)})$. Moreover, ``$\preceq$'' denotes the
 pointwise order between expectations, \ie $f_1 \preceq f_2$ iff
 $f_1(s) \leq f_2(s)$ for all states $s \in \State$.

\subsection{Inductive Characterization}
\label{sec:wp-rules}
\citet{McIver:2001b} showed that the expectation transformers \textsf{wp} and \textsf{wlp} can be
defined by induction on the program's structure. We now recall their result,
taking an alternative approach to handle recursion: While
\citeauthor{McIver:2001b} use fixed point techniques, we
follow~\eg~\citet{Hehner:AI:79} and define the semantics of a recursive
procedure as the limit of an approximation sequence.  We believe that this
approach is sometimes more intuitive and closer to the operational view of
programs.
%and is also \todo{We actually need \emph{both} the inlining \emph{and} the fixed point characterization of recursion to prove our correspondence theorem, so this sentence should be adapted at some point. ---Benni}{closer to the operational view} of recursion discussed in \autoref{sec:operational}.

In the same way as the semantics of loops is defined as the limit of their finite
unrollings, we define the semantics of recursive procedures as the limit of
their finite inlinings. 
%
% To define the semantics of a
% call to the procedure we consider the sequence of programs
% %
% \begin{align*}
% %
% c_0 &\boldsymbol{\colon}\;\; \textbf{\textsf{abort}} \\[3pt]
% %
% c_1 &\boldsymbol{\colon}\;\; \begin{aligned}[t]
% &\If~(x \leq 0)~\{\Ass{y}{1} \}~\Else\\[-1pt]
% &\quad \bigl\{ \PChoice{\Ass{\Delta}{1}}{0.9}{\Ass{\Delta}{2}};\,
% \Ass{x}{x{-}\Delta};\\[-2pt]
% &\qquad \textbf{\textsf{abort}};\, \Ass{x}{x{+}\Delta};\, \Ass{y}{y \cdot x} \bigr\}
% \end{aligned}\\[-10pt]
% %
% &\,\vdots \\[-2pt]
% %
% c_{n+1} &\boldsymbol{\colon}\;\; \begin{aligned}[t]
% &\If~(x \leq 0)~\{\Ass{y}{1} \}~\Else\\[-1pt]
% &\quad \bigl\{ \PChoice{\Ass{\Delta}{1}}{0.9}{\Ass{\Delta}{2}};\,
% \Ass{x}{x{-}\Delta};\\[-2pt]
% &\qquad \mathbf{c_n}; \, \Ass{x}{x{+}\Delta};\, \Ass{y}{y \cdot x} \bigr\}
% \end{aligned} 
% \end{align*}
% %
% and take the limit of the sequence $\bigl\langle \wllp{c_n} \bigr\rangle_n$. As
% a sanity check that this limit gives the expected result, note that $c_{n+1}$
% faithfully represents the behavior of $\Call{\PName}$ for all inputs with $x
% \leq n$.
%
Formally, the \emph{$n$-th inlining} $\Calln{\PName}{n}{\decl}$ of
procedure \PName \wrt declaration $\decl$ is defined inductively by
\begin{align*}
\Calln{\PName}{0}{\decl} &\:=\:  \Abort\\
\Calln{\PName}{n+1}{\decl} &\:=\: \decl(P)
\subst{\Call{\PName}}{\Calln{\PName}{n}{\decl}}~,
\end{align*}
where $c\subst{\Call{\PName}}{c'}$ denotes the syntactic replacement of
every occurrence of $\Call{\PName}$ in $c$ by $c'$.\footnote{The formal definition of
  this syntactic replacement proceeds by a routine induction on the structure of
  $c$; see \autoref{fig:command-subst} in \autoref{sec:subst} for details.} The
family of commands $\Calln{\PName}{n}{\decl}$ define a sequence of
approximations to $\Call{\PName}$ where $\Calln{\PName}{0}{\decl}$ is the ``poorest''
approximation, while the larger the $n$, the more precise the approximation
becomes. Observe that, in general, $\Calln{\PName}{n+1}{\decl}$ mimics the exact
behavior of $\Call{\PName}$ for all executions that finish after at most $n$
recursive calls.

The expectation transformer semantics over $\pRGCL$ is provided in
\autoref{fig:wp-sem}. The action of transformers on procedure calls is defined
as the limit of their action over the $n$-th inlining of the procedures. For the
rest of the language constructs, we follow \citet{McIver:2001b}. Let us briefly
explain each of the rules. $\wpd{\Skip}{\decl}$ behaves as the
identity since $\Skip$ has no effect. The pre--expectation of an assignment is
obtained by updating the program state and then applying the post--expectation,
\ie $\wpd{\Ass{x}{E}}{\decl}$ takes post--expectation $f$ to pre--expectation
$f\subst{x}{E}=\lambda s\mydot f(s\subst{x}{\sem{E}(s)})$. $\wpd{\Abort}{\decl}$
maps any post--expectation to the constant pre--expectation
$\CteFun{0}$. Observe that expectation $\CteFun{0}$ is the probabilistic
counterpart of predicate $\false$.  $\wpd{\Cond{G}{c_1}{c_2}}{\decl}$ behaves
either as $\wpd{c_1}{\decl}$ or $\wpd{c_2}{\decl}$ according to the evaluation
of $G$. $\wpd{\PChoice{c_1}{p}{c_2}}{\decl}$ is obtained as a convex combination
of $\wpd{c_1}{\decl}$ and $\wpd{c_2}{\decl}$, weighted according to
$p$. $\wpd{\Call{\PName}}{\decl}$ behaves as the limit of $\wpsymbol$ on the
sequence of finite truncations (or inlinings) of $\PName$. We take the supremum
because the sequence is increasing. Observe that we advertently include no
declaration in $\wp{\Calln{\PName}{n}{\decl}}\!(f)$ because
$\Calln{\PName}{n}{\decl}$ is a closed command for every $n$. Finally,
$\wpd{c_1;c_2}{\decl}$ is obtained as the functional composition of
$\wpd{c_1}{\decl}$ and $\wpd{c_2}{\decl}$. The $\wlpsymbol$ transformer follows the
same rules as $\wpsymbol$, except for the \Abort statement and procedure
calls. $\wlpd{\Abort}{\decl}$ takes any post--expectation to pre--expectation
$\CteFun{1}$. (Expectation $\CteFun{1}$ is the probabilistic counterpart of 
predicate $\true$.)  $\wlpd{\Call{\PName}}{\decl}$ also behaves as the limit of
$\wlpsymbol$ on the sequence of finite truncations of $\PName$. This time we
take the infimum because the sequence is decreasing.

\begin{figure}[t]
\scalebox{0.95}{
$
\begin{array}{ll}
\specialrule{0.8pt}{0pt}{2pt}
\boldsymbol{c} & \boldsymbol{\wpd{c}{\decl}\!(f)}\\
\specialrule{0.8pt}{2pt}{2pt}
%
% No operation
\Skip   & f \\[1.5pt]
% Assignment
\Ass{x}{E}  & f\!\subst{x}{E} \\[1.5pt]
% Abortion
\Abort & \CteFun{0} \\[1.5pt]
% Guarded command 
\Cond{G}{c_1}{c_2}  & 
    \ToExp{G} \cdot \wpd{c_1}{\decl}\!(f) + \ToExp{\lnot G} \cdot \wpd{c_2}{\decl}
    \!(f) \\[3pt]
% Probabilistic choice
\PChoice{c_1}{p}{c_2} &   p \cdot \wpd{c_1}{\decl}\!(f)  + (1{-}p) \cdot
\wpd{c_2}{\decl}\!(f) \\[3pt]
% Procedure call
\Call{\PName} &  \sup_{n} \wp{\Calln{\PName}{n}{\decl}}\!(f)  \\[3pt]
% Sequential composition 
c_1;c_2 & \wpd{c_1}{\decl} \bigl(\wpd{c_2}{\decl}\!(f)\bigr)
\\[5pt]
\specialrule{0.8pt}{0pt}{2pt}
\boldsymbol{c} & \boldsymbol{\wlpd{c}{\decl}\!(f)}\\
\specialrule{0.8pt}{2pt}{2pt}
\Abort  & \CteFun{1} \\[1.5pt]
\Call{\PName} &  \inf_{n} \wlp{\Calln{\PName}{n}{\decl}}\!(f)
\end{array}
$}
\caption{Expectation transformer semantics of \pRGCL programs.  The
  $\wlp{\:\cdot\:}$ transformer follows the same rules as $\wp{\:\cdot\:}$, expect for
  $\Abort$ and procedure calls. Sum, product, supremum and infimum
  over expectations are all defined pointwise.}
\label{fig:wp-sem}
\end{figure}

\begin{example}
Reconsider $c_{\mathsf{coins}} = c_1;\,c_2$ from
\autoref{sec:wp-def} with
\[
c_1  \boldsymbol{\colon}\; 
\PChoice{\Ass{x}{0}}{\nicefrac{1}{2}}{\Ass{x}{1}} 
\text{\ \ and \ \ } 
c_2  \boldsymbol{\colon}\;
\PChoice{\Ass{y}{0}}{\nicefrac{1}{3}}{\Ass{y}{1}}\;.
\]
We use our weakest pre--expectation calculus to formally determine the
probability that the outcome of the two coins coincide:
\begin{align*}
\MoveEqLeft[1]
\wp{c_{\mathsf{coins}}}\!(\CharFun{x{=}y}) \\
& =~ \wp{c_1} \! \bigl( \wp{c_2}\!(\CharFun{x{=}y})  \bigr)\\
& =~ \wp{c_1} \! \bigl( \tfrac{1}{3} \cdot \wp{\Ass{y}{0}}\!(\CharFun{x{=}y}) +
  \tfrac{2}{3} \cdot \wp{\Ass{y}{1}}\!(\CharFun{x{=}y})  \bigr) \\
& =~ \wp{c_1} \! \bigl( \tfrac{1}{3} \cdot \CharFun{x{=}0} +
  \tfrac{2}{3} \cdot \CharFun{x{=}1}  \bigr) \\
& =~ \tfrac{1}{2} \cdot \wp{\Ass{x}{0}} \! \bigl( \tfrac{1}{3} \cdot \CharFun{x{=}0} +
  \tfrac{2}{3} \cdot \CharFun{x{=}1}  \bigr) \\
&\qquad+ \tfrac{1}{2} \cdot \wp{\Ass{x}{1}} \! \bigl( \tfrac{1}{3} \cdot \CharFun{x{=}0} +
  \tfrac{2}{3} \cdot \CharFun{x{=}1}  \bigr)  \\
& =~ \tfrac{1}{2} \cdot \bigl( \tfrac{1}{3} \cdot \CharFun{0{=}0} +
  \tfrac{2}{3} \cdot \CharFun{0{=}1}  \bigr) + \tfrac{1}{2} \cdot \bigl( \tfrac{1}{3} \cdot \CharFun{1{=}0} +
  \tfrac{2}{3} \cdot \CharFun{1{=}1}  \bigr)  \\
& =~ \tfrac{1}{2} \cdot \CteFun{\tfrac{1}{3}} + \tfrac{1}{2} \cdot
  \CteFun{\tfrac{2}{3}} ~=~ \CteFun{\tfrac{1}{2}} \tag*{$\triangle$}
\end{align*} 
\end{example}

The transformers $\wpsymbol$ and $\wlpsymbol$ enjoy several appealing algebraic properties, which
we summarize below.
\begin{lemma}[Basic properties of $\wllpsymbol$]
\label{thm:wp-basic-prop}
For every program $\prog{c}{\decl}$, every $f_1, f_2$, and increasing $\omega$--chain
$\chain{f}{\preceq}$ in $\UEX$, $g_1,g_2$, and every decreasing $\omega$--chain
$\chain{g}{\succeq}$ in $\BEX$, and scalars
$\alpha_1,\alpha_2 \in \Reals_{\geq 0}$ it holds:
\begin{flushleft}
\begin{tabular}{@{}l@{\hspace{1em}}l}
   \textnormal{Continuity:} & $\sup\nolimits_n \wpd{c}{\decl}\!(f_n) \:=\:
                              \wpd{c}{\decl} \! (\sup\nolimits_n  f_n) $\\[0.5ex]
& $\inf\nolimits_n \wlpd{c}{\decl}\! (g_n) \:=\: \wlpd{c}{\decl}\!
  (\inf\nolimits_n g_n)$ \\[1ex]
	\textnormal{Monotonicity:}			& $f_1 \preceq f_2 \implies \wpd{c}{\decl}\!(f_1) \preceq \wpd{c}{\decl}\!(f_2)$\\[0.5ex]
& $g_1 \preceq g_2 \implies \wlpd{c}{\decl}\!(g_1) \preceq \wlpd{c}{\decl}\!(g_2)$\\[1ex]

	\textnormal{Linearity:} 			& $\wpd{c}{\decl}\!(\alpha_1
                                                  \cdot f_1 + \alpha_2 \cdot
                                                  f_2)$\\[0.5ex]
  &$\:=\: \alpha_1 \cdot
\wpd{c}{\decl}\!(f_1) + \alpha_2 \cdot \wpd{c}{\decl}\!(f_2)$\\[1ex]
	\textnormal{Preserv.~of $\CteFun{0}$,$\CteFun{1}$:} 	& $\wpd{c}{\decl}\!(\CteFun{0}) =
                                          \CteFun{0}\ $ and $\ \wlpd{c}{\decl}\!(\CteFun{1}) =
                                          \CteFun{1}$
\end{tabular}
\end{flushleft}
\end{lemma}
\begin{proof}
See Appendix~\ref{app:basicproperties}.
\end{proof}

\paragraph{Program termination.} Since the termination behavior of a program is
given by the probability that it establishes $\true$, we can readily use
the transformer \wpsymbol to reason about program termination. It
suffices to consider the weakest pre--expectation of the program \wrt post--expectation
$\CharFun{\true} = \CteFun{1}$. Said otherwise,
$\wpd{c}{\decl}\!(\CteFun{1})(s)$ gives the termination probability of program
$\prog{c}{\decl}$ from state $s$. In particular, if the program terminates with
probability $1$, we say that it \emph{terminates almost--surely}.

\subsection{Characterization based on Fixed Points}
Next we use a continuity argument on the transformer \wllpsymbol to prove that its
action on recursive procedures can also be defined using fixed point techniques.
This alternative characterization rests on a subsidiary transformer
\ewllp{\:\cdot\:}{\theta}, which is a slight variant of \wllp{\:\cdot\:}. The main
difference between these transformers is the mechanism that they use to
give semantics to procedure calls: \wllp{\:\cdot\:} relies on a declaration $\decl$,
while \ewllp{\:\cdot\:}{\theta} relies on a so--called (\emph{liberal})
\emph{semantic environment} $\theta \colon \UEX \To \UEX$
($\theta \colon \BEX \To \BEX$) which is meant to directly encode the semantics
of procedure calls. Then $\ewllp{\Call{\PName}}{\theta}\!(f)$ gives $\theta(f)$, while
for all other program constructs $c$, $\ewllp{c}{\theta}\!(f)$ agrees with
$\wllp{c}\!(f)$; see \autoref{fig:ewp-sem} in \autoref{sec:app-fixed-point-sem}
for details. For technical reasons, in the remainder of our development we will
consider only continuous semantic environments in
$\SEnv \eqdef  \{f \mid f\colon\UEX \To \UEX \text{ is upper
  continuous} \}$ and $\LSEnv \eqdef  \{f \mid f\colon\BEX \To \BEX \text{ is
  lower continuous} \}$.%
\footnote{A (liberal) semantic environment $\theta$ is \emph{upper}
   (\emph{lower}) \emph{continuous} iff for every increasing
   $\omega$-chain $\chain{f}{\preceq}$ (decreasing $\omega$-chain
   $\chain{f}{\succeq}$), $\sup_n \theta(f_n) = \theta (\sup_n f_n)$ ($\inf_n \theta(f_n) = \theta (\inf_n f_n)$).}
 This is a natural assumption since we are interested only in semantic
environments that are obtained as the \wllpsymbol--semantics of a \pRGCL
program, which are continuous by \autoref{thm:wp-basic-prop}.

% \footnote{Observe that here we deliberately omitted declaration $\decl$ in
%  $\wllpd{c}{\decl}(f)$ because declarations affect the action of \wllpsymbol
%  only on procedure calls.}
% In the remainder we let $\SEnv \eqdef \UEX \To \UEX$ ($\LSEnv \eqdef \BEX \To
% \BEX$) be the set of (liberal) semantic environments.

The semantics of recursive procedures can now be readily given as the fixed point
of a semantic environment transformer.
\begin{theorem}[Fixed point characterization for procedure calls]
\label{thm:fp-rec}
Given a declaration $\decl \colon \{\PName\} \To \Cmd$ for procedure $\PName$,  
\begin{samepage}
\begin{align*}
\wpd{\Call{\PName}}{\decl} \:&=\: \lfp{\sqsubseteq}{\lambda \theta\!:\!\SEnv \mydot
\ewp{\decl(\PName)}{\theta}}\\
\wlpd{\Call{\PName}}{\decl} \:&=\: \gfp{\sqsubseteq}{\lambda \theta\!:\!\LSEnv \mydot
\ewlp{\decl(\PName)}{\theta}}~.
\end{align*}
\end{samepage}
\end{theorem}
\begin{proof}
  See Appendix~\ref{sec:app-fixed-point-sem}.
\end{proof}
\noindent The fixed points above are taken \wrt the pointwise
order ``$\sqsubseteq$" over semantic environments: given
$\theta_1, \theta_2 \in \SEnv$ (resp.\ $\theta_1, \theta_2 \in \LSEnv$),
$\theta_1 \sqsubseteq \theta_2$ iff $\theta_1(f) \preceq \theta_2(f)$ for all
$f \in \UEX$ (resp.\ $f \in \BEX$).

%\begin{proof}
%  Consider the case of transformer $\wpsymbol$ and let
%  $F=\lambda \theta\!:\!\SEnv \mydot \ewp{\decl(\PName)}{\theta}$. The crux of
%  the proof consists in using a continuity argument to show that
%  $\lfp{\sqsubseteq}{F}$ can be obtained by fixed point iteration from the
%  semantic environment
%  \mbox{$\bot_\SEnv = \lambda f\!:\!\UEX \mydot \CteFun{0}$}, that is,
%  $\lfp{\sqsubseteq}{F} = \sup\nolimits_n F^n(\bot_\SEnv)$, where $F^n$ denotes
%  composition of $F$ with itself $n$ times. The proof then concludes by showing
%  that $F^n(\bot_\SEnv) = \wp{\Calln{\PName}{n}{\decl}}$ for all $n$.  For
%  details see Appendix~\ref{sec:app-fixed-point-sem}. The case of transformer
%  $\wlpsymbol$ is analogous.
%\end{proof}
%
%
\autoref{thm:fp-rec} reveals an inherent difference between the complexities of
reasoning about loops and general recursion: 
The semantics of loops can be given as the fixed point of an expectation transformer (see \eg \cite{Morgan:RW:1996}), while the
semantics of recursion requires the fixed point of a (\emph{higher
  order}) environment transformer. 
% As a consequence, proving our results for recursive programs requires new
% insights not present in the proofs for $\While$--programs.
This fact was already noticed by \citet[p.~xvii]{Dijkstra} and
later on confirmed by \citet[p.~517]{Nelson:TOPLAS:89:} for non--probabilisitic programs.

%\todo{\tiny I feel that during this entire subsection we are switching
%   quite heavily back and forth from $\wp{c}$ to $\wpd{c}{\mathcal
%     D}(f)$---Benni \\ @Fede:  it is clarified twice that when the program is closed, we
%   drop the declaration. Otherwise we should half of the times add the phrase ``For
%   any---irrelevant---declaration $\mathcal D \ldots$''.}{}

\section{Correctness of Recursive Programs}
\label{sec:proof-rules}
In this section we introduce some proof rules for effectively reasoning about
the behavior of recursive programs. 
For that we require the notion of \emph{constructive
  derivability}. Given logical formulae $A$ and $B$, we use $\deriv{A}{B}$ to
denote that $B$ can be derived assuming $A$. In particular, we will consider
claims of the form
\[
\deriv{\wllp{\Call{\PName}}\!(f_1) \bowtie g_1 }{\wllp{c}\!(f_2) \bowtie g_2}~,
\]
where $\bowtie \: \in \! \{\preceq, \succeq \}$, $f_1,g_1$ give the
specification of $\Call{\PName}$ and $f_2,g_2$ the specification of $c$. 
Notice that in
such a claim we omit any procedure declaration as the derivation is
independent of $P$'s body. %
%
% In effect, from the above derivability claim it
% follows that for all environment $\decl^\star$,
% %
% \begin{equation}\label{eq:deriv-elim}
% \wllpd{\Call{\PName}}{\decl^\star}\!(f_1) \preceq g_1   \implies  
% \wllpd{c}{\decl^\star}\!(f_2) \preceq g_2
% \end{equation}
% which is used for proving our rules sound.

Our first two rules are extensions of well--known rules for ordinary recursive programs (see \eg
\cite{Hesselink:FAC:93}) to a probabilistic
setting:
\[
\small 
\begin{array}{l}
\infrule{\deriv{~\wp{\Call{\PName}}\!(f) \preceq
    g}{\wp{\decl(\PName)}\!(f) \preceq g~}}{\wpd{\Call{\PName}}{\decl}\!(f) \preceq g}%
\,\lrule{wp-rec}\\[3.5ex]
\infrule{\deriv{~g \preceq \wlp{\Call{\PName}}\!(f)}{g \preceq \wlp{\decl(\PName)}\!(f)~}}{g \preceq \wlpd{\Call{\PName}}{\decl}\!(f)}%
\,\lrule{wlp-rec}
\end{array}
\]
So for proving that a procedure call satisfies a
specification (given by $f,g$), it suffices to show that the procedure's body satisfies the specification,
assuming that the recursive calls in the body do, too.

\begin{example}\label{ex:puzz-3-upp-bound}
 Reconsider the procedure $\PName_{\mathsf{rec_3}}$ with declaration
  \[
\decl(\PName_{\mathsf{rec_3}})\, \boldsymbol{\colon}\; \;
\PChoice{\Skip}{\nicefrac{1}{2}} {\Call{\PName_{\mathsf{rec_3}}};~ \Call{\PName_{\mathsf{rec_3}}};~ \Call{\PName_{\mathsf{rec_3}}}} 
\]
presented in the introduction. We prove that it terminates with probability \emph{at
most} $\gr = \tfrac{\sqrt{5}-1}{2}$ from any initial state. Formally, this is
captured by $\wpd{\Call{\PName}}{\decl}\!(\CteFun{1}) \:\preceq\: \CteFun{\gr}$.
To prove this, we apply rule \lrule{wp-rec}. We must then establish the
 derivability claim
\[
\deriv%
{\wp{\Call{\PName}}\!(\CteFun{1}) \preceq \CteFun{\gr}}
{\wp{\decl(\PName_{\mathsf{rec_3}})}\!(\CteFun{1} ) \preceq \CteFun{\gr}}~.
\]
%
% (Recall that in derivability claims, declarations are dropped
% out because they play no roll.) 
The derivation goes as follows:
%
%\medskip
%\begin{xtabular*}{\textwidth}{>{$}c<{$}@{\;\;\;}  >{$}l<{$}@{}}
\belowdisplayskip=-1\baselineskip
\begin{align*}
\begin{array}{c@{\;\;\;} l@{}}
&\wp{\decl(\PName_{\mathsf{rec_3}})}\!(\CteFun{1}) \\[2pt]
= & \qquad \by{def.~of \wpsymbol}\\[2pt]
&\tfrac{1}{2} \cdot \wp{\Skip}\!(\CteFun{1}) + \tfrac{1}{2} \cdot
\wp{\Call{\PName_{\mathsf{rec_3}}};~ \Call{\PName_{\mathsf{rec_3}}};~ \Call{\PName_{\mathsf{rec_3}}}}\!(\CteFun{1}) \\[2pt]
= & \qquad  \by{def.~of \wpsymbol}\\[2pt]
&\CteFun{\tfrac{1}{2}} + \tfrac{1}{2} \cdot
\wp{\Call{\PName_{\mathsf{rec_3}}};~
  \Call{\PName_{\mathsf{rec_3}}}}\!\bigl(\wp{\Call{\PName_{\mathsf{rec_3}}}}\!(\CteFun{1})\bigr)\\[2pt]
\preceq & \qquad \by{assumption, monot.~of \wpsymbol}\\[2pt]
&\CteFun{\tfrac{1}{2}} + \tfrac{1}{2} \cdot
\wp{\Call{\PName_{\mathsf{rec_3}}};
  \Call{\PName_{\mathsf{rec_3}}}}\!(\CteFun{\gr})\\[2pt]
= & \qquad \by{def.~of \wpsymbol, scalab.~of \wpsymbol twice}\\[2pt]
&\CteFun{\tfrac{1}{2}} + \tfrac{1}{2} \, \gr \cdot
\wp{\Call{\PName_{\mathsf{rec_3}}}}\!\bigl(\wp{\Call{\PName_{\mathsf{rec_3}}}}\!(\CteFun{1})\bigr)\\[2pt]
\preceq & \qquad \by{assumption, monot.~of \wpsymbol}\\[2pt]
&\CteFun{\tfrac{1}{2}} + \tfrac{1}{2} \, \gr \cdot
\wp{\Call{\PName_{\mathsf{rec_3}}}}\!(\CteFun{\gr})\\[2pt]
= & \qquad \by{scalab.~of \wpsymbol}\\[2pt]
&\CteFun{\tfrac{1}{2}} + \tfrac{1}{2} \, \gr^2 \cdot
\wp{\Call{\PName_{\mathsf{rec_3}}}}\!(\CteFun{1})\\[2pt]
\preceq & \qquad \by{assumption, monot.~of \wpsymbol}\\[2pt]
&\CteFun{\tfrac{1}{2}} + \CteFun{\tfrac{1}{2}} \, \gr^3\\[2pt]
= & \qquad \by{algebra}\\[2pt] 
&\CteFun{\gr} 
\end{array} \\[-1\normalbaselineskip]\tag*{$\triangle$}
\end{align*}\normalsize
%\end{xtabular*} 
%\medskip
%
\end{example}
An appealing feature of our approximation semantics is that to prove the following
soundness result we do not need to resort to a continuity argument on the expectation
transformers.

\begin{theorem}[Soundness of rules \lrule{\wllpsymbol-rec}]
Rules \lrule{wp-rec} and \lrule{wlp-rec} 
are sound \wrt the \wllpsymbol semantics in \autoref{fig:wp-sem}.
\end{theorem}
\begin{proof}
See
Appendix~\ref{sec:app-om-rule-sound}. 
\end{proof}

% As a side remark, observe that the soundness proof of the rules requires no
% continuity argument of $\wllpsymbol$. This is another benefit of adopting an
% approximating---rather than a fixed point---semantics.

Rules \lrule{w(l)p-rec} allow deriving only one--sided bounds for the
weakest (liberal) pre--expectation of a procedure call. It is also possible to
derive two--sided bounds by means of the following rules:
\begin{equation*}
\resizebox{1\hsize}{!}{$
\!\!\begin{array}{l}
\infrule{l_0=\CteFun{0}, \qquad \qquad \qquad u_0=\CteFun{0}, \\[0.1ex]
~l_n \preceq \wp{\Call{\PName}}\!(f)  \preceq u_n%
      \derivsymbol l_{n\!+\!1} \preceq \wp{\decl(\PName)}\!(f) \preceq u_{n\!+\!1}\:} 
{\sup_n l_n \preceq \wpd{\Call{\PName}}{\decl}\!(f) \preceq \sup_n u_n }%
 \lrule{wp-rec$_\omega$}\\[3.5ex]
\infrule{ l_0=\CteFun{1}, \qquad \qquad \qquad u_0=\CteFun{1}, \\[0.1ex]
~l_n \preceq \wlp{\Call{\PName}}\!(f)  \preceq u_n%
      \derivsymbol l_{n\!+\!1} \preceq \wlp{\decl(\PName)}\!(f) \preceq u_{n\!+\!1}\:} 
{\inf_n l_n \preceq \wlpd{\Call{\PName}}{\decl}\!(f) \preceq \inf_n u_n }%
 \lrule{wlp-rec$_\omega$}\\[\normalbaselineskip]
\end{array}\!\!$}
\end{equation*}

In constrast to rules \lrule{w(l)p-rec}, these rules require exhibiting two
sequences of expectations $\langle l_n \rangle$ and $\langle u_n \rangle$ rather
than a single expectation $g$ to bound the weakest (liberal) pre--expectation of
a procedure call. Intuitively $l_n$ ($u_n$) represents a lower (upper) bound for
the weakest pre--expectation of the $n$-inlining of the procedure, \ie from the
premises of the rules we will have
$l_n \preceq \wllp{\Calln{\PName}{n}{\decl}}\!(f) \preceq u_n$ for all
$n \in \Nats$. 

Observe that both rules can be specialized to reason about one--sided bounds. For
instance, by setting $u_{n+1}=\CteFun{\infty}$ in \lrule{wp-rec$_\omega$} we can reason about lower bounds of
$\wpd{\Call{\PName}}{\decl}\!(f)$, which is not supported by rule
\lrule{wp-rec}. Similarly, by taking $l_n=\CteFun{0}$ in rule
\lrule{wlp-rec$_\omega$} we can reason about upper bounds of
$\wlpd{\Call{\PName}}{\decl}\!(f)$.

\begin{example}\label{ex:puzz-3-low-bound}
  Reconsider the procedure $\PName_{\mathsf{rec_3}}$ from
  \autoref{ex:puzz-3-upp-bound}. Now we prove that the procedure terminates with probability
  \emph{at least} $\gr = \tfrac{\sqrt{5}-1}{2}$ from any initial state. To this
  end, we rely on the fact that $\gr$ can be characterized by the asymptotic
  behavior of the sequence $\langle \gr_n \rangle$, where $\gr_0 = 0$ and
  $\gr_{n+1} = \tfrac{1}{2} + \tfrac{1}{2} \, \gr_n^3$. In symbols,
  $\gr = \sup_n \gr_n$. We wish then to prove that
  \[
  \sup\nolimits_n \, \CteFun{\gr_n} \:\preceq\: \wpd{\Call{\PName_{\mathsf{rec_3}}}}{\decl}\!(\CteFun{1})~.
  \]
  To establish this formula we apply the one side variant of rule
  \lrule{wp-rec$_\omega$} to reason about lower bounds of
  $\wpd{\Call{\PName_{\mathsf{rec_3}}}}{\decl}\!(\CteFun{1})$, that is, we
  implicitly take $u_{n+1}=\CteFun{\infty}$. We must then establish
\[
\deriv%
{\CteFun{\gr_{n}} \preceq \wp{\Call{\PName_{\mathsf{rec_3}}}}\!(\CteFun{1})}
{\CteFun{\gr_{n+1}} \preceq\wp{\decl(\PName_{\mathsf{rec_3}})}\!(\CteFun{1})}~.
\]
The derivation follows the same steps as those taken in
\autoref{ex:puzz-3-upp-bound} to give upper bounds on
$\wpd{\Call{\PName_{\mathsf{rec_3}}}}{\decl}\!(\CteFun{1})$. 
%   % 
%   \begin{align*}
%     \begin{array}{c@{\;\;\;} l@{} }
%       &\wp{\decl_{\mathsf{rec_3}(\PName)}}\!(\CteFun{1}) \\[2pt]
% %
% = & \qquad \by{def.~of \wpsymbol}\\[2pt]
% %
% &\tfrac{1}{2} \cdot \wp{\Skip}\!(\CteFun{1}) + \tfrac{1}{2} \cdot
% \wp{\Call{\PName}; \Call{\PName}; \Call{\PName}}\!(\CteFun{1}) \\[2pt]
% %
% = & \qquad  \by{def.~of \wpsymbol}\\[2pt]
% %
% &\CteFun{\tfrac{1}{2}} + \tfrac{1}{2} \cdot
% \wp{\Call{\PName};
%   \Call{\PName}}\!\bigl(\wp{\Call{\PName}}(\CteFun{1})\bigr)\\[2pt]
% %
% \succeq & \qquad \by{assumption, monot.~of \wpsymbol}\\[2pt]
% %
% &\CteFun{\tfrac{1}{2}} + \tfrac{1}{2} \cdot
% \wp{\Call{\PName};
%   \Call{\PName}}\!(\CteFun{\gr_n})\\[2pt]
% %
% = & \qquad \by{def.~of \wpsymbol, scalab.~of \wpsymbol twice}\\[2pt]
% %
% &\CteFun{\tfrac{1}{2}} + \tfrac{1}{2} \, \gr_n \cdot
% \wp{\Call{\PName}}\!\bigl(\wp{\Call{\PName}}(\CteFun{1})\bigr)\\[2pt]
% %
% \succeq & \qquad \by{assumption, monot.~of \wpsymbol}\\[2pt]
% %
% &\CteFun{\tfrac{1}{2}} + \tfrac{1}{2} \, \gr_n \cdot
% \wp{\Call{\PName}}\!(\CteFun{\gr_n})\\[2pt]
% %
% = & \qquad \by{scalab.~of \wpsymbol}\\[2pt]
% %
% &\CteFun{\tfrac{1}{2}} + \tfrac{1}{2} \, \gr_n^2 \cdot
% \wp{\Call{\PName}}\!(\CteFun{1})\\[2pt]
% %
%       \succeq & \qquad \by{assumption, monot.~of \wpsymbol}\\[2pt]
% %
% &\CteFun{\tfrac{1}{2}} + \CteFun{\tfrac{1}{2}} \, \gr_n^3\\[2pt]
% %
% = & \qquad \by{def.~of $\langle \gr_n \rangle$}\\[2pt] 
% %
% &\CteFun{\gr_{n+1}}
% \end{array} %\\[-\normalbaselineskip]\tag*{$\triangle$}
% \end{align*}
%
  Combining the result proved with that in \autoref{ex:puzz-3-upp-bound},
  we conclude that $\gr = \tfrac{\sqrt{5}-1}{2}$ is the exact termination
  probability of $\prog{\Call{\PName_{\mathsf{rec_3}}}}{\decl}$.\hfill
    $\triangle$
\end{example}

Lastly, we can establish the correctness our rules.
\begin{theorem}[Soundness of rules \lrule{\wllpsymbol-rec$_\omega$}]
  Rules \lrule{w(l)p-rec$_\omega$} are sound \wrt the
  \wllpsymbol semantics in \autoref{fig:wp-sem}.
\end{theorem}
\begin{proof}
  See Appendix \ref{sec:app-om-rule-sound}.
\end{proof}

To conclude the section we would like to point out that the rule
\lrule{wp-rec$_\omega$} is related to previous work on proof rules. 
It can be viewed as a
generalization of \citeauthor{Jones:1992}'s loop rule~\cite{Jones:1992} to the
case of recursion (even though \citeauthor{Jones:1992} originally presented a
one--sided version) and as an adaptation of \citeauthor{Audebaud:2009}'s
rule~\cite{Audebaud:2009} to our weakest pre--expectation semantics. 
The counterpart of the rule for partial correctness, on the
other hand, is, to the best of our knowledge, novel.

\section{The Expected Runtime of Programs}
\label{sec:eet}
% In the previous sections we have developed tools for reasoning about functional
% properties of probabilistic programs. To complete our study of probabilistic
% programs we now develop tools for reasoning about more operational
% aspects. Concretely we present a calculus for reasoning about the expected or
% average runtime of \pRGCL programs. This calculus builds upon our previous work
% in \cite{Kaminski:ETAPS:2016} and extends it incorporate recursion.
To further our study of recursive probabilistic programs we now develop a calculus for
reasoning about the expected or average runtime of \pRGCL programs. This
calculus builds upon our previous work in \cite{Kaminski:ETAPS:2016} and is able to handle recursive procedures.

\subsection{The Expected Runtime Transformer \boldeetsymbol}

We assume a runtime model where executing a $\Skip$ statement, an
assignment, evaluating the guard in a conditional branching and invoking a
procedure\footnote{Loosely speaking, the overall runtime of a procedure call is
  then one plus the runtime of executing the procedure's body.} consumes one unit of
time. On the other hand, combining two programs by means of a sequential
composition or a probabilistic choice consumes no additional time other than
that consumed by the original programs. Likewise, halting a program
execution with an $\Abort$ statement consumes no unit of time.

Since the runtime of a program varies according to the initial state from which
it is executed, our aim is to associate to each program $\prog{c}{\decl}$ a
mapping that takes each state $s$ to the expected time until
 $\prog{c}{\decl}$ terminates on $s$. Such mappings will range
over the set of \emph{runtimes}
$\Runtimes \eqdef \big\{\rt ~\big|~ \rt\colon\State \To [0,\, \infty]\big\}$.\footnote{Strictly speaking, the set of runtimes $\Runtimes$ coincides with the set
of unbounded expectations $\UEX$ but we prefer to distinguish the two sets
since they are to represent different objects. 
We will, however, keep the same
notations for runtimes as for expectations, for example $\rt\subst{x}{E}$, $\rt_1
\preceq \rt_2$, etc.}

To associate each program to its runtime we use a continuation passing style
formalized by the transformer 
\[
\eet{\:\cdot\:}\!\colon \Runtimes \To \Runtimes~.
\]
If $\rt \in \Runtimes$ represents the runtime of the computation that follows
program $\prog{c}{\decl}$, then $\eetd{c}{\decl}\!(\rt)$ represents the overall runtime
of $\prog{c}{\decl}$, plus the computation following $\prog{c}{\decl}$. Runtime
$\rt$ is usually referred to as the \emph{continuation} of $\prog{c}{\decl}$. In
particular, by setting the continuation of a program to zero we recover
the runtime of the plain program. That is, for every initial state $s$,
\[
\eetd{c}{\decl}\!(\ctert{0})(s)
\]
gives the expected runtime of program $\prog{c}{\decl}$ from state $s$.
% or, on account of our runtime model, the expected number of $\Skip$,
% assignments, guard evaluations and procedure calls in the execution of
% $\prog{c}{\decl}$.
% As for expectations, we use bold $\ctert{k}$ to denote the constant runtime
% $\lambda s \mydot k$ for $k \in [0, \infty]$.

% We use a continuation passing style to capture the runtime of programs because
% in this way we can unify under a single approach---that of runtime or
% expectation transformer---the reasoning about both the runtime and functional
% properties of programs. Now we shall see that the definition of \eetsymbol\
% transformer shares many similarities with that of \wpsymbol transformer. In
% effect, the theory underlying both calculus is the same. This will be confirmed
% by the fact that both calculus admit the same proof rules for reasoning about
% recursive procedures.

\begin{figure}[t]
$\small 
\begin{array}{@{~}ll@{~}}
\specialrule{0.8pt}{0pt}{2pt}
\boldsymbol{c} & \boldeetsymbol\boldsymbol{[c,\decl](\rt)}\\
\specialrule{0.8pt}{2pt}{6pt}
%
% No operation
\Skip   & \ctert{1} +  \rt \\[2.5pt]
% Assignment
\Ass{x}{E}  & \ctert{1} +  \rt\subst{x}{E} \\[2.5pt]
% Abortion
\Abort & \CteFun{0} \\[3.5pt]
% Guarded command 
\Cond{G}{c_1}{c_2}  & \ctert{1} + 
    \ToExp{G} \cdot \eetd{c_1}{\decl}\! (\rt) + \ToExp{\lnot G} \cdot \eetd{c_2}{\decl}\!(\rt) \\[3.5pt]
% Probabilistic choice
\PChoice{c_1}{p}{c_2} &   p \cdot \eetd{c_1}{\decl}\!(\rt)  + (1{-}p) \cdot
\eetd{c_2}{\decl}\!(\rt) \\[3pt]
% Procedure call
 \Call{\PName} &  \lfpsymbol_\sqsubseteq \left( \lambda \eta\!:\!\RtEnv \mydot
\ctertenv{1} \oplus \eeet{\decl(\PName)}{\eta} \right) \!(\rt) \\[4.5pt]
% Sequential composition 
c_1;c_2 & \eetd{c_1}{\decl}\bigl(\eetd{c_2}{\decl} \!(\rt)\bigr)
\end{array}
$
\caption{Rules for the expected runtime transformer
  \eetsymbol. $\lfp{\sqsubseteq}{F}$ denotes the least fixed point of
  transformer $F \colon \RtEnv \To \RtEnv$ \wrt the pointwise order
  ``$\sqsubseteq$" between runtime environments%: $\eta_1 \sqsubseteq \eta_2$ iff
  %$\eta_1(\rt) \preceq \eta_2(\rt)$ for every $\rt \in \Runtimes$
.}
\label{fig:eet}
\end{figure}

The transformer $\eetd{c}{\decl}$ is defined by induction on the structure of $c$, following
the rules in \autoref{fig:eet}. The rules are defined so as to correspond to the
aforementioned runtime model. That is, $\eetd{c}{\decl}\!(\ctert{0})$ captures
the expected number of assignments, guard evaluations, procedure calls and
$\Skip$ statements in the execution of $\prog{c}{\decl}$. Most rules are
self--explanatory. $\eetd{\Skip}{\decl}$ adds one unit of time to the
continuation since $\Skip$ does not modify the program state and its execution
takes one unit of time. $\eetd{\Ass{x}{E}}{\decl}$ also adds one unit of time,
but to the continuation evaluated in the state resulting from the
assignment. $\eetd{\Abort}{\decl}$ yields always the constant runtime
$\ctert{0}$ since $\Abort$ aborts any subsequent program execution (making their
runtime irrelevant) and consumes no time. $\eetd{\Cond{G}{c_1}{c_2}}{\decl}$
adds one unit of time to the runtime of either of its branches, depending on the
value of the guard. $\eetd{\PChoice{c_1}{p}{c_2}}{\decl}$ gives the weighted
average between the runtime of its branches, each of them weighted according to
its probability. $\eetd{c_1; c_2}{\decl}$ first applies $\eetd{c_2}{\decl}$ to
the continuation and then $\eetd{c_1}{\decl}$ to the resulting runtime of this
application. Finally, $\eetd{\Call{\PName}}{\decl}$ is defined using fixed point
techniques.

To understand the intuition behind the definition of
$\eetd{\Call{\PName}}{\decl}$ recall that $\Call{\PName}$ consumes one unit of
time more than the body of $P$. To capture this fact we make use of the
auxiliary runtime transformer
\mbox{$\eeet{\:\cdot\:}{\eta} \colon \Runtimes \To \Runtimes$} (\cf expectation
transformer $\ewp{\:\cdot\:}{\theta}$). This transformer behaves as $\eetsymbol$
except that for defining its action on a procedure call, it relies on a so--called \emph{runtime environment} $\eta$ in
$\RtEnv \eqdef \{\eta \mid \eta \colon\Runtimes \To \Runtimes \text{ is upper
  continuous} \}$
instead of on a procedure declaration. Concretely,
$\eeetd{\Call{\PName}}{\decl}{\eta}$ takes continuation $\rt$ to $\eta(\rt)$ and
for all other program constructs, $\eeet{\:\cdot\:}{\eta}$ follows the same rule as
$\eet{\:\cdot\:}$. Using this transformer we can (implicitly) define
$\eetd{\Call{\PName}}{\decl}$ by the equation
\[
\eetd{\Call{\PName}}{\decl} \:=\: \ctertenv{1} \oplus
\eeet{\decl}{\eetd{\Call{\PName}}{\decl}}~,
\]
where $\ctertenv{1} = \lambda \rt\!:\! \Runtimes \mydot \ctert{1}$ represents
the constantly $\ctert{1}$ runtime transformer and ``$\oplus$'' the point--wise
sum between runtime transformers, \ie for
$\gamma_1, \gamma_2 \colon \Runtimes \To \Runtimes$, we let
$(\gamma_1 \oplus \gamma_2) (\rt) \eqdef \gamma_1(\rt) + \gamma_2(\rt)$. The
above equation leads to the fixed point characterization of
$\eetd{\Call{\PName}}{\decl}$ in \autoref{fig:eet}. 

We remark that, as opposed to \wllpsymbol, it is not posible to define the
action $\eetd{\Call{\PName}}{\decl}$ of \eetsymbol\ on a procedure call in terms
of its action $\eet{\Calln{\PName}{n}{\decl}}$ on the finite inlinings. This is
because when computing $\eet{\Calln{\PName}{n}{\decl}}\!(\rt)$, to be correct
the transformer should add one unit of time each time a procedure call was
inlined, and this is not recoverable from
$\Calln{\PName}{n}{\decl}$.\footnote{If we adopt a model where the runtime of a
  procedure call coincides with the runtime of its body, we could just take
  $\eetd{\Call{\PName}}{\decl}\!(\rt) = \sup_n
  \eet{\Calln{\PName}{n}{\decl}}\!(\rt)$.}

This concludes our definition of the transformer \eetsymbol. We devote the remainder
of the section to study several of its properties. We begin with
\autoref{thm:eet-prop} summarizing some algebraic properties.

\begin{theorem}[Basic properties of $\eetsymbol$]
\label{thm:eet-prop}
For any program $\prog{c}{\decl}$, any constant runtime
$\ctert{k} = \lambda s \mydot k$ for $k \in \PosReals$, any $\rt, u \in \Runtimes$, and any increasing $\omega$--chain $\chain{\rt}{\preceq}$
of runtimes, it holds:
\begin{flushleft}
\begin{tabular}{@{}l@{\hspace{2em}}l}
   \textnormal{Continuity:} & $\sup\nolimits_n \eetd{c}{\decl}\!(\rt_n) \:=\:
                              \eetd{c}{\decl}\!(\sup\nolimits_n \rt_n)$; \\[1ex]
	\textnormal{Monotonicity:}			& $\rt \preceq u ~\implies~ \eetd{c}{\decl}\!(\rt) \,\preceq\,
\eetd{c}{\decl}\!(u)$;\\[1ex]
	\textnormal{Propagation} 			& $\eetd{c}{\decl}\!(\ctert{k} + \rt) \:=\: \ctert{k} + \eetd{c}{\decl}\!(\rt)$\\
	\textnormal{of constants:}				& provided $\prog{c}{\decl}$ is $\Abort$--free;\\[1ex]
	\textnormal{Preservation} 	& $\eetd{c}{\decl}\!(\infty) \:=\: \infty$\\
  \textnormal{of infinity:}						&  provided $\prog{c}{\decl}$ is $\Abort$--free.
\end{tabular}
\end{flushleft}
\end{theorem}
\begin{proof}
Monotonicity follows from continuity. Other properties are prooven by induction on $c$; see Appendix~\ref{sec:eet-basic}. 
\end{proof}

The next result establishes a connection between
$\eetsymbol$ and $\wpsymbol$. 
\begin{theorem}
\label{thm:eet-wp}
For every program $\prog{c}{\decl}$ and runtime $\rt$,
\[
\eetd{c}{\decl}\!(\rt) \:=\: \eetd{c}{\decl}\!(\ctert{0}) +  \wpd{c}{\decl}\!(\rt)~.
\]
\end{theorem}
\begin{proof}
By induction on the program structure, considering the stronger version of the
statement
%and by proving
%
\[
\eetd{c}{\decl}\!(\rt_1 + \rt_2) \:=\: \eetd{c}{\decl}\!(\rt_1) +
\wpd{c}{\decl}\!(\rt_2)~.
\]
See Appendix~\ref{sec:eet-wp} for details. 
\end{proof}

\autoref{thm:eet-wp} allows giving a very short proof of a
well--known result relating expected runtimes and termination probabilities: If a program has finite
expected runtime, it terminates almost surely.

\begin{theorem}
\label{thm:rt-ast}
For every $\Abort$--free program $\prog{c}{\decl}$ and initial state $s$ of the
program,
\[
\eetd{c}{\decl}\!(\ctert{0})(s) < \infty \implies \wpd{c}{\decl}\!(\ctert{1})(s) = 1~.
\]
\end{theorem}
\begin{proof}
  By instantiating \autoref{thm:eet-wp} with $t=\ctert{1}$ and using the
  propagation of constants property of $\eetsymbol$ (\autoref{thm:eet-prop}) to
  decompose $\eetd{c}{\decl}\!(\ctert{1})$ as
  $\ctert{1} + \eetd{c}{\decl}\!(\ctert{0})$.
% The result follows immediately from \autoref{thm:eet-wp} and the propagation of
% constant property of \eetsymbol\ (\autoref{thm:eet-prop}). Formally,
% %
% \begin{align*}
% \begin{array}{c@{\:\:} l@{} }
% & \true \\[2pt]
% %
% \Rightarrow & \qquad \by{\autoref{thm:eet-wp}}\\[2pt]
% %
% & \eetd{c}{\decl}\! (\ctert{1}) (s) \:=\: \eetd{c}{\decl}\! (\ctert{0})(s)+
% \wpd{c}{\decl} \! (\ctert{1})(s) \\[2pt]
% %
% \Leftrightarrow & \qquad \by{\autoref{thm:eet-prop}}\\[2pt]
% %
% & 1 + \eetd{c}{\decl}\! (\ctert{0}) (s) \:=\: \eetd{c}{\decl}\!(\ctert{0})(s) +
% \wpd{c}{\decl}\! (\ctert{1})(s) \\[2pt]
% %
% \Leftrightarrow & \qquad \by{$\eetd{c}{\decl}(\ctert{0})(s) < \infty$}\\[2pt]
% %
% & 1 \:=\: \wpd{c}{\decl}\! (\ctert{1})(s) 
% \end{array} \\[-\normalbaselineskip]\tag*{\qedhere}
% \end{align*}
\end{proof}

Observe that in \autoref{thm:rt-ast} we cannot drop the $\Abort$--free
requirement on the program. To see this, consider the program
$c = \PChoice{\Skip}{\nicefrac{1}{2}}{\Abort}$. The program has a finite runtime
($\eet{c}\!(\ctert{0}) = \nicefrac{\ctert 1}{\ctert 2} < \boldsymbol\infty $)
and terminates, however, with probability less than one
($\wp{c}\!(\ctert{1}) = \nicefrac{\ctert 1}{\ctert 2} < \ctert 1 $). Moreover,
observe that \autoref{thm:rt-ast} is only valid on the stated direction: A
probabilistic program can terminate almost--surely and require, still, an
expected infinite time to reach termination. This phenomenon is illustrated, for
instance, by the one dimensional random walk; see \eg~\cite[\S
7]{Kaminski:ETAPS:2016}.

Even though \autoref{thm:rt-ast} constitutes a well--known and natural result on
probabilistic programs, our contribution here is to give the first fully formal
proof of such a result.

\subsection{Proof Rules for Recursive Programs}
The runtime of procedure calls, which includes, in particular, recursive
programs, is defined using fixed points. To avoid reasoning about fixed points
we propose some proof rules based on invariants.

We show that an adaptation of the proof rules for procedure calls from our
$\wpsymbol$--calculus is sound for the $\eetsymbol$--calculus.
% This is a remarkable feature of the calculi since we can handle under a uniform
% framework two different aspects of programs: functional properties and runtime.
%
The rules are:
\smallskip
\[
\begin{array}{l}
\small
\begin{array}{c}
\infrule{\deriv{~\eet{\Call{\PName}}\!(\rt) \preceq  \ctert{1} {+} u }{
  \eet{\decl(\PName)}\!(\rt) \preceq u~}}{\eetd{\Call{\PName}}{\decl}\!(\rt)
  \preceq \ctert{1} {+} u}%
\,\lrule{eet-rec}
\end{array}\\[4ex]
\begin{array}{c}
\infrule{ l_0=\CteFun{0}, \qquad \qquad u_0=\CteFun{0}, \\
 \hspace{0em}\ctert{1} {+} l_n \preceq \eet{\Call{\PName}}\!(\rt) \preceq  \ctert{1} {+} u_n \\
      \hspace{2em} \derivsymbol l_{n+1} \preceq \eet{\decl(\PName)}\!(\rt) \preceq u_{n+1}} 
{~\ctert{1} {+} \sup_n l_n \preceq \eetd{\Call{\PName}}{\decl}\!(\rt) \preceq \ctert{1} {+} \sup_n u_n~}%
\, \lrule{eet-rec$_\omega$}\\[2ex]
\end{array}
\normalsize
\end{array}
\]
%

% Observe that we have decomposed the counterpart of rule \lrule{wp-rec$_\omega$}
% into two different rules: \lrule{eet-rec$_\omega^{\scriptscriptstyle \succeq}$}
% for bounding the runtime of a procedure call from below and
% \lrule{eet-rec$_\omega^{\scriptscriptstyle \preceq}$} for bounding the runtime
% of a procedure call from above.

\noindent Compared to the proof rules from the \wpsymbol--calculus, these proof rules
require incrementing by one unit some of the bounds. Loosely speaking, this is
because the runtime of a procedure call is one plus the runtime of its body,
whereas the semantics of a procedure call fully agrees with the semantics of its
body.

\begin{example}
  To illustrate the use of the rules, consider the faulty factorial procedure with
  declaration
\begin{align*}
\decl(\PName_{\mathsf{fact}}) &\boldsymbol{\colon}\; 
\Cond{x \leq 0}{\{\Ass{y}{1}}{ \PChoice{c_1}{\nicefrac{5}{6}}{c_2};\, \Ass{y}{y
                                \!\cdot\! x}}\,,
\end{align*}
where $c_1 = \Ass{x}{x{-}1};\, \Call{\PName_{\mathsf{fact}}};\, \Ass{x}{x{+}1}$ and $c_2 =
\Ass{x}{x{-}2}; \allowbreak \Call{\PName_{\mathsf{fact}}};\allowbreak\, \Ass{x}{x{+}2}$. We
prove that on input $x=k \geq 0$, the expected runtime of the procedure is $2 +
\alpha_k$, where
\[
\alpha_k \:=\: \frac{1}{49} \Bigl(121 + 210 k + 432 \bigl( -\tfrac{1}{6} \bigr)^{k+1} \Bigr)~.
\]
Since the term $432 (\nicefrac{-1}{6})^{k+1}$ is negligible, we can approximate the
procedure's runtime by $4.5 + 4.3k$. We can formally capture our exact runtime
assertion by
\[
\eetd{\Call{\PName_{\mathsf{fact}}}}{\decl}\!(\ctert{0}) \:=\:  \ctert{1} + \sup\nolimits_n t_n~,
\]
where $t_n = \ctert{1} + \ToExp{x<0} \cdot \ctert{1} + \ToExp{0 \leq x \leq n} \cdot
\alpha_x \,+ \ToExp{x > n} \cdot \alpha_{n+1}$.
To see this, observe that the sequence $\langle \alpha_k \rangle$ is increasing and
therefore,
$\sup\nolimits_n t_n = \ctert{1} + \ToExp{x<0} \cdot \ctert{1} + \ToExp{0 \leq
  x} \cdot \alpha_x$.
We prove the runtime assertion using rule \lrule{eet-rec$_\omega$} with
instantiations $t = \ctert{0}$ and $l_n = u_n = t_n$ for $n\geq 1$. We have to
discharge the premise
\[
\eet{\Call{\PName_{\mathsf{fact}}}}\!(\ctert{0}) = \ctert{1} + t_n \derivsymbol
\eet{\decl(\PName_{\mathsf{fact}})}\!(\ctert{0}) = t_{n+1}~.
\]
Since some simple calculations yield
\begin{multline*}
  \eet{\decl(\PName_{\mathsf{fact}})}\!(\CteFun{0}) \:=\: \ctert{1} +
  \ToExp{x \leq 0} \cdot \CteFun{1} \\
 + \ToExp{x > 0} \cdot \bigl( \tfrac{5}{6} \cdot \eet{c_1}\!(\CteFun{1})
 +\tfrac{1}{6} \cdot \eet{c_2}\!(\CteFun{1}) \big)~,
\end{multline*}
our next step is to compute $\eet{c_1}\!(\CteFun{1})$ (the calculations 
are identical for $\eet{c_2}\!(\CteFun{1})$). To do so, we rely on assumption 
$\eet{\Call{\PName}}\!(\ctert{0}) = \ctert{1} + t_n$ and the propagation of
constants property of $\eetsymbol$.
\begin{align*}
\eet{c_1}\!(\CteFun{1}) 
& \:=\: \eet{\Ass{x}{x{-}1};\, \Call{\PName_{\mathsf{fact}}}}
  \bigl(\eet{\Ass{x}{x{+}1}}\!(\CteFun{1}) \bigr)\\
&\:=\: \CteFun{2} + \eet{\Ass{x}{x{-}1};\, \Call{\PName_{\mathsf{fact}}}}
  \!(\CteFun{0})\\
&\:=\: \CteFun{2} + \eet{\Ass{x}{x{-}1}} \!(\CteFun{1} + t_n)\\
&\:=\: \CteFun{4} + t_n\subst{x}{x+1}
\end{align*}
The derivation then concludes by showing that 
\begin{multline*}
  t_{n+1} ~=~ \ctert{1} +
  \ToExp{x \leq 0} \cdot \CteFun{1} \\
 + \ToExp{x > 0} \cdot \Bigl( \tfrac{5}{6} \bigl(\CteFun{4} + t_n\subst{x}{x{+}1} \bigr) + \tfrac{1}{6}  \bigl(\CteFun{4} + t_n\subst{x}{x{+}2} \bigr) \Big)~,
\end{multline*}
which after some term reordering reduces to proving that $\alpha_0 =1$, $\alpha_1 =7$ and
$\alpha_{k+2} =5 + \tfrac{5}{6} \alpha_{k+1} + \tfrac{1}{6} \alpha_{k}$. \hfill
$\triangle$
\end{example}

We conclude the section establishing the soundness of the rules.
\begin{theorem}[Soundness of rules \lrule{eet-rec},
  \lrule{eet-rec$_\omega$}]
  \label{thm:eet-rules-sound}
  Rules \lrule{eet-rec} and \lrule{eet-rec$_\omega$} are sound \wrt the
  \eetsymbol--calculus in \autoref{fig:eet}.
\end{theorem}
\begin{proof}
See Appendix~\ref{sec:eet-rules-sound}.
\end{proof}
% Consider rule \lrule{eet-rec} and let runtime environment $\eta^\star$ map $\rt$
% to $u$ and all other runtimes to (the constant runtime)
% $\boldsymbol{\infty}$. Applying Park's Lemma\footnote{If $H\colon \mathcal{D}
%   \To \mathcal{D}$ is an upper continuous function over an upper $\omega$--cpo
%   $(\mathcal{D},\sqsubseteq)$ with bottom element, then $H(d) \sqsubseteq d$
%   implies $\lfp{\sqsubseteq}{H} \sqsubseteq d$ for every $d \in
%   \mathcal{D}$.}~\cite{Wechler:MTCS:92}  one can show that the
% conclusion of the rule follows from $u \succeq
% \eeet{\decl(\PName)}{\eta^\star}(t)$, which can be discharged using the rule
% premise. See Appendix~\ref{sec:eet-rules-sound} for details. 

% Consider now rule \lrule{eet-rec$_\omega^{\scriptscriptstyle \succeq}$} and let
% $F(\eta) = \ctertenv{1} \oplus \eeet{\decl(\PName)}{\eta}$. Exploiting the
% continuity of $F$ we show that $\eetd{\Call{\PName}}{\decl}$ can be obtained by
% fixed point iteration from the runtime environment
% $\bot_\RtEnv = \lambda \rt\!:\!\Runtimes \mydot \ctert{0}$, that is,
% $\eetd{\Call{\PName}}{\decl} \!(\rt)= \sup\nolimits_n F^n(\bot_\RtEnv) (\rt)$,
% where $F^n(\bot_\RtEnv) = F(\ldots F(F(\bot_\RtEnv)) \ldots)$ denotes the
% repeated application of $F$ from $\bot_\RtEnv$ $n$ times. The proof then
% concludes by showing that the rule premise entails
% $\ctert{1} + l_n \preceq F^{n+1}(\bot_\RtEnv)(\rt)$; see
% Appendix~\ref{sec:eet-rules-sound} for details. The case of the remainder rule
% \lrule{eet-rec$_\omega^{\scriptscriptstyle \preceq}$} is analogous.

\section{Operational Semantics}
\label{sec:operational}
We provide an operational semantics for \pRGCL programs in terms of pushdown
Markov chains with rewards (PRMC) \cite{DBLP:journals/fmsd/BrazdilEKK13} and
prove the transformer $\wpsymbol$ to be sound with respect to this
semantics. Due to space limitations, this section contains an informal
introduction only. Corresponding formal definitions are
found in Appendix~\ref{sec:operational-model}.

For simplicity, we assume a canonical labeling for each command $c \in \Cmd$ together with auxiliary functions $\Init$, $\SuccOneSymbol$, $\SuccTwoSymbol$ and $\StmtOfLabelSymbol$ determining the initial location, the first and second successor of a location and the program statement corresponding to a label.
As an example, the labels attached to each statement of program $c$ from \autoref{ex:puzz-3-upp-bound} are as follows:
\begin{align*}
 c \, \boldsymbol{\colon}\; \;
 \{ \Skip^{1} \} \: [\nicefrac{1}{2}]^{2} \: \{
  \Call{\PName}^{3};\, \Call{\PName}^{4};\, \Call{\PName}^{5}~\}~.
\end{align*}
The definition of the auxiliary functions is straightforward. For instance, we
have $\Init(c) = 2$, $\SuccOneSymbol(1) = \Term$, $\SuccTwoSymbol(2) = 3$, and
$\StmtOfLabel{2} = c$, where $\Term$ is a special symbol indicating termination of a procedure. Moreover, label $\Sink$ stands for termination of the whole program.

Our operational semantics of \pRGCL programs is given as an execution relation, where each step is of the form
\begin{align*}
 \OpTrans{\OpState{\ell}{s}}{\gamma,\, p,\, \gamma'}{\OpState{\ell'}{s'}}.
\end{align*}
Here, $\ell,\ell'$ are program labels, $s,s' \in \State$ are program states, $\gamma$ is a program label being popped from and $\gamma'$ a finite sequence of labels being pushed on the stack, respectively. $p \in [0,1]$ denotes the probability of executing this step.

This execution relation corresponds to the transition relation of a PRMC, where each pair $\OpState{\ell}{s}$ is a state and the stack alphabet is given by the set of all labels of a given \pRGCL program. Moreover, given $f \in \UEX$, a reward of $f(s)$ is assigned to each state of the form $\OpState{\Sink}{s}$. Otherwise, the reward of a state is $0$.
\autoref{fig:operational} shows the rules defining the operational semantics of \pRGCL programs.
The rules in \autoref{fig:operational} are self--explanatory. In case of a procedure call, the calls successor label is pushed on the stack and execution continues with the called procedure. Whenever a procedure terminates, i.e. reaches a state $\OpState{\Term}{s}$, and the stack is non--empty, a return address is popped from and execution continues at this address. 

\autoref{fig:operational:example} shows the PRMC of example program $c$.
The initial state is $2$ (the probabilistic choice).
Say the right branch is chosen; we move to 3.
The statement at 3 is a call, and the address after the call is 4; so 4 is pushed and the procedure body is reentered.
Say now the left branch is chosen; we move to 1 (the $\Skip$) and then terminate, i.e.\ we move to $\downarrow$.
Recall that return address 4 is on top of the stack;  4 is popped, we move to 4 to continue execution.

The expected reward that PRMC $\Automaton$ associated to program
$\prog{c}{\decl}$ reaches a set of target states
$\mathcal{T}$ from initial state $\OpState{\ell}{s}$ is defined as
\begin{align*}
 \ExpRew{\OPRMC{c,\decl}{s}{f}}{\mathcal{T}} ~=~ \sum_{\pi \in \Pi(\OpState{\ell}{s},\mathcal{T})} \Prob{\Automaton}{\pi} \cdot \rew{\pi}~,
\end{align*}
where $\pi$ is a path from $\OpState{\ell}{s}$ to some target state, $\Prob{\Automaton}{\pi}$ is the probability of $\pi$ and $\rew{\pi}$ is the reward collected along $\pi$.

\begin{figure*}[th]
\small
 \begin{align*}
    & \infrule{\StmtOfLabel{\ell} = \Skip \quad \SuccOne{\ell} = \ell'}{\OpTrans{\OpState{\ell}{s}}{\gamma,\, 1,\, \gamma}{\OpState{\ell'}{s}}}%
    \,\lrule{\textrm{skip}}%
      \qquad\qquad\qquad
    \infrule{\StmtOfLabel{\ell} = \Ass{x}{E} \quad \SuccOne{\ell} = \ell'}{\OpTrans{\OpState{\ell}{s}}{\gamma,\, 1,\, \gamma}{\OpState{\ell'}{s\big[x \mapsto s(E)\big]}}}%
    \,\lrule{\textrm{assign}}%
      \qquad\qquad\qquad
    \infrule{\StmtOfLabel{\ell} = \Abort}{\OpTrans{\OpState{\ell}{s}}{\gamma,\, 1,\, \gamma}{\OpState{\ell}{s}}}%
    \,\lrule{\textrm{abort}}%
    \\[.15\baselineskip]
    & \infrule{\StmtOfLabel{\ell} = \Cond{G}{c_1}{c_2} \quad s \models G \quad \SuccOne{\ell} = \ell'}{\OpTrans{\OpState{\ell}{s}}{\gamma,\, 1,\, \gamma}{\OpState{\ell'}{s}}}%
    \,\lrule{\textrm{if1}}%
     \qquad\qquad\quad~
    \infrule{\StmtOfLabel{\ell} = \Cond{G}{c_1}{c_2} \quad s \not\models G \quad \SuccTwo{\ell} = \ell'}{\OpTrans{\OpState{\ell}{s}}{\gamma,\, 1,\, \gamma}{\OpState{\ell'}{s}}}%
    \,\lrule{\textrm{if2}}
    \\[.15\baselineskip]
    & \infrule{\StmtOfLabel{\ell} = \PChoice{c_1}{p}{c_2} \quad \SuccOne{\ell} = \ell'}{\OpTrans{\OpState{\ell}{s}}{\gamma,\, p,\, \gamma}{\OpState{\ell'}{s}}}%
    \,\lrule{\textrm{prob1}}%
    \qquad\qquad\qquad\qquad\qquad\qquad\qquad\qquad~~
    \infrule{\StmtOfLabel{\ell} = \PChoice{c_1}{p}{c_2} \quad \SuccTwo{\ell} = \ell'}{\OpTrans{\OpState{\ell}{s}}{\gamma,\, 1-p,\, \gamma}{\OpState{\ell'}{s}}}%
    \,\lrule{\textrm{prob2}}
    \\[.15\baselineskip]
    & \infrule{\StmtOfLabel{\ell} = \Call{\PName} \quad \SuccOne{\ell} = \ell'}{\OpTrans{\OpState{\ell}{s}}{\gamma,\, 1,\, \gamma \cdot \ell' }{\OpState{\Init\big(\decl(P)\big)}{s}}}%
    \,\lrule{\textrm{call}}%
    \qquad\qquad\qquad~\,
    \infrule{\vphantom{\downarrow}}{\OpTrans{\OpState{\Term}{s}}{\ell',\, 1,\, \varepsilon}{\OpState{\ell'}{s}}}%
    \,\lrule{\textrm{return}}%
    \qquad\qquad\qquad~\,
    \infrule{\vphantom{\downarrow}}{\OpTrans{\OpState{\Term}{s}}{\gamma_0,\, 1,\, \gamma_0}{\OpState{\Sink}{s}}}%
    \,\lrule{\textrm{terminate}}
 \end{align*}
\normalsize
 \caption{Rules for defining an operational semantics for \pRGCL
   programs. For sequential composition there is no dedicated rule as the control flow is encoded via the $\SuccOneSymbol$ and the $\SuccTwoSymbol$ functions.}
 \label{fig:operational}
\end{figure*}

\begin{figure}[t]
 \begin{center}
  \scalebox{0.8}{
\begin{tikzpicture}[node distance = 2cm]
   \node [state] (2) {$2$};
   \node[above of = 2, node distance = 1cm] (0) {};
   \node [state, left of = 2] (3) {$3$};
   \node [state, right of = 2] (1) {$1$};
   \node [state, right of = 1] (t) {$\Term$};
   \node [state, above right of = t] (4) {$4$};
   \node [state, below right of = t] (5) {$5$};
   \node [state, right of = t] (s) {$\Sink$};
   
   \path
   (0) edge[->] (2)
   (2) edge[->] node [above] {\scriptsize{$\gamma,\nicefrac 1 2, \gamma$}} (3)
   (2) edge[->] node [above] {\scriptsize{$\gamma,\nicefrac 1 2, \gamma$}} (1)
   (1) edge[->] node [above] {\scriptsize{$\gamma,1, \gamma$}} (t)
   (3) edge[->, bend left=60] node [above] {\scriptsize{$\gamma,1, \gamma \cdot 4$}} (2)
   (t) edge[->] node [right] {\scriptsize{$4,1,\varepsilon$}} (4)
   (t) edge[->, loop above] node [above] {\scriptsize{$\Term,1,\varepsilon$}} (s)
   (t) edge[->] node [above] {\scriptsize{$\gamma_0,1,\gamma_0$}} (s)
   (t) edge[->] node [right] {\scriptsize{$5,1,\varepsilon$}} (5)
   (4) edge[->, bend right] node [above] {\scriptsize{$\gamma,1, \gamma \cdot 5$}} (2)
   (5) edge[->, bend left] node [below] {\scriptsize{$\gamma,1, \gamma \cdot \Term$}} (2)
   %(init) edge [] node [right] {\scriptsize{$\nicefrac 1 2$}} (sdummy)
   ;
\end{tikzpicture}
}
 \end{center}
 \caption{PRMC of program $c$ from \autoref{ex:puzz-3-upp-bound}. Since $c$ affects no variables, the second component of states is omitted.}
 \label{fig:operational:example}
\end{figure}
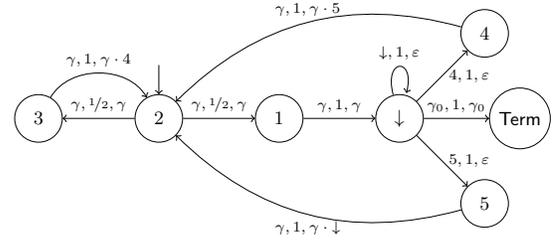

We are now in a position to state the relationship between the operational model and the denotational semantics:
\begin{theorem}[Correspondence Theorem]
\label{thm:correspondance}
 Let $c \in \Cmd$, $f \in \UEX$, and $\mathcal{T} = \{\OpState{\Sink}{s} ~|~ s
 \in \State\}$~\footnote{$\mathcal{T}$ denotes the set of states representing
   successful termination of the pushdown automaton.}. Then for each $s \in \State$, we have
 \begin{align*}
    \ExpRew{\OPRMC{c,\decl}{s}{f}}{\mathcal{T}} \:=\: \wpd{c}{\decl}\!(f)(s)~.
 \end{align*}
\end{theorem}
\begin{proof}
See Appendix~\ref{sec:eet-soundness}.
\end{proof}

In the spirit of \cite{DBLP:journals/pe/GretzKM14} a similar result can be obtained for $\textsf{wlp}$.
For that one needs a liberal expected reward being defined as the expected reward plus the probability of not reaching the target states at all.
One can then show a similar correspondence to $\textsf{wlp}$.

\section{Extensions}
\label{sec:extensions}
%We discuss some extensions and adaptations of our results.

\paragraph{Mutual recursion.}
% For the sake of presentation simplicity, we have so far assumed the presence of
% a single procedure in all our development.  
Both our \wpsymbol-- and \eetsymbol--calculus can be extended to handle multiple
procedures. 
Say we want to handle $m$ (possibly mutually recursive) procedures $\PName_1,\ldots, \PName_m$ with
declaration $\decl \in \Cmd^m$. The definition of $\wpd{\Call{\PName_i}}{\decl}$
remains the same, we only need to adapt the definition of the $n$-inlining
$\Calln{\PName_i}{n}{\decl}$ of procedure $\PName_i$ as to inline the calls of
all procedures:
\begin{align*}
%\Calln{\PName_i}{0}{\decl} &\:=\:  \Abort\\
\Calln{\PName_i}{n+1}{\decl} &\:=\: \decl(P)[
\Call{\PName_1}/ \Calln{\PName_1}{n}{\decl}, \ldots,\Call{\PName_m}/ \Calln{\PName_m}{n}{\decl}]~.
\end{align*}
As for the \eetsymbol--calculus, a runtime environment is now a tuple
$\eta = (\eta_1, \ldots, \eta_m)$, where $\eta_i$ is meant to provide the
behavior of procedure $\PName_i$ in $\eeet{\cdot}{\eta}$, \ie
$\eeet{\Call{\PName_i}}{\eta} = \eta_i$. The action of \eetsymbol\ on procedure
calls is then defined simultaneously as\footnote{For determining the
  \emph{least} fixed point, environments are compared component-wise, \ie
  $ (\eta_1, \ldots, \eta_m) \allowbreak \sqsubseteq (\nu_1, \ldots, \nu_m)$
  iff $\eta_i \sqsubseteq \nu_i$ for all $i=1 \ldots m$.}
\begin{multline*}
\big(\eetd{\Call{\PName_1}}{\decl},\, \ldots,\, \eetd{\Call{\PName_m}}{\decl} \!\big) ~=\\
 \lfpsymbol \Bigl( \lambda \eta \mydot \left(
\ctertenv{1} \oplus \eeet{\decl(\PName_1)}{\eta},\, \ldots,\, \ctertenv{1} \oplus
\eeet{\decl(\PName_m)}{\eta}  \right) \Bigr)~.
\end{multline*}

\noindent The proof rules for reasoning about procedure calls in both calculi
are easily adapted. We show only the case of \lrule{wp-rec}; the others admit a
similar adaptation.
\begin{flushleft}
\resizebox{\columnwidth}{!}%
{%
$
%\small 
\begin{array}{@{}l@{}}
\infrule{
\!\!\begin{array}{c@{\ }c@{\ \ }l@{}}
  \wp{\Call{\PName_1}}\!(f_1) \preceq g_1, ...\, ,
  \wp{\Call{\PName_m}}\!(f_m) \preceq g_m& \derivsymbol &\wp{\decl(\PName_1)}\!(f_1) \preceq
  g_1\\[-1ex]
      & \vdots & \\[0.25ex]
  \wp{\Call{\PName_1}}\!(f_1) \preceq g_1, ...\, ,
  \wp{\Call{\PName_m}}\!(f_m) \preceq g_m& \derivsymbol &\wp{\decl(\PName_m)}\!(f_m) \preceq
  g_m
\end{array}
}
{\wpd{\Call{\PName_i}}{\decl}\!(f_i) \preceq g_i \quad \text{for all $i=1\ldots m$}}
\end{array}
$}
\end{flushleft}
The rule reasons about all the procedures simultaneously. Roughly speaking, the
rule premise requires deriving the specification $g_i$  for the body of each procedure $\PName_i$,
assuming the corresponding specification for each procedure call in it. The rule
conclusion establishes the specification of the set of procedures
altogether.

\paragraph{Random samplings.} All our results remain valid if the \pRGCL
language allows for random samplings (from distributions with discrete
support). In a random sampling $\Ass{x}{\mu}$, $\mu$ represents a probability
distribution which is sampled and its outcome is assigned to program variable
$x$. In \autoref{sec:casestudy} we exploit this extension to model a
probabilistic variant of the binary search.

\paragraph{Alternative runtime models.}
The \eetsymbol--calculus can be easily adapted to capture alternative runtime
models. For instance we can capture the model where we are interested in
counting only the number of procedure calls and also more fine--grained models
such as that where the time consumed by an assignment (or guard evaluation)
depends on some notion of \emph{size} of the expression being assigned (guard
being evaluated). Likewise, the \eetsymbol--calculus can be easily adapted so 
as to take into account the costs of flipping the (possibly biased) coin from
probabilistic choices.

\paragraph{Soundness of the \boldeetsymbol--calculus.}
We can also establish the soundness of the  \eetsymbol--calculus \wrt the
operational semantics based on PRMC. This only requires changes in the reward
function.

\section{Case Study}
\label{sec:casestudy}
In this section we show the applicability of our approach analyzing a
probabilistic, so--called Sherwood~\cite{McConnell:2008}, variant of the binary
search. The main difference \wrt the classical version is that in each recursive
call the pivot element is picked uniformly at random from the remaining array,
aligning this way worst--, best-- and average--case of the algorithm runtime.
% Such so called Sherwood algorithms
% \cite{sherwood} can be of interest e.g.\ in privacy--sensitive scenarios where
% an attacker should not be able to infer information via timing attacks, i.e.\
% simply by observing the run--time of the algorithm.

The algorithm we analyze searches for value $\mathit{val}$ in array
$a[\mathit{left}\,..\,\mathit{right}]$. It is encoded by procedure $B$ with
declaration $\decl$ presented in \autoref{fig:binarysearch:partial:exists}. We
use random assignment
$\Ass{\mathit{mid}}{\textsf{uniform}(\mathit{left},\, \mathit{right}})$ to model
the random election of the pivot. For simplicity, we assume that the random
assignment is performed in constant time $1$ if
$\mathit{left} \leq \mathit{right}$ and that it diverges if
$\mathit{left} > \mathit{right}$.

\paragraph{Partial correctness.}
We verify the following partial correctness property: When $B$ is invoked in a
state where $\mathit{left} \leq \mathit{right}$,
$a[\mathit{left}\,..\,\mathit{right}]$ is sorted, and $\mathit{val}$ occurs in
$a[\mathit{left}\,..\,\mathit{right}]$, then the invocation of $B$ stores in
$\mathit{mid}$ the index where $\mathit{val}$ lies.  Formally,
\begin{align*}
	g ~\preceq~ 	&\wlpd{\Call B}{\decl}\!(f),~ \text{with}\\
	g ~=~ 		& [\mathit{left} \leq \mathit{right}] \cdot \big[\textsf{sorted}(\mathit{left},\, \mathit{right})\big]\\
				& {} \cdot \big[\exists x \in [\mathit{left},\, \mathit{right}]\colon a[x] = \mathit{val}\big]\\
	f ~=~ 		& \big[a[\mathit{mid}] = \mathit{val}\big]~,
\end{align*}
where $\big[\textsf{sorted}(y,\, z)\big]$ is the indicator function of $a[y\,..\,z]$ being sorted.
In order to prove $g \preceq \wlp{\Call B}\!(f)$ we apply rule
[\textsf{wlp-rec}]. We are then left to prove 
\[
\deriv{g \preceq \wlp{\Call B}\!(f)}{g \preceq \wlp{\decl(B)}\!(f)}~.
\]
The way in which we propagate post--expectation $f$ from the exit point
of the procedure till its entry point, obtaining pre--expectation $g$, is fully detailed in
\autoref{fig:binarysearch:partial:exists}. To do so we use assumption $g \preceq
\wlp{\Call B}\!(f)$ and monotonicity of $\wlpsymbol$. 
\begin{figure}[t]
\begin{align*}
%
%    \text{\scriptsize $(*)$}
\text{\scriptsize $g{\preceq}$}&\;\; \text{\footnotesize $\textcolor{gray}{\frac{[\mathit{left} {<} \mathit{right}]}{\mathit{right} {-} \mathit{left} {+} 1} \sum_{i = \mathit{left}}^{\mathit{right}}\left(\!\!\!\!\begin{array}{l}\big[a[\mathit{i}] {<} \mathit{val}\big] {\cdot} g[\mathit{left}/\textsf{min}(i+1, \mathit{right})]\\+ \big[a[\mathit{i}] {>} \mathit{val}\big] {\cdot} g[\mathit{right}/\textsf{max}(i-1, \mathit{left})]\\+ \big[a[\mathit{i}] {=} \mathit{val}\big]\end{array}\!\!\!\!\!\right)}$}\\[-2pt]
    &\hphantom{\boldsymbol{\colon}}\;\; \text{\footnotesize $\textcolor{gray}{~ + [\mathit{left} = \mathit{right}] \cdot \big[a[\mathit{left}] = \mathit{val}\big]}$}\\[-2pt]
\text{\scriptsize 1} &\boldsymbol{\colon}\;\; \boldsymbol{\Ass{\mathit{mid}}{\textbf{\textsf{uniform}}(\mathit{left},\, \mathit{right})};}\\[-4pt]
    &\hphantom{\boldsymbol{\colon}}\;\; \text{\footnotesize $\textcolor{gray}{[\mathit{left} < \mathit{right}]\cdot\Big(\big[a[\mathit{mid}] < \mathit{val}\big] \cdot g[\mathit{left}/\cdots]}$}\\[-4pt]
 &\hphantom{\boldsymbol{\colon}}\;\; \text{\footnotesize $\textcolor{gray}{\qquad + \big[a[\mathit{mid}] > \mathit{val}\big] \cdot g[\mathit{right}/\cdots]}$}\\[-3pt]
    &\hphantom{\boldsymbol{\colon}}\;\; \text{\footnotesize
      $\textcolor{gray}{\qquad + \big[a[\mathit{mid}] = \mathit{val}\big]\Big) + [\mathit{left} \geq \mathit{right}] \cdot f}$}\\[-4pt]
\text{\scriptsize 2} &\boldsymbol{\colon}\;\; \boldsymbol{\textbf{\textsf{if}} ~ (\mathit{left} < \mathit{right})\{}\\[-2pt]
    &\hphantom{\boldsymbol{\colon}}\;\; \qquad\text{\footnotesize $\textcolor{gray}{\big[a[\mathit{mid}] {<} \mathit{val}\big] {\cdot} g[\mathit{left}/\cdots] + \big[a[\mathit{mid}] {>} \mathit{val}\big] {\cdot} g[\mathit{right}/\cdots]}$}\\[-2pt]
    &\hphantom{\boldsymbol{\colon}}\;\; \qquad\text{\footnotesize $\textcolor{gray}{\qquad + \big[a[\mathit{mid}] {=} \mathit{val}\big] \cdot f}$}\\[-2pt]
\text{\scriptsize 3} &\boldsymbol{\colon}\;\; \qquad\boldsymbol{\textbf{\textsf{if}} ~ (a[\mathit{mid}] < \mathit{val})\{}\\[-2pt]
    &\hphantom{\boldsymbol{\colon}}\;\; \qquad\qquad \text{\footnotesize $\textcolor{gray}{g[\mathit{left}/\textsf{min}(\mathit{mid} + 1,\, \mathit{right})]}$}\\[-4pt]
\text{\scriptsize 4} &\boldsymbol{\colon}\;\; \qquad\qquad\boldsymbol{\Ass{\mathit{left}}{\textbf{\textsf{min}}(\mathit{mid} + 1,\, \mathit{right})};}\\[-5pt]
    &\hphantom{\boldsymbol{\colon}}\;\; \qquad\qquad \text{\footnotesize $\textcolor{gray}{g}$}\\[-3pt]
\text{\scriptsize 5} &\boldsymbol{\colon}\;\; \qquad\qquad\boldsymbol{\Call B}\\[-4pt]
    &\hphantom{\boldsymbol{\colon}}\;\; \qquad\qquad \text{\footnotesize $\textcolor{gray}{f}$}\\[-3pt]
\text{\scriptsize 6} &\boldsymbol{\colon}\;\; \qquad\boldsymbol{\} ~ \textbf{\textsf{else}} ~ \{}\\[-2pt]
    &\hphantom{\boldsymbol{\colon}}\;\; \qquad\qquad\text{\footnotesize $\textcolor{gray}{\big[a[\mathit{mid}] > \mathit{val}\big] \cdot g[\mathit{right}/\cdots] + \big[a[\mathit{mid}] \leq \mathit{val}\big] \cdot f}$}\\[-2pt]
\text{\scriptsize 7} &\boldsymbol{\colon}\;\; \qquad\qquad\boldsymbol{\textbf{\textsf{if}} ~ (a[\mathit{mid}] > \mathit{val})\{}\\[-2pt]
    &\hphantom{\boldsymbol{\colon}}\;\; \qquad\qquad\qquad \text{\footnotesize $\textcolor{gray}{g[\mathit{right}/\textsf{max}(\mathit{mid} - 1,\, \mathit{left})]}$}\\[-2pt]
\text{\scriptsize 8} &\boldsymbol{\colon}\;\; \qquad\qquad\qquad\boldsymbol{\Ass{\mathit{right}}{\textbf{\textsf{max}}(\mathit{mid} - 1,\, \mathit{left})};}\\[-6pt]
    &\hphantom{\boldsymbol{\colon}}\;\; \qquad\qquad\qquad \text{\footnotesize $\textcolor{gray}{g}$}\\[-4pt]
\text{\scriptsize 9} &\boldsymbol{\colon}\;\; \qquad\qquad\qquad\boldsymbol{\Call B}\\[-5pt]
    &\hphantom{\boldsymbol{\colon}}\;\; \qquad\qquad\qquad \text{\footnotesize $\textcolor{gray}{f}$}\\[-4pt]
\text{\scriptsize 10} &\boldsymbol{\colon}\;\; \qquad\qquad\boldsymbol{\} ~ \textbf{\textsf{else}} ~ \{}~ \text{\footnotesize $\textcolor{gray}{f}$} ~ \boldsymbol{\Skip}~ \text{\footnotesize $\textcolor{gray}{f}$} ~\boldsymbol{\}} ~ \text{\footnotesize $\textcolor{gray}{f}$} \\[-2pt]
\text{\scriptsize 11} &\boldsymbol{\colon}\;\; \qquad\boldsymbol{\}} ~ \text{\footnotesize $\textcolor{gray}{f}$} \\[-2pt]
\text{\scriptsize 12} &\boldsymbol{\colon}\;\; \boldsymbol{\} ~ \textbf{\textsf{else}} ~ \{}~ \text{\footnotesize $\textcolor{gray}{f}$} ~ \boldsymbol{\Skip}~ \text{\footnotesize $\textcolor{gray}{f}$} ~\boldsymbol{\}}~ \text{\footnotesize $\textcolor{gray}{f}$}
\end{align*}
\caption{Declaration $\text{\small $\mathcal D$}$ (boldface) of the
  probabilistic binary search procedure $B$ together with the proof (lightface)
  that $\Call B$ finds the index of $\mathit{val}$ when started in a sorted
  array $a[\mathit{left}\,..\,\mathit{right}]$ which contains value
  $\mathit{val}$.
 We write $\text{\footnotesize \textcolor{gray}{$j$}} ~ \boldsymbol{C}
  ~ \text{\footnotesize \textcolor{gray}{$h$}}$ for $j \preceq \wp{C}(h)$. 
% Recall that $g = [\mathit{left} \leq \mathit{right}] \cdot
% \big[\textsf{sorted}(\mathit{left},\, \mathit{right})\big] \cdot \big[\exists x
% \in [\mathit{left},\, \mathit{right}]\colon a[x] = \mathit{val}\big]$ and $f =
% \big[a[\mathit{mid}] = \mathit{val}\big]$, and that we assume $g \preceq
% \wlp{\Call B}(f)$.
}
\label{fig:binarysearch:partial:exists}
\end{figure}
%
%

% We will retrace the crucial points of the proof:
% By assumption we have $g \preceq \wlp{\Call B}(f)$ for both calls occurring in the declaration.
% From there, $(*)$ is obtained by applying the definition of $\textsf{wlp}$ while assuming $g = \wlp{\Call B}(f)$.
% By monotonicity of $\textsf{wlp}$ the annotation $(*)$ is then a safe lower bound for $\wlp{\decl}(f)$.

% It is left to show that $g \preceq (*)$:
% For any state $\sigma$ we have $g(\sigma) \in \{0,\, 1\}$ as $g$ is a standard predicate.
% Therefore we need to consider only the states in which $g$ becomes $1$, i.e.\ those states in which either
% %
% \begin{enumerate}
% 	\item $\mathit{left} = \mathit{right}$ and $\mathit{val}$ occurs in $a[\mathit{left}\,..\,\mathit{right}]$, then $[\mathit{left} = \mathit{right}] \cdot \big[a[\mathit{left}] = \mathit{val}\big]$ becomes $1$ as well, or
% %
% 	\item $\mathit{left} < \mathit{right}$, $a[\mathit{left}\,..\,\mathit{right}]$ is sorted, and $\mathit{val}$ occurs in $a[\mathit{left}\,..\,\mathit{right}]$, then at any splitting index $i \in [\mathit{left},\, \mathit{right}]$ either $a[i] < \mathit{val}$ and the array right of $i$ is sorted and contains $\mathit{val}$, or $a[i] > \mathit{val}$ and the array left of $i$ is sorted and contains $\mathit{val}$, or $a[i] = \mathit{val}$. For any choice of $i$ thus the summand becomes 1. As there are exactly $\mathit{right} - \mathit{left} + 1$ summands and $\mathit{left} < \mathit{right}$ the whole term again becomes 1.
% \end{enumerate}
%
Dually, we can verify that when $\mathit{val}$ is not in
the array, the value of $a[\mathit{mid}]$ after termination of $B$ is different
from $\mathit{val}$. A detailed derivation of this property is provided in
Appendix \ref{app:casestudy}, \autoref{fig:binarysearch:partial:not-exists}.

% Analogously, we can verify that whenever $B$ is invoked in a state where $\mathit{left} \leq \mathit{right}$, $a[\mathit{left}\,..\,\mathit{right}]$ is sorted, and the value $\mathit{val}$ \emph{does not occur} in $a[\mathit{left}\,..\,\mathit{right}]$, invocation of $B$ finds some index $\mathit{mid}$ such that the value at this position is unequal to $\mathit{val}$.
% Formally, we can prove
% \begin{align*}
% 	g ~\preceq~ 	&\wlpd{\Call B}{\decl}(f),~ \text{with}\\
% %
% 	g ~=~ 		& [\mathit{left} \leq \mathit{right}] \cdot \big[\textsf{sorted}(\mathit{left},\, \mathit{right})\big]\\
% 				& \qquad \cdot \big[\forall x \in [\mathit{left},\, \mathit{right}]\colon a[x] \neq \mathit{val}\big]\\
% 	f ~=~ 		& \big[a[\mathit{mid}] \neq \mathit{val}\big]~,
% \end{align*}
% again by applying the [\textsf{wlp-rec}]--rule.
% For details, see Appendix \ref{app:casestudy}, \autoref{fig:binarysearch:partial:not-exists}.

\paragraph{Expected runtime.}
We perform a runtime analysis of the algorithm for those inputs where
$\mathit{val}$ does not occur in the array. Under this assumption we can
distinguish two cases: either $\mathit{val}$ is smaller than every element in the
array or larger than all of them. 
%
% Next, we analyze the expected runtime of the algorithm.
% As the runtime depends heavily on the input data, we will suppose certain knowledge about the input data:
% We will assume the second scenario for which we have proven the partial correctness of the algorithm, namely that $\mathit{val}$ does not occur in $a[\mathit{left}\,..\,\mathit{right}]$.
% In this scenario there are two cases:
% Either (1) $\mathit{val}$ is larger than every value occurring in $a[\mathit{left}\,..\,\mathit{right}]$, or (2) $\mathit{val}$ is smaller than every value occurring in $a[\mathit{left}\,..\,\mathit{right}]$.
%

For the first case we show that the expected runtime of the algorithm is upper bounded by $\boldsymbol{1} + u$, with
\begin{align*}
	u 	~=~ &[\mathit{left} > \mathit{right}] \cdot \boldsymbol{\infty} + \boldsymbol{3} \\
		&{} + [\mathit{left} < \mathit{right}] \cdot \left(\boldsymbol{5} \cdot H_{\mathit{right} - \mathit{left} + 1} - \boldsymbol{\nicefrac{5}{2}}\right)~,
\end{align*}
and $H_k$ bing the $k$-th harmonic number. 
%For showing that $u$ is in fact an upper bound of the expected runtime of invoking $B$, we have to show
Formally, we show that
\begin{align*}
	\eet{\Call B}\!(\ctert{0}) ~\preceq~ \ctert{1} + u~
\end{align*}
applying rule $\lrule{eet-rec}$. We must then establish
\[
\deriv{\eet{\Call B}\!(\ctert{0}) \preceq \ctert{1} + u}{\eet{\decl}\!(\ctert{0}) \preceq u}~.
\] 
%
%
% which can be done by application of the $\lrule{eet-rec}$ rule.
% By this rule it suffices to show that assuming $\eet{\Call B}(\boldsymbol{0}) \preceq \boldsymbol{1} + u$ one can prove $\eet{\decl}(\boldsymbol{0}) \preceq u$.
%
The details of this derivation are provided in \autoref{fig:binarysearch:runtime:alwayssmaller}.
\begin{figure}[t]
\begin{align*}
%
%\text{\scriptsize (**)}
\text{\scriptsize $u{ = }$}&\;\; \text{\footnotesize $\textcolor{gray}{[\mathit{left} > \mathit{right}] \cdot \boldsymbol{\infty} + \boldsymbol{3}  + [\mathit{left} < \mathit{right}]}$}\\[-2pt]
    &\hphantom{\boldsymbol{\colon}}\;\; \text{\footnotesize $\textcolor{gray}{{} \cdot \left( \boldsymbol{5} + \sum_{i = \mathit{left}}^{\mathit{right}}\left(\begin{array}{l}\frac{[\textsf{min}(i+1,\, \mathit{right})  < \mathit{right}]}{\mathit{right} - \mathit{left} + 1}\\ {} \cdot \left(\boldsymbol{5} \cdot H_{\mathit{right} - \textsf{min}(i+1,\, \mathit{right})  < \mathit{right} + 1}\right.\\ \left.\qquad{} - \boldsymbol{\nicefrac{5}{2}} \right)\end{array}\right)\right)}$}\\[-2pt]
\text{\scriptsize 1} &\boldsymbol{\colon}\;\; \boldsymbol{\Ass{\mathit{mid}}{\textbf{\textsf{uniform}}(\mathit{left},\, \mathit{right})};}\\[-4pt]
    &\hphantom{\boldsymbol{\colon}}\;\; \text{\footnotesize $\textcolor{gray}{\boldsymbol{2} + [\mathit{left} < \mathit{right}]\cdot\Big(\boldsymbol{2} + \big[a[\mathit{mid}] < \mathit{val}\big] \cdot ( \boldsymbol{3} }$}\\[-3pt]
    &\hphantom{\boldsymbol{\colon}}\;\; \text{\footnotesize $\textcolor{gray}{\qquad\qquad {}+ [\textsf{min}(\mathit{mid} + 1,\, \mathit{right}) < \mathit{right}]}$}\\[-3pt]
    &\hphantom{\boldsymbol{\colon}}\;\; \text{\footnotesize $\textcolor{gray}{\qquad\qquad\qquad {}\cdot \left(\boldsymbol{5}\cdot H_{\mathit{right} - \textsf{min}(\mathit{mid} + 1,\, \mathit{right}) + 1}  - \boldsymbol{\nicefrac{5}{2}}\right)}$}\\[-3pt]    
 &\hphantom{\boldsymbol{\colon}}\;\; \text{\footnotesize $\textcolor{gray}{\qquad {} + \big[a[\mathit{mid}] > \mathit{val}\big] \cdot (\cdots)\Big)}$}\\[-4pt]
\text{\scriptsize 2} &\boldsymbol{\colon}\;\; \boldsymbol{\textbf{\textsf{if}} ~ (\mathit{left} < \mathit{right})\{}\\[-2pt]
    &\hphantom{\boldsymbol{\colon}}\;\; \qquad\text{\footnotesize $\textcolor{gray}{ \boldsymbol{3} + \big[a[\mathit{mid}] < \mathit{val}\big] \cdot u[\mathit{left}/\textsf{min}(\mathit{mid}+1,\, \mathit{right})]}$}\\[-2pt]
    &\hphantom{\boldsymbol{\colon}}\;\; \qquad\text{\footnotesize $\textcolor{gray}{{} + \big[a[\mathit{mid}] > \mathit{val}\big] \cdot (\cdots)}$}\\[-2pt]
\text{\scriptsize 3} &\boldsymbol{\colon}\;\; \qquad\boldsymbol{\textbf{\textsf{if}} ~ (a[\mathit{mid}] < \mathit{val})\{}\\[-2pt]
    &\hphantom{\boldsymbol{\colon}}\;\; \qquad\qquad \text{\footnotesize $\textcolor{gray}{\boldsymbol{2} + u[\mathit{left}/\textsf{min}(\mathit{mid}+1,\, \mathit{right})]}$}\\[-2pt]
\text{\scriptsize 4} &\boldsymbol{\colon}\;\; \qquad\qquad\boldsymbol{\Ass{\mathit{left}}{\textbf{\textsf{min}}(\mathit{mid} + 1,\, \mathit{right})};}\\[-3pt]
    &\hphantom{\boldsymbol{\colon}}\;\; \qquad\qquad \text{\footnotesize $\textcolor{gray}{\boldsymbol{1} + u}$}\\[-3pt]
\text{\scriptsize 5} &\boldsymbol{\colon}\;\; \qquad\qquad\boldsymbol{\Call B}\\[-5pt]
    &\hphantom{\boldsymbol{\colon}}\;\; \qquad\qquad \text{\footnotesize $\textcolor{gray}{\boldsymbol{0}}$}\\[-5pt]
\text{\scriptsize 6} &\boldsymbol{\colon}\;\; \qquad\boldsymbol{\} ~ \textbf{\textsf{else}} ~ \{}\\[-4pt]
    &\hphantom{\boldsymbol{\colon}}\;\; \qquad\qquad\text{\footnotesize $\textcolor{gray}{\boldsymbol{2} +\big[a[\mathit{mid}] > \mathit{val}\big] \cdot (\cdots)}$}\\[-2pt]
\text{\scriptsize 7} &\boldsymbol{\colon}\;\; \qquad\qquad\boldsymbol{\textbf{\textsf{if}} ~ (a[\mathit{mid}] > \mathit{val})\{}\\[-3pt]
    &\hphantom{\boldsymbol{\colon}}\;\; \qquad\qquad\qquad \text{\footnotesize $\textcolor{gray}{\boldsymbol{2} + u[\mathit{right}/\textsf{max}(\mathit{mid}-1,\, \mathit{left})]}$}\\[-3pt]
\text{\scriptsize 8} &\boldsymbol{\colon}\;\; \qquad\qquad\qquad\boldsymbol{\Ass{\mathit{right}}{\textbf{\textsf{max}}(\mathit{mid} - 1,\, \mathit{left})};}\\[-4pt]
    &\hphantom{\boldsymbol{\colon}}\;\; \qquad\qquad\qquad \text{\footnotesize $\textcolor{gray}{\boldsymbol{1} + u}$}\\[-3pt]
\text{\scriptsize 9} &\boldsymbol{\colon}\;\; \qquad\qquad\qquad\boldsymbol{\Call B}\\[-5pt]
    &\hphantom{\boldsymbol{\colon}}\;\; \qquad\qquad\qquad \text{\footnotesize $\textcolor{gray}{\boldsymbol{0}}$}\\[-5pt]
\text{\scriptsize 10} &\boldsymbol{\colon}\;\; \qquad\qquad\boldsymbol{\} ~ \textbf{\textsf{else}} ~ \{}~ \text{\footnotesize $\textcolor{gray}{\boldsymbol{1}}$} ~ \boldsymbol{\Skip}~ \text{\footnotesize $\textcolor{gray}{\boldsymbol{0}}$} ~\boldsymbol{\}} ~ \text{\footnotesize $\textcolor{gray}{\boldsymbol{0}}$} \\
\text{\scriptsize 11} &\boldsymbol{\colon}\;\; \qquad\boldsymbol{\}} ~ \text{\footnotesize $\textcolor{gray}{\boldsymbol{0}}$} \\
\text{\scriptsize 12} &\boldsymbol{\colon}\;\; \boldsymbol{\} ~ \textbf{\textsf{else}} ~ \{}~ \text{\footnotesize $\textcolor{gray}{\boldsymbol{1}}$} ~ \boldsymbol{\Skip}~ \text{\footnotesize $\textcolor{gray}{\boldsymbol{0}}$} ~\boldsymbol{\}}~ \text{\footnotesize $\textcolor{gray}{\boldsymbol{0}}$}
\end{align*}
\caption{Runtime analysis of the probabilistic binary search procedure for the
  case that every value occurring in $a[\mathit{left}\,..\,\mathit{right}]$ is
  smaller than $\mathit{val}$. We write $\text{\footnotesize
    \textcolor{gray}{$j$}} ~ \boldsymbol{C} ~ \text{\footnotesize
    \textcolor{gray}{$h$}}$ for $j \succeq \eet{C}(h)$.
% Recall that $u = \boldsymbol{1} + [\mathit{left} > \mathit{right}] \cdot
% \boldsymbol{\infty} + \boldsymbol{3} + [\mathit{left} < \mathit{right}] \cdot
% \left(\boldsymbol{5} \cdot H_{\mathit{right} - \mathit{left} + 1} -
%   \boldsymbol{\nicefrac{5}{2}}\right)$ and that we assume $\eet{\Call
%   B}(\boldsymbol{0}) \preceq \boldsymbol{1} + u$.
}
\label{fig:binarysearch:runtime:alwayssmaller}
\end{figure}
%
%
% We will retrace the crucial points of the proof:
% By assumption we have $\eet{\Call B}(\boldsymbol{0}) \preceq \boldsymbol{1} + u$ for both calls occurring in the declaration.
% From there, $(**)$ is obtained by applying the definition of $\textsf{ert}$ while assuming $\eet{\Call B}(\boldsymbol{0}) = \boldsymbol{1} + u$ and $a[i] < \mathit{val}$ for all $i \in [\mathit{left},\,  \mathit{right}]$, and by recalling that the uniform sampling is assumed to have infinite runtime if $\mathit{left} > \mathit{right}$.
% By monotonicity of $\textsf{ert}$ the annotation $(**)$ is then a safe upper bound for $\eet{\decl}(\boldsymbol{0})$.

% It is left to show that $(**) \preceq u$ (in fact we can even show equality):
% Let $n = \mathit{right} - \mathit{left} + 1$ and consider that both $\mathit{right} \not< \mathit{right}$ and $\mathit{right} +1 \not< \mathit{right}$ (so the last two summands are 0).
% Consider furthermore that the equality
% \begin{align*}
% 	\left. \sum_{i = 2}^{n-1}\middle( 5 \cdot H_i - \frac 5 2 \right) ~=~ \frac 5 2 \cdot n \cdot (2\cdot H_n - 3)
% \end{align*}
% holds.
% Then $(**)$ can be rewritten as:
% \begin{align*}
% 		&[\mathit{left} > \mathit{right}] \cdot \boldsymbol{\infty} + \boldsymbol{3} \\
% 		& {}+ [\mathit{left} < \mathit{right}] \left( 5 + \frac{5\cdot n \cdot \left(2\cdot H_n - 3\right)}{n \cdot 2}\right)\\
% 	=~	&[\mathit{left} > \mathit{right}] \cdot \boldsymbol{\infty} + \boldsymbol{3} + [\mathit{left} < \mathit{right}] \left(5\cdot H_n - \nicefrac 5 2\right) ~=~ u 
% \end{align*}
%

Similarly, when $\mathit{val}$ is greater than every element in the array, the
expected runtime is upper bounded by $\boldsymbol{1} + u$, with
% Analogously, we can verify that an upper bound for the runtime of invoking $B$ in Case (2), i.e.\ when all values in $a[\mathit{left}\,..\,\mathit{right}]$ are larger than $\mathit{val}$, is given by $\boldsymbol{1} + u$ with  
%
\begin{align*}
	u 	~=~ &[\mathit{left} > \mathit{right}] \cdot \boldsymbol{\infty} + \boldsymbol{3} \\
		&{} + [\mathit{left} < \mathit{right}] \cdot \left(\boldsymbol{6} \cdot H_{\mathit{right} - \mathit{left} + 1} - \boldsymbol{3}\right)~.
\end{align*}
The verification for this case is analogous therefore omitted.

Combining the two cases we conclude that when the sought--after value does not occur in
the array, the algorithm terminates in expected time in $\Theta\big(\log
n\big)$, where $n = \mathit{right} - \mathit{left} + 1$ is the size of the
array, since $H_k  \in \Theta(\log k)$.

\section{Related Work}
\label{sec:related}
\paragraph{\textsf{wp}--style reasoning for recursive programs.}
Recursion has been treated for non--proba\-bilis\-tic programs.
Hesselink~\cite{Hesselink:FAC:93} provided several proof rules for recursive procedures, both for total and partial correctness. 
Our first two proof rules are extensions of his rules to the probabilistic setting.
Predicate transformer semantics for recursive non--deterministic procedures has been provided by Bonsangue and Kok~\cite{Bonsangue:FAC:94} and Hesselink~\cite{Hesselink:FAC:89}.
Nipkow~\cite{Nipkow:CSL:02} provides an operational semantics and a Hoare logic for recursive (parameterless) non--deterministic procedures.
Zhang \emph{et al.}~\cite{DBLP:conf/tphol/ZhangMHH02} establishes the equivalence between an operational semantics and a weakest pre--condition semantics for recursive programs in Coq.
To some extent our transfer theorem between probabilistic pushdown automata and the \wpsymbol--semantics can be considered as a probabilistic extension of this work.

\paragraph{Deductive reasoning for recursive probabilistic programs.}
Jones provided several proof rules for recursive probabilistic programs in her Ph.D.\ dissertation~\cite{Jones:1992}.
One of our proof rules is a generalisation of Jones' proof rule to general recursion.
McIver and Morgan~\cite{McIver:2001b} also provide a \wpsymbol--semantics of probabilistic recursive programs. 
While \cite{McIver:2001b}  use fixed point techniques, we follow~\eg~\citet{Hehner:AI:79} and define the semantics of a recursive
procedure as the limit of an approximation sequence.  
In contrast to our approach based on procedures, \cite{McIver:2001b} introduced recursion through the language constructor $\mathbf{rec}~\mathcal{B}$, where $\mathcal{B}$ is a program--semantics transformer. 
(Intuitively $\mathcal{B}$ encodes how the recursive procedure defined (and invoked) by $\mathbf{rec}~\mathcal{B}$ transforms the outcome of its recursive calls).
Our approach provides a strict separation between program syntax and semantics. 
Moreover our approach based on procedure calls can model mutual recursion in a natural way (see Section~\ref{sec:extensions}), while the approach in~\cite{McIver:2001b} approach does not accommodate so naturally to such cases.
Audebaud and Paulin-Mohring~\cite{Audebaud:2009} present a mechanized method for proving properties of randomized algorithms in the Coq proof assistant.
Their approach is based on higher--order logic, in particular using a monadic interpretation of programs as probabilistic distributions.
Our proof rule for obtaining two--sided bounds on recursive programs is directly adapted from their work.
They however do neither relate their work to an operational model and nor support the analysis of expected runtimes.

\paragraph{Semantics of recursive probabilistic programs.}
Gupta \emph{et al.} consider the interplay between constraints, probabilistic choice, and recursion in the context of a (concurrent) constraint--based probabilistic programming language.
They provide an operational semantics using labeled transition systems and (weak) bisimulation as well as a denotational semantics.
Recursion is treated operationally by considering the limit of syntactic finite approximations.
In the denotational semantics, the mixture of probabilities and constraints violates basic monotonicity properties for a standard treatment of recursion. 
Their main result is that the transition system semantics modulo weak bisimulation is fully abstract with respect to the input--output relation of processes.
They do neither consider non--determinism nor reasoning about recursive probabilistic programs.
Pfeffer and Koller~\cite{DBLP:conf/aaai/PfefferK00} provide a measure--theoretic semantics of recursive Bayesian networks and show that every recursive probabilistic relational database has a probability measure as model.
This is complemented by an inference algorithm that obtains approximations by basically unfolding the recursive Bayesian network.
Recently, Toronto \emph{et al.}~\cite{DBLP:conf/esop/TorontoMH15} provided a measure--theoretic semantics for a probabilistic programming language with recursion.
Their interpretation of recursive programs is however restricted to (almost surely) terminating programs. 

\paragraph{Probabilistic pushdown automata.}
The analysis of probabilistic pushdown automata, which correspond to the model of recursive Markov chains, has been well--investigated.
Key computational problems for analyzing classes of these models can be reduced to computing the least fixed point solution of corresponding classes of monotone polynomial systems of non--linear equations.
For subclasses of these models termination probabilities, $\omega$--regular properties, and expected runtimes can be algorithmically obtained.
Recent surveys are provided by Etessami~\cite{etessami:2016} and Brazdil \emph{et al.}~\cite{DBLP:journals/fmsd/BrazdilEKK13}.
Our transfer theorem indicates that (some of) these results are transferable to obtaining weakest pre--expectations for recursive probabilistic programs having a finite--control probabilistic push--down automata.
A detailed study is outside the scope of this paper and left for future work.

\section{Conclusion}
\label{sec:conclusion}
We have presented two wp-calculi: one for reasoning about correctness, and one for analysing expected rum-times of recursive probabilistic programs.
The wp-calculi have been related, equipped with proof rules, and exemplified by analysing a Sherwood version of binary search.
A relation with a straightforward operational interpretation using pushdown Markov chains has been established.
We believe that this work provides a good basis for the automation of the analysis of recursive probabilistic programs. 
Future work consists of applying our calculi to other recursive randomized algorithms (such as quick sort with random pivot selection).
Other future work includes investigating a generalisation of Colussi's technique~\cite{Colussi:1984} to transform a recursive program and its correctness proof into a non-recursive program with its accompanying correctness proof.  This would allow to transfer---typically simpler---correctness proofs of the recursive probabilistic programs to non-recursive ones.

% \acks
% Acknowledgments, if needed.

% We recommend abbrvnat bibliography style.

\bibliographystyle{abbrvnat}
\bibliography{biblio}

\begin{thebibliography}{33}
\providecommand{\natexlab}[1]{#1}
\providecommand{\url}[1]{\texttt{#1}}
\expandafter\ifx\csname urlstyle\endcsname\relax
  \providecommand{\doi}[1]{doi: #1}\else
  \providecommand{\doi}{doi: \begingroup \urlstyle{rm}\Url}\fi

\bibitem[Audebaud and Paulin-Mohring(2009)]{Audebaud:2009}
P.~Audebaud and C.~Paulin-Mohring.
\newblock {P}roofs of randomized algorithms in \textsf{Coq}.
\newblock \emph{Science of Comp.\ Progr.}, 74\penalty0 (8):\penalty0 568 --
  589, 2009.

\bibitem[Bonsangue and Kok(1994)]{Bonsangue:FAC:94}
M.~Bonsangue and J.~Kok.
\newblock The weakest precondition calculus: Recursion and duality.
\newblock \emph{Formal Aspects of Computing}, 6\penalty0 (1):\penalty0
  788--800, 1994.

\bibitem[Br{\'{a}}zdil et~al.(2013)Br{\'{a}}zdil, Esparza, Kiefer, and
  Kucera]{DBLP:journals/fmsd/BrazdilEKK13}
T.~Br{\'{a}}zdil, J.~Esparza, S.~Kiefer, and A.~Kucera.
\newblock Analyzing probabilistic pushdown automata.
\newblock \emph{Formal Methods in System Design}, 43\penalty0 (2):\penalty0
  124--163, 2013.

\bibitem[Carbin et~al.(2013)Carbin, Misailovic, and
  Rinard]{DBLP:conf/oopsla/CarbinMR13}
M.~Carbin, S.~Misailovic, and M.~C. Rinard.
\newblock Verifying quantitative reliability for programs that execute on
  unreliable hardware.
\newblock In \emph{Proc.\ of OOPSLA}, pages 33--52. {ACM}, 2013.

\bibitem[Colussi(1984)]{Colussi:1984}
L.~Colussi.
\newblock Recursion as an effective step in program development.
\newblock \emph{ACM Trans. Program. Lang. Syst.}, 6\penalty0 (1):\penalty0
  55--67, Jan. 1984.

\bibitem[Dean(2006)]{DBLP:journals/dam/Dean06}
B.~C. Dean.
\newblock A simple expected running time analysis for randomized ``divide and
  conquer'' algorithms.
\newblock \emph{Discrete Appl. Math.}, 154\penalty0 (1):\penalty0 1--5, 2006.

\bibitem[Dijkstra(1976)]{Dijkstra}
E.~W. Dijkstra.
\newblock \emph{{A Discipline of Programming}}.
\newblock {Prentice Hall}, 1976.

\bibitem[Etessami(2016)]{etessami:2016}
K.~Etessami.
\newblock Analysis of probabilistic processes and automata theory.
\newblock In \emph{Handbook of Automata Theory}. 2016.
\newblock (to appear).

\bibitem[Fioriti and Hermanns(2015)]{luis}
L.~M.~F. Fioriti and H.~Hermanns.
\newblock Probabilistic termination: Soundness, completeness, and
  compositionality.
\newblock In \emph{Proc.\ of POPL}, pages 489--501. ACM, 2015.

\bibitem[Gordon et~al.(2014)Gordon, Henzinger, Nori, and
  Rajamani]{DBLP:conf/icse/GordonHNR14}
A.~D. Gordon, T.~A. Henzinger, A.~V. Nori, and S.~K. Rajamani.
\newblock Probabilistic programming.
\newblock In \emph{Future of Software Engineering (FOSE)}, pages 167--181.
  {ACM}, 2014.

\bibitem[Gretz et~al.(2014)Gretz, Katoen, and
  McIver]{DBLP:journals/pe/GretzKM14}
F.~Gretz, J.-P. Katoen, and A.~McIver.
\newblock Operational versus weakest pre-expectation semantics for the
  probabilistic guarded command language.
\newblock \emph{Perform. Eval.}, 73:\penalty0 110--132, 2014.

\bibitem[Hehner(1979)]{Hehner:AI:79}
E.~Hehner.
\newblock do considered od: A contribution to the programming calculus.
\newblock \emph{Acta Informatica}, 11\penalty0 (4):\penalty0 287--304, 1979.
\newblock \doi{10.1007/BF00289091}.

\bibitem[Hesselink(1989)]{Hesselink:FAC:89}
W.~H. Hesselink.
\newblock Predicate-transformer semantics of general recursion.
\newblock \emph{Acta Informatica}, 26\penalty0 (4):\penalty0 309--332, 1989.

\bibitem[Hesselink(1993)]{Hesselink:FAC:93}
W.~H. Hesselink.
\newblock Proof rules for recursive procedures.
\newblock \emph{Formal Aspects of Computing}, 5\penalty0 (6):\penalty0
  554--570, 1993.
\newblock \doi{10.1007/BF01211249}.

\bibitem[Jones(1989)]{Jones:1992}
C.~Jones.
\newblock \emph{Probabilistic Non-determinism}.
\newblock PhD thesis, University of Edinburgh, 1989.

\bibitem[Kaminski and Katoen(2015)]{DBLP:conf/mfcs/KaminskiK15}
B.~L. Kaminski and J.~Katoen.
\newblock On the hardness of almost-sure termination.
\newblock In \emph{Prof.\ of MFCS, Part {I}}, volume 9234 of \emph{LNCS}, pages
  307--318. Springer, 2015.

\bibitem[Kaminski et~al.(2016)Kaminski, Katoen, Matheja, and
  Olmedo]{Kaminski:ETAPS:2016}
B.~L. Kaminski, J.-P. Katoen, C.~Matheja, and F.~Olmedo.
\newblock {W}eakest precondition reasoning for expected run--times of
  probabilistic programs.
\newblock In \emph{Proc.\ of ESOP}, LNCS, 2016.
\newblock To appear.

\bibitem[{Kaminski} et~al.(2016){Kaminski}, {Katoen}, {Matheja}, and
  {Olmedo}]{Kaminski:arXiv:2016}
B.~L. {Kaminski}, J.-P. {Katoen}, C.~{Matheja}, and F.~{Olmedo}.
\newblock Weakest precondition reasoning for expected run--times of
  probabilistic programs.
\newblock \emph{ArXiv e-prints}, 2016.

\bibitem[Karp(1994)]{DBLP:journals/jacm/Karp94}
R.~M. Karp.
\newblock Probabilistic recurrence relations.
\newblock \emph{J. {ACM}}, 41\penalty0 (6):\penalty0 1136--1150, 1994.

\bibitem[Kozen(1981)]{Kozen:81}
D.~Kozen.
\newblock {S}emantics of {P}robabilistic {P}rograms.
\newblock \emph{J. Comput. Syst. Sci.}, 22\penalty0 (3):\penalty0 328--350,
  1981.

\bibitem[McConnell(2008)]{McConnell:2008}
J.~McConnell.
\newblock \emph{Analysis of Algorithms -- An Active Learning Approach}.
\newblock Jones and Bartlett Publishers, Inc., 2008.

\bibitem[McIver and Morgan(2001)]{McIver:2001b}
A.~McIver and C.~Morgan.
\newblock Partial correctness for probabilistic demonic programs.
\newblock \emph{Theor. Comp. Sc.}, 266\penalty0 (1–2):\penalty0 513 -- 541,
  2001.

\bibitem[McIver and Morgan(2004)]{McIver:2004}
A.~McIver and C.~Morgan.
\newblock \emph{Abstraction, Refinement And Proof For Probabilistic Systems}.
\newblock Springer, 2004.

\bibitem[Mitzenmacher and Upfal(2005)]{Mitzenmacher:2005}
M.~Mitzenmacher and E.~Upfal.
\newblock \emph{Probability and Computing: Randomized Algorithms and
  Probabilistic Analysis}.
\newblock Cambridge University Press, 2005.

\bibitem[Morgan(1996)]{Morgan:RW:1996}
C.~Morgan.
\newblock {P}roof rules for probabilistic loops.
\newblock In \emph{Proceedings of the BCS-FACS 7th Refinement Workshop}.
  Springer, 1996.

\bibitem[Nelson(1989)]{Nelson:TOPLAS:89:}
G.~Nelson.
\newblock A generalization of {D}ijkstra's calculus.
\newblock \emph{ACM Trans. Program. Lang. Syst.}, 11\penalty0 (4):\penalty0
  517--561, Oct. 1989.

\bibitem[Nipkow(2002)]{Nipkow:CSL:02}
T.~Nipkow.
\newblock Hoare logics for recursive procedures and unbounded nondeterminism.
\newblock In \emph{Proc.\ of CSL}, volume 2471 of \emph{LNCS}, pages 103--119.
  Springer, 2002.

\bibitem[Pfeffer(2016)]{Pfeffer:2016}
A.~Pfeffer.
\newblock \emph{Practical Probabilistic Programming}.
\newblock Manning Publications, 2016.

\bibitem[Pfeffer and Koller(2000)]{DBLP:conf/aaai/PfefferK00}
A.~Pfeffer and D.~Koller.
\newblock Semantics and inference for recursive probability models.
\newblock In \emph{Proc.\ of AAAI}, pages 538--544. {AAAI} Press / The {MIT}
  Press, 2000.

\bibitem[Toronto et~al.(2015)Toronto, McCarthy, and
  Horn]{DBLP:conf/esop/TorontoMH15}
N.~Toronto, J.~McCarthy, and D.~V. Horn.
\newblock Running probabilistic programs backwards.
\newblock In \emph{Proc.\ of ESOP}, volume 9032 of \emph{LNCS}, pages 53--79.
  Springer, 2015.

\bibitem[Wechler(1992)]{Wechler:MTCS:92}
W.~Wechler.
\newblock \emph{Universal Algebra for Computer Scientists}, volume~25 of
  \emph{{EATCS} Monographs on Theor.\ Comp.\ Science}.
\newblock Springer, 1992.

\bibitem[Winskel(1993)]{Winskel:1993}
G.~Winskel.
\newblock \emph{The Formal Semantics of Programming Languages: An
  Introduction}.
\newblock {MIT} Press, 1993.

\bibitem[Zhang et~al.(2002)Zhang, Munro, Harman, and
  Hu]{DBLP:conf/tphol/ZhangMHH02}
X.~Zhang, M.~Munro, M.~Harman, and L.~Hu.
\newblock Weakest precondition for general recursive programs formalized in
  {C}oq.
\newblock In \emph{Proc.\ of TPHOL}, volume 2410 of \emph{LNCS}, pages
  332--348. Springer, 2002.

\end{thebibliography}

\clearpage
\appendix
\section{Appendix}
\label{sec:appendix}
For our proofs about transformer \wpsymbol, we observe that ``$\preceq$''
endows the set of unbounded expectations \UEX with the structure of an upper
$\omega$--cpo\footnote{Given a binary relation $\leq$ over a set $A$, we say
  that $(A,\leq)$ is an \emph{upper (resp.\ lower) $\omega$-cpo} if $\leq$ is
  reflexive, transitive and antisymmetric, and every increasing $\omega$-chain
  \chain{a}{\leq} (resp.\ decreasing $\omega$-chain \chain{a}{\geq}) in $A$ has
  a supremum $\sup_n a_n$ (resp.\ an infimum $\inf_n a_n$) in $A$.}, where the
supremum of an increasing $\omega$--chain $\chain{f}{\preceq}$ is given
pointwise, \ie $(\sup_n f_n) (s) \eqdef \sup_n \, f_n(s)$. Likewise,
``$\preceq$'' endows the set of bounded expectations \BEX with the structure of
a lower $\omega$--cpo, where the infimum of a decreasing $\omega$--chain
$\chain{f}{\succeq}$ is given pointwise, \ie
$(\inf_n f_n) (s) \eqdef \inf_n \, f_n(s)$. Upper $\omega$--cpo $(\UEX,\preceq)$
has as botom element the constant expectation $\CteFun{0}$, while lower
$\omega$--cpo $(\BEX,\preceq)$ has as top element the constant expectation
$\CteFun{1}$. 

In what follows, we usually refer to the set of upper continuous expectation
transformers\footnote{A function $f \colon A \To B$ between two upper (resp.\
  lower) $\omega$-cpos $(A,\leq_A)$ and $(B,\leq_B)$ is \emph{upper}
  (resp. \emph{lower}) \emph{continuous} iff for every increasing $\omega$-chain
  $\chain{a}{\leq_A}$ (resp.\ decreasing $\omega$-chain $\chain{a}{\geq_A}$),
  $\sup_n f(a_n) = f (\sup_n a_n)$ (resp.\ $\inf_n f(a_n) = f (\inf_n a_n)$).}
over $(\UEX, \preceq)$ and the set of lower continuous expectation transformers
over $(\BEX,\preceq)$. We use $\ucont{\UEX}{\UEX}$ and $\lcont{\BEX}{\BEX}$ to
denote such sets.

\subsection{Basic Properties of the $\textsf{w(l)p}$--Transformer}
\label{app:basicproperties}

\begin{proof}[Proof of Continuity]
We prove continuity by induction on the program structure.
Let $f_0 \preceq f_1 \preceq f_2 \preceq \cdots$ and $g_0 \succeq g_1 \succeq g_2 \succeq \cdots$
For the base cases we have:
\paragraph{$\boldsymbol{\Skip}$:}
\begin{align*}
	\wpd{\Skip}{\decl}\left(\sup_n f_n\right) ~=~ 	&\sup_n f_n ~=~ \sup_n \wpd{\Skip}{\decl}(f_n)
\intertext{and}
	\wlpd{\Skip}{\decl}\left(\inf_n g_n\right) ~=~ 	&\inf_n g_n ~=~ \inf_n \wlpd{\Skip}{\decl}(g_n)
\end{align*}
\paragraph{$\boldsymbol{\Ass{x}{E}}$:}
\begin{align*}
	\wpd{\Ass{x}{E}}{\decl}\left(\sup_n f_n\right) &~=~ \left(\sup_n f_n\right)\subst{x}{E}\\
	 &~=~ \sup_n f_n\subst{x}{E}\\
	 &~=~ \sup_n \wpd{\Ass{x}{E}}{\decl}(f_n)
\intertext{and}
	\wlpd{\Ass{x}{E}}{\decl}\left(\inf_n g_n\right) &~=~ \left(\inf_n g_n\right)\subst{x}{E}\\
	 &~=~ \inf_n g_n\subst{x}{E}\\
	 &~=~ \inf_n \wlpd{\Ass{x}{E}}{\decl}(g_n)
\end{align*}
\paragraph{$\boldsymbol{\Abort}$:}
\begin{align*}
	\wpd{\Abort}{\decl}\left(\sup_n f_n\right) &~=~ \boldsymbol{0} ~=~ \sup_n \boldsymbol{0}\\
	& ~=~ \sup_n \wpd{\Abort}{\decl}(f_n)
\intertext{and}
	\wlpd{\Abort}{\decl}\left(\inf_n g_n\right) &~=~ \boldsymbol{1}~=~ \inf_n \boldsymbol{1} \\
	& ~=~ \inf_n \wpd{\Abort}{\decl}(g_n)
\end{align*}
For the induction hypothesis we assume that for any two programs $c_1$ and $c_2$ continuity holds.
Then we can perform the induction step:
\paragraph{$\boldsymbol{\Cond{G}{c_1}{c_2}}$:}
\begin{align*}
		&\wpd{\Cond{G}{c_1}{c_2}}{\decl}\left(\sup_n f_n\right)\\
	=~	&\ToExp{G} \cdot \wpd{c_1}{\decl}\left(\sup_n f_n\right) + \ToExp{\lnot G} \cdot \wpd{c_2}{\decl}\left(\sup_n f_n\right)\\
	=~	&\ToExp{G} \cdot \sup_n \wpd{c_1}{\decl}(f_n) + \ToExp{\lnot G} \cdot \sup_n \wpd{c_2}{\decl}(f_n)\\
	=~	&\sup_n \ToExp{G} \cdot \wpd{c_1}{\decl}(f_n) + \ToExp{\lnot G} \cdot \wpd{c_2}{\decl}(f_n)\\
	=~	&\sup_n \wpd{\Cond{G}{c_1}{c_2}}{\decl}(f_n)
\intertext{and}
		&\wlpd{\Cond{G}{c_1}{c_2}}{\decl}\left(\inf_n g_n\right)\\
	=~	&\ToExp{G} \cdot \wlpd{c_1}{\decl}\left(\inf_n g_n\right) + \ToExp{\lnot G} \cdot \wpd{c_2}{\decl}\left(\inf_n g_n\right)\\
	=~	&\ToExp{G} \cdot \inf_n \wlpd{c_1}{\decl}(g_n) + \ToExp{\lnot G} \cdot \inf_n \wlpd{c_2}{\decl}(g_n)\\
	=~	&\inf_n \ToExp{G} \cdot \wlpd{c_1}{\decl}(g_n) + \ToExp{\lnot G} \cdot \wlpd{c_2}{\decl}(g_n)\\
	=~	&\inf_n \wlpd{\Cond{G}{c_1}{c_2}}{\decl}(g_n)
\end{align*}
\paragraph{$\boldsymbol{\PChoice{c_1}{p}{c_2}}$:}
\begin{align*}
		&\wpd{\PChoice{c_1}{p}{c_2}}{\decl}\left(\sup_n f_n\right)\\
	=~	&p \cdot \wpd{c_1}{\decl}\left(\sup_n f_n\right) + (1-p) \cdot \wpd{c_2}{\decl}\left(\sup_n f_n\right)\\
	=~	&p \cdot \sup_n \wpd{c_1}{\decl}(f_n) + (1-p) \cdot \sup_n \wpd{c_2}{\decl}(f_n)\\
	=~	&\sup_n p \cdot \wpd{c_1}{\decl}(f_n) + (1-p) \cdot \wpd{c_2}{\decl}(f_n)\\
	=~	&\sup_n \wpd{\PChoice{c_1}{p}{c_2}}{\decl}(f_n)
\intertext{and}
		&\wlpd{\PChoice{c_1}{p}{c_2}}{\decl}\left(\inf_n g_n\right)\\
	=~	&p \cdot \wlpd{c_1}{\decl}\left(\inf_n g_n\right) + (1-p) \cdot \wpd{c_2}{\decl}\left(\inf_n g_n\right)\\
	=~	&p \cdot \inf_n \wlpd{c_1}{\decl}(g_n) + (1-p) \cdot \inf_n \wlpd{c_2}{\decl}(g_n)\\
	=~	&\inf_n p \cdot \wlpd{c_1}{\decl}(g_n) + (1-p) \cdot \wlpd{c_2}{\decl}(g_n)\\
	=~	&\inf_n \wlpd{\PChoice{c_1}{p}{c_2}}{\decl}(g_n)
\end{align*}
\paragraph{$\boldsymbol{c_1;~ c_2}$:}
\begin{align*}
	\wpd{c_1;~ c_2}{\decl}\left(\sup_n f_n\right) ~=~	&\wpd{c_1}{\decl}\left(\wpd{c_2}{\decl}\left(\sup_n f_n\right)\right)\\
	=~	&\wpd{c_1}{\decl}\left(\sup_n \wpd{c_2}{\decl}(f_n)\right)\\
	=~	&\sup_n \wpd{c_1}{\decl}(\wpd{c_2}{\decl}(f_n))\\
	=~	&\sup_n \wpd{c_1;~c_2}{\decl}(f_n)
\intertext{and}
		\wlpd{c_1;~ c_2}{\decl}\left(\inf_n g_n\right) ~=~	&\wlpd{c_1}{\decl}\left(\wlpd{c_2}{\decl}\left(\inf_n g_n\right)\right)\\
	=~	&\wlpd{c_1}{\decl}\left(\inf_n \wlpd{c_2}{\decl}(g_n)\right)\\
	=~	&\inf_n \wlpd{c_1}{\decl}(\wpd{c_2}{\decl}(g_n))\\
	=~	&\inf_n \wlpd{c_1;~c_2}{\decl}(g_n)
\end{align*}
\paragraph{$\boldsymbol{\Call P}$:}
\begin{align*}
		\wpd{\Call P}{\decl}\left(\sup_n f_n\right) ~=~	&\sup_k\wp{\Calln{\PName}{k}{\decl}}\left(\sup_n f_n\right)
\intertext{and}
		\wlpd{\Call P}{\decl}\left(\inf_n g_n\right) ~=~	&\inf_k \wlp{\Calln{\PName}{k}{\decl}}\left(\inf_n g_n\right)
\end{align*}
Since $\Calln{\PName}{k}{\decl}$ is call--free for every $n$ and we have already proven continuity for all call--free programs, we have
\begin{align*}
	\wp{\Calln{\PName}{k}{\decl}}\left(\sup_n f_n\right) ~=~	& \sup_n \wp{\Calln{\PName}{k}{\decl}}(f_n)
\intertext{and}
	\wlp{\Calln{\PName}{k}{\decl}}\left(\inf_n g_n\right) ~=~	& \inf_n \wlp{\Calln{\PName}{k}{\decl}}(g_n)
		\end{align*}
for every $n$ and hence
\begin{align*}
	\wpd{\Call P}{\decl}\left(\sup_n f_n\right) 	&~=~ \sup_k \sup_n \wp{\Calln{\PName}{k}{\decl}}(f_n)\\
	& ~=~ \sup_n \sup_k \wp{\Calln{\PName}{k}{\decl}}(f_n)\\
	& ~=~ \sup_n \wp{\Call \PName}{\decl}(f_n)
\intertext{and}
	\wlpd{\Call P}{\decl}\left(\inf_n g_n\right) 	&~=~ \inf_k \inf_n \wlp{\Calln{\PName}{k}{\decl}}(g_n)\\
	& ~=~ \inf_n \inf_k \wlp{\Calln{\PName}{k}{\decl}}(g_n)\\
	& ~=~ \inf_n \wlp{\Call \PName}{\decl}(g_n)~.
\end{align*}
\end{proof}

\begin{proof}[Proof of Monotonicity]
Assume $f_1 \preceq f_2$. Then
\begin{align*}
	\wpd{c}{\decl}(f_2) ~&=~ \wpd{c}{\decl}(\sup \{f_1,\, f_2\})\\
	 ~&=~ \sup \{\wpd{c}{\decl}(f_1),\, \wpd{c}{\decl}(f_2)\} \tag{continuity of $\textsf{wp}$}
\intertext{which implies $\wpd{c}{\decl}(f_1) \preceq \wpd{c}{\decl}(f_2)$, and}
\wlpd{c}{\decl}(f_1) ~&=~ \wpd{c}{\decl}(\inf \{f_1,\, f_2\})\\
	 ~&=~ \inf \{\wlpd{c}{\decl}(f_1),\, \wlpd{c}{\decl}(f_2)\}~,\tag{continuity of $\textsf{wlp}$}
\end{align*}
which implies $\wlpd{c}{\decl}(f_1) \preceq \wlpd{c}{\decl}(f_2)$.
\end{proof}

\begin{proof}[Proof of Linearity]
We prove linearity by induction on the program structure.
For the base cases we have:
\paragraph{$\boldsymbol{\Skip}$:}
\abovedisplayskip=-1\baselineskip
\begin{align*}
		&\wpd{\Skip}{\decl}(\alpha_1 \cdot f_1 + \alpha_2 \cdot f_2)\\
	=~	&\alpha_1 \cdot f_1 + \alpha_2 \cdot f_2\\
	=~	&\alpha_1 \cdot \wpd{\Skip}{\decl}(f_1) + \alpha_2 \cdot \wpd{\Skip}{\decl}(f_2)
\end{align*}\normalsize
\paragraph{$\boldsymbol{\Ass{x}{E}}$:}
\begin{align*}
		&\wpd{\Ass{x}{E}}{\decl}(\alpha_1 \cdot f_1 + \alpha_2 \cdot f_2)\\
	=~	&(\alpha_1 \cdot f_1 + \alpha_2 \cdot f_2)\subst{x}{E}\\
	=~	&\alpha_1 \cdot f_1\subst{x}{E} + \alpha_2 \cdot f_2\subst{x}{E}\\
	=~	&\alpha_1 \cdot \wpd{\Ass{x}{E}}{\decl}(f_1) + \alpha_2 \cdot \wpd{\Ass{x}{E}}{\decl}(f_2)
\end{align*}
\paragraph{$\boldsymbol{\Abort}$:}
\abovedisplayskip=-1\baselineskip
\begin{align*}
		&\wpd{\Abort}{\decl}(\alpha_1 \cdot f_1 + \alpha_2 \cdot f_2)\\
	=~	&\boldsymbol{0}\\
	=~	&\alpha_1 \cdot \boldsymbol{0} + \alpha_2 \cdot \boldsymbol{0}\\
	=~	&\alpha_1 \cdot \wpd{\Abort}{\decl}(f_1) + \alpha_2 \cdot \wpd{\Abort}{\decl}(f_2)
\end{align*}
\normalsize
For the induction hypothesis we assume that for any two programs $c_1$ and $c_2$ linearity holds.
Then we can perform the induction step:
\paragraph{$\boldsymbol{\Cond{G}{c_1}{c_2}}$:}
\begin{align*}
		&\wpd{\Cond{G}{c_1}{c_2}}{\decl}(\alpha_1 \cdot f_1 + \alpha_2 \cdot f_2)\\
	=~	&\ToExp{G} \cdot \wpd{c_1}{\decl}(\alpha_1 \cdot f_1 + \alpha_2 \cdot f_2)\\
		&{} + \ToExp{\lnot G} \cdot \wpd{c_2}{\decl}(\alpha_1 \cdot f_1 + \alpha_2 \cdot f_2)\\
	=~	&\ToExp{G} \cdot \left(\alpha_1 \cdot \wpd{c_1}{\decl}(f_1) + \alpha_2 \cdot \wpd{c_1}{\decl}(f_2)\right)\\
		&{} + \ToExp{\lnot G} \cdot \left(\alpha_1 \cdot \wpd{c_2}{\decl}(f_1) + \alpha_2 \cdot \wpd{c_2}{\decl}(f_2)\right)\\
	=~	&\alpha_1 \cdot \left( \ToExp{G} \cdot \wpd{c_1}{\decl}(f_1) + \ToExp{\lnot G} \cdot \wpd{c_2}{\decl}(f_1) \right)\\
		&{} + \alpha_2 \cdot \left( \ToExp{G} \cdot \wpd{c_1}{\decl}(f_2) + \ToExp{\lnot G} \cdot \wpd{c_2}{\decl}(f_2)\right)\\
	=~	&\alpha_1 \cdot \wpd{\Cond{G}{c_1}{c_2}}{\decl}(f_1)\\
		&{} + \alpha_2 \cdot \wpd{\Cond{G}{c_1}{c_2}}{\decl}(f_2)
\end{align*}
\paragraph{$\boldsymbol{\PChoice{c_1}{p}{c_2}}$:}
\begin{align*}
		&\wpd{\PChoice{c_1}{p}{c_2}}{\decl}(\alpha_1 \cdot f_1 + \alpha_2 \cdot f_2)\\
	=~	&p \cdot \wpd{c_1}{\decl}(\alpha_1 \cdot f_1 + \alpha_2 \cdot f_2)\\
		&{} + (1-p) \cdot \wpd{c_2}{\decl}(\alpha_1 \cdot f_1 + \alpha_2 \cdot f_2)\\
	=~	&p \cdot \left(\alpha_1 \cdot \wpd{c_1}{\decl}(f_1) + \alpha_2 \cdot \wpd{c_1}{\decl}(f_2)\right)\\
		&{} + (1-p) \cdot \left(\alpha_1 \cdot \wpd{c_2}{\decl}(f_1) + \alpha_2 \cdot \wpd{c_2}{\decl}(f_2)\right)\\
	=~	&\alpha_1 \cdot \left( p \cdot \wpd{c_1}{\decl}(f_1) + (1-p) \cdot \wpd{c_2}{\decl}(f_1) \right)\\
		&{} + \alpha_2 \cdot \left( p \cdot \wpd{c_1}{\decl}(f_2) + (1-p) \cdot \wpd{c_2}{\decl}(f_2)\right)\\
	=~	&\alpha_1 \cdot \wpd{\PChoice{c_1}{p}{c_2}}{\decl}(f_1)\\
		&{} + \alpha_2 \cdot \wpd{\PChoice{c_1}{p}{c_2}}{\decl}(f_2)
\end{align*}
\paragraph{$\boldsymbol{c_1;~ c_2}$:}
\begin{align*}
		&\wpd{c_1;~ c_2}{\decl}(\alpha_1 \cdot f_1 + \alpha_2 \cdot f_2)\\
	=~	&\wpd{c_1}{\decl}(\wpd{c_2}{\decl}(\alpha_1 \cdot f_1 + \alpha_2 \cdot f_2))\\
	=~	&\wpd{c_1}{\decl}(\alpha_1 \cdot \wpd{c_2}{\decl}(f_1) + \alpha_2 \cdot \wpd{c_2}{\decl}(f_2))\\
	=~	&\alpha_1 \cdot \wpd{c_1}{\decl}(\wpd{c_2}{\decl}(f_1))\\
		&{} + \alpha_2 \cdot \wpd{c_1}{\decl}(\wpd{c_2}{\decl}(f_2))\\
	=~	&\alpha_1 \cdot \wpd{c_1;~c_2}{\decl}(f_1) + \alpha_2 \cdot \wpd{c_1;~c_2}{\decl}(f_2)
\end{align*}
\paragraph{$\boldsymbol{\Call P}$:}
\abovedisplayskip=-1\baselineskip
\begin{align*}
		&\wpd{\Call P}{\decl}(\alpha_1 \cdot f_1 + \alpha_2 \cdot f_2)\\
	=~	&\sup_n \wp{\Calln{\PName}{n}{\decl}}(\alpha_1 \cdot f_1 + \alpha_2 \cdot f_2)
\end{align*}
\normalsize
Since $\Calln{\PName}{n}{\decl}$ is call--free for every $n$ and we have already proven linearity for all call--free programs, we have
\begin{align*}
		&\wp{\Calln{\PName}{n}{\decl}}(\alpha_1 \cdot f_1 + \alpha_2 \cdot f_2)\\
	=~	&\alpha_1 \cdot \wp{\Calln{\PName}{n}{\decl}}(f_1) + \alpha_2 \cdot \wp{\Calln{\PName}{n}{\decl}}(f_2)
\end{align*}
for every $n$ and hence
\belowdisplayskip=-1\baselineskip
\begin{align*}
		&\sup_n \wp{\Calln{\PName}{n}{\decl}}(\alpha_1 \cdot f_1 + \alpha_2 \cdot f_2)\\
	=~	&\sup_n \alpha_1 \cdot \wp{\Calln{\PName}{n}{\decl}}(f_1) + \alpha_2 \cdot \wp{\Calln{\PName}{n}{\decl}}(f_2)\\
	=~	&\alpha_1 \cdot \sup_n \wp{\Calln{\PName}{n}{\decl}}(f_1) + \alpha_2 \cdot \sup_n \wp{\Calln{\PName}{n}{\decl}}(f_2)\\
	=~	&\alpha_1 \cdot \wpd{\Call P}{\decl}(f_1) + \alpha_2 \cdot \wpd{\Call P}{\decl}(f_2)
\end{align*}\normalsize
\end{proof}
\begin{proof}[Proof of Preservation of $\boldsymbol{0}$ and $\boldsymbol{1}$]
We prove preservation of $\boldsymbol{0}$ and $\boldsymbol{1}$ by induction on the program structure.
For the base cases we have:
\paragraph{$\boldsymbol{\Skip}$:}
\begin{align*}
	\wpd{\Skip}{\decl}(\boldsymbol{0}) ~=~ 	&\boldsymbol{0}
\intertext{and}
	\wlpd{\Skip}{\decl}(\boldsymbol{1}) ~=~	&\boldsymbol{1}
\end{align*}
\paragraph{$\boldsymbol{\Ass{x}{E}}$:}
\begin{align*}
	\wpd{\Ass{x}{E}}{\decl}(\boldsymbol{0}) ~=~ 	&\boldsymbol{0}\subst{x}{E} ~=~ \boldsymbol{0}
\intertext{and}
	\wlpd{\Ass{x}{E}}{\decl}(\boldsymbol{1}) ~=~ 	&\boldsymbol{1}\subst{x}{E} ~=~ \boldsymbol{1}
\end{align*}
\paragraph{$\boldsymbol{\Abort}$:}
\begin{align*}
	\wpd{\Abort}{\decl}(\boldsymbol{0}) ~=~	&\boldsymbol{0}
\intertext{and}
	\wlpd{\Abort}{\decl}(\boldsymbol{1}) ~=~	&\boldsymbol{1}
\end{align*}
For the induction hypothesis we assume that for any two programs $c_1$ and $c_2$ preservation of $\boldsymbol{0}$ and $\boldsymbol{1}$ holds.
Then we can perform the induction step:
\paragraph{$\boldsymbol{\Cond{G}{c_1}{c_2}}$:}
\begin{align*}
		&\wpd{\Cond{G}{c_1}{c_2}}{\decl}(\boldsymbol{0})\\
	=~	&\ToExp{G} \cdot \wpd{c_1}{\decl}(\boldsymbol{0}) + \ToExp{\lnot G} \cdot \wpd{c_2}{\decl}(\boldsymbol{0})\\
	=~	&\ToExp{G} \cdot \boldsymbol{0} + \ToExp{\lnot G} \cdot \boldsymbol{0}\\
	=~	&\boldsymbol{0}
\intertext{and}
		&\wlpd{\Cond{G}{c_1}{c_2}}{\decl}(\boldsymbol{1})\\
	=~	&\ToExp{G} \cdot \wlpd{c_1}{\decl}(\boldsymbol{1}) + \ToExp{\lnot G} \cdot \wlpd{c_2}{\decl}(\boldsymbol{1})\\
	=~	&\ToExp{G} \cdot \boldsymbol{1} + \ToExp{\lnot G} \cdot \boldsymbol{1}\\
	=~	&\boldsymbol{1}
\end{align*}
\paragraph{$\boldsymbol{\PChoice{c_1}{p}{c_2}}$:}
\begin{align*}
		&\wpd{\PChoice{c_1}{p}{c_2}}{\decl}(\boldsymbol{0})\\
	=~	&p \cdot \wpd{c_1}{\decl}(\boldsymbol{0}) + (1-p) \cdot \wpd{c_2}{\decl}(\boldsymbol{0})\\
	=~	&p \cdot \boldsymbol{0} + (1-p) \cdot \boldsymbol{0}\\
	=~	&\boldsymbol{0}
\intertext{and}
		&\wlpd{\Cond{G}{c_1}{c_2}}{\decl}(\boldsymbol{1})\\
	=~	&p \cdot \wlpd{c_1}{\decl}(\boldsymbol{1}) + (1-p) \cdot \wlpd{c_2}{\decl}(\boldsymbol{1})\\
	=~	&p \cdot \boldsymbol{1} + (1-p) \cdot \boldsymbol{1}\\
	=~	&\boldsymbol{1}
\end{align*}
\paragraph{$\boldsymbol{c_1;~ c_2}$:}
\begin{align*}
	\wpd{c_1;~ c_2}{\decl}(\boldsymbol{0}) ~=~	&\wpd{c_1}{\decl}(\wpd{c_2}{\decl}(\boldsymbol{0}))\\
	=~	&\wpd{c_1}{\decl}(\boldsymbol{0})\\
	=~	&\boldsymbol{0}
\intertext{and}
		\wlpd{c_1;~ c_2}{\decl}(\boldsymbol{1}) ~=~	&\wlpd{c_1}{\decl}(\wlpd{c_2}{\decl}(\boldsymbol{1}))\\
	=~	&\wlpd{c_1}{\decl}(\boldsymbol{1})\\
	=~	&\boldsymbol{1}
\end{align*}
\paragraph{$\boldsymbol{\Call P}$:}
\begin{align*}
		\wpd{\Call P}{\decl}(\boldsymbol{0}) ~=~		&\sup_n \wp{\Calln{\PName}{n}{\decl}}(\boldsymbol{0})
\intertext{and}
		\wlpd{\Call P}{\decl}(\boldsymbol{1}) ~=~		&\inf_n \wlp{\Calln{\PName}{n}{\decl}}(\boldsymbol{1})
\end{align*}
Since $\Calln{\PName}{n}{\decl}$ is call--free for every $n$ and we have already proven preservation of $\boldsymbol{0}$ and $\boldsymbol{1}$ for all call--free programs, we have
\begin{align*}
	\wp{\Calln{\PName}{n}{\decl}}(\boldsymbol{0}) ~=~	& \boldsymbol{0}
\intertext{and}
	\wlp{\Calln{\PName}{n}{\decl}}(\boldsymbol{1}) ~=~	& \boldsymbol{1}
		\end{align*}
for every $n$ and hence
\begin{align*}
	\wpd{\Call P}{\decl}(\boldsymbol{0}) ~=~ 	&\sup_n \wp{\Calln{\PName}{n}{\decl}}(\boldsymbol{0}) ~=~ \boldsymbol{0}
\intertext{and}
	\wlpd{\Call P}{\decl}(\boldsymbol{1}) ~=~ 	&\inf_n \wlp{\Calln{\PName}{n}{\decl}}(\boldsymbol{1}) ~=~ \boldsymbol{1}~.
\end{align*}
\end{proof}

\subsection{Fixed Point Characterization of Recursive Procedures}
\label{sec:app-fixed-point-sem}

Establishing the results from \autoref{thm:fp-rec} requires a subsidiary
result connecting $\wllp{\cdot}$ with $\ewllp{\cdot}{}$ in the presence of
non--recursive procedure calls.
%\
\begin{lemma}
\label{thm:wp-ewp}
 For every command $c$ and closed command $c'$,
\[
\ewp{c}{\wp{c'}}  \:=\: \wpd{c}{\PName \triangleright c'}~.
\]
\end{lemma}
\begin{proof}
  By induction on the structure of $c$. Except for procedure calls, the proof
  for all other program constructs follows immediately from de definition of
  $\wp{\cdot}$, $\ewp{\cdot}{(\cdot)}$ and the inductive hypotheses in the case of
  compound instructions. For the case of procedure calls, the proof relies on
  the fact that as $c'$ is a closed command,
  $\Calln{\PName}{n}{\PName \, \triangleright \, c'} = c'$ for all $n \geq
  1$. Concretely, we reason as follows:
\begin{align*}
\begin{array}{c@{\:\:} l@{} }
& \ewp{c}{\wp{c'}}\!(f) \displaybreak[0]\\[2pt]
= & \qquad \by{def.~$\ewp{\cdot}{(\cdot)}$}\displaybreak[0]\\[2pt]
& \wp{c'}\!(f) \displaybreak[0]\\[2pt]
= & \qquad \by{sup.\ of a constant sequence}\displaybreak[0]\\[2pt]
& \sup_{n} \wp{c'}\!(f)  \displaybreak[0]\\[2pt]
= & \qquad \by{observation above}\displaybreak[0]\\[2pt]
& \sup_{n} \wp{\Calln{\PName}{n+1}{\PName \, \triangleright \, c'}} \!(f)  \\[6pt]
= & \qquad \by{$\wp{\Calln{\PName}{0}{\PName \, \triangleright \, c'}} \!(f) = \CteFun{0}$}\\[6pt]
& \sup_{n} \wp{\Calln{\PName}{n}{\PName \, \triangleright \, c'}}  \!(f)  \\[6pt]
= & \qquad \by{def.~\wp{\cdot}}\displaybreak[0]\\[2pt]
& \wpd{\Call{\PName}}{\PName \triangleright c'}
\end{array} \\[-\normalbaselineskip]\tag*{\qedhere}
\end{align*}

\end{proof}

Now we are in a position to prove \autoref{thm:fp-rec}. Consider first the case
of fixed point characterization
\[
\wpd{\Call{\PName}}{\decl} \:=\: \lfpsymbol_{\sqsubseteq} \, \Bigl(\underbrace{\lambda \theta\!:\!\SEnv \mydot
\ewp{\decl(\PName)}{\theta}}_F \Bigr)~.
\]
Its proof comprises two major steps:

\begin{enumerate}

\item \emph{Use the continuity of $F \colon (\SEnv, \sqsubseteq) \To
  (\SEnv, \sqsubseteq)$ established by \autoref{thm:ewp-cont-env} to 
conclude that
\[
\lfp{\sqsubseteq}{F} \:=\: \sup\nolimits_n F^n(\bot_\SEnv)~,
\]
}
where $F^n$ denotes the composition of $F$ with itself $n$ times (\ie
$F^{0}=\mathit{id}$ and $F^{n+1} = F \circ F^n$) and
$\bot_\SEnv = \lambda f\!:\!\UEX \mydot \CteFun{0}$ is the constantly
$\CteFun{0}$ environment.
\item \emph{Show that
\[
\forall f\colon\! \UEX \mydot F^n(\bot_\SEnv)(f) \:=\: \wp{\Calln{\PName}{n}{\decl}}\!(f) 
\]
for all $n\geq 0$.}  
\end{enumerate}

Then the proof follows immediately since by definition of \wpsymbol, we have
\begin{multline*}
\wpd{\Call{\PName}}{\decl}\!(f) \:=\: \sup\nolimits_{n}
\wp{\Calln{\PName}{n}{\decl}}\!(f) \\ 
\:=\: \sup\nolimits_n F^n(\bot_\SEnv) (f) \:=\: \lfp{\sqsubseteq}{F} (f)~.
\end{multline*}
 
We now consider each of these two steps in details. Step 1 follows immediately
from an application of Kleene's Fixed Point Theorem. Step 2 proceeds by
induction on $n$. The base case is straightforward:
\begin{multline*}
F^0(\bot_\SEnv)(f) \:=\: \bot_\SEnv (f) \:=\: \CteFun{0} \\
=\: \wp{\Abort}(f) \:=\: \wp{\Calln{\PName}{0}{\decl}}\!(f)~.
\end{multline*}
For the inductive case we have
\begin{align*}
\begin{array}{c@{\:\:} l@{} }
& F^{n+1}(\bot_\SEnv)(f) \displaybreak[0]\\[2pt] 
= & \qquad \by{def.~of $F^{n+1}$}\displaybreak[0]\\[2pt]
& F \bigl(F^{n}(\bot_\SEnv) \bigr) (f) \displaybreak[0]\\[2pt]
= & \qquad \by{def.~of $F$}\displaybreak[0]\\[2pt]
& \ewp{\decl(\PName)}{F^{n}(\bot_\SEnv)} \!(f)\displaybreak[0]\\[2pt]
= & \qquad \by{I.{}H.{}}\displaybreak[0]\\[2pt]
& \ewp{\decl(\PName)}{\wp{\Calln{\PName}{n}{\decl}}} \!(f)\displaybreak[0]\\[2pt]
= & \qquad \by{\autoref{thm:wp-ewp}}\displaybreak[0]\\[2pt]
& \wpd{\decl(\PName)}{\PName \triangleright \Calln{\PName}{n}{\decl}}\!(f)\displaybreak[0]\\[2pt]
= & \qquad \by{\autoref{thm:subst-env}}\displaybreak[0]\\[2pt]
& \wp{\decl(\PName)
 \subst{\Call{\PName}}{\Calln{\PName}{n}{\decl}}}\!(f)\displaybreak[0]\\[2pt]
= & \qquad \by{def. $n$-inl.}\displaybreak[0]\\[2pt]
& \wp{\Calln{\PName}{n+1}{\decl}}\!(f) \displaybreak[0]\\[2pt]
\end{array}
\end{align*}

Now we turn to the fixed point characterization 
\[
\wlpd{\Call{\PName}}{\decl} \:=\: \gfpsymbol_{\sqsubseteq} \, \Bigl(\underbrace{\lambda \theta\!:\!\LSEnv \mydot
\ewlp{\decl(\PName)}{\theta}}_G \Bigr)~.
\]
The proof follows a dual argument. We first apply Kleene's Fixed Point Theorem
to show that
\[
\gfp{\sqsubseteq}{G} \:=\: \inf\nolimits_n G^n(\top_\LSEnv)~,
\]
where $\top_\LSEnv = \lambda f\!:\!\BEX \mydot \CteFun{1}$ is the constantly
$\CteFun{1}$ environment. Next we show by induction on $n$ that
\[
\forall f\colon\! \BEX \mydot G^n(\top_\LSEnv)(f) \:=\: \wlp{\Calln{\PName}{n}{\decl}}\!(f) 
\]
The proof concludes combining these two results since
\begin{multline*}
\wlpd{\Call{\PName}}{\decl}\!(f) \:=\: \inf\nolimits_{n}
\wlp{\Calln{\PName}{n}{\decl}}\!(f) \\ 
\:=\: \inf\nolimits_n G^n(\top_\LSEnv) (f) \:=\: \gfp{\sqsubseteq}{G} (f)~.
\end{multline*}

\begin{lemma}
\label{thm:diag-sup-cpo}
\textnormal{\cite[p.~127]{Winskel:1993}}
 Suppose $a_{n,m}$ are elements of upper $\omega$-cpo $(A,\leq)$
  with the property that $a_{n,m} \leq a_{n',m'}$ whenever $n \leq n'$ and $m
  \leq m'$. Then,
\[
\sup\nolimits_n \, (\sup\nolimits_m a_{n,m}) \:=\: \sup\nolimits_m \,
(\sup\nolimits_n a_{n,m}) \:=\: \sup\nolimits_i \,a_{i,i}~.
\]
\end{lemma}

\begin{lemma}[Monotone Sequence Theorem]
\label{thm:MST}
If $\langle a_n \rangle$ is a monotonic increasing sequence in a closed interval
$[L,\, U] \subseteq [-\infty, \allowbreak +\infty]$, then the supremum $\sup_n a_n$
coincides with $\lim_{n \To \infty} a_n$. Dually, if $\langle a_n \rangle$ is a
monotonic decreasing sequence in a closed interval
$[L,\, U] \subseteq [-\infty, +\infty]$, the infimum $\inf_n a_n$ coincides with
$\lim_{n \To \infty} a_n$.
\end{lemma}

\subsection{Soundness of \wllpsymbol Rules}
\label{sec:app-om-rule-sound}

\begin{fact}
\label{fact:deriv-elim}
To carry on the proofs we use the fact that from 
\[
\deriv{\wllp{\Call{\PName}}\!(f_1) \bowtie g_1 }{\wllp{c}\!(f_2) \bowtie g_2}~,
\]
it follows that for all environment $\decl^\star$,
\begin{equation*}%\label{eq:deriv-elim}
\wllpd{\Call{\PName}}{\decl^\star}\!(f_1) \bowtie g_1   \implies  
\wllpd{c}{\decl^\star}\!(f_2) \bowtie g_2~.
\end{equation*}
\end{fact}

We provide detailed proofs for rules \lrule{wp-rec} and \lrule{wp-rec$_\omega$};
the proof of rules \lrule{wlp-rec} and \lrule{wlp-rec$_\omega$} follows a dual
argument.

\medskip
\noindent \textbf{Soundness of rule \lrule{wp-rec}.} 
Since by definition,
  $\wpd{\Call{\PName}}{\decl}\!(f) \allowbreak = \allowbreak \sup_{n}
  \wpd{\Calln{\PName}{n}{\decl}}{\decl}\!(f)$,
  to establish the conclusion of the rule it suffices to show that
\[
\forall n \mydot \wp{\Calln{\PName}{n}{\decl}}\!(f) \preceq g~,
\]
which we do by induction on $n$. The base case is immediate since
$\Calln{\PName}{0}{\decl} = \Abort$ and $\wp{\Abort}(f) = \CteFun{0}$. For the
inductive case, we reason as follows:

% we rely on a subsidiary fact which says that for any command
% $c$ and closed command $c'$,
% %
% \begin{equation}\label{eq:wp-subst}
% \wp{c \subst{\Call{\PName}}{c'}} = \wpd{c}{\PName \triangleright c'}~.
% \end{equation}
% %
% The property is rather intuitive and its proof proceeds by a routine induction
% on the structure of $c$; see \autoref{sec:subst} for a proof
% sketch. Now, we turn to the inductive case.
%
\belowdisplayskip=0pt
\[
\!\!\!\begin{array}{c@{\:\:} l@{} r}
& \wp{\Calln{\PName}{n+1}{\decl}}\!(f) \preceq g & \by{def. $n$-inl.}\displaybreak[0]\\[2pt]
\Leftrightarrow & \wp{\decl(\PName)
 \subst{\Call{\PName}}{\Calln{\PName}{n}{\decl}}}\!(f) \preceq g &
                                                       \by{\autoref{thm:subst-env}
                                                                   }\displaybreak[0]\\[2pt]
\Leftrightarrow & \wpd{\decl(\PName)}
  {\PName \, \triangleright \, \Calln{\PName}{n}{\decl}}
  \!(f) \preceq g & \by{rule prem, \autoref{fact:deriv-elim}}\displaybreak[0]\\[2pt]
 \Leftarrow & \wpd{\Call{\PName}}
   {\PName \, \triangleright \, \Calln{\PName}{n}{\decl}}
   \!(f) \preceq g & \by{\autoref{thm:subst-env}}\displaybreak[0]\\[2pt]
 \Leftrightarrow & \wp{ \Call{\PName}
  \subst{\Call{\PName}}{\Calln{\PName}{n}{\decl}}}\!(f) \preceq g & \by{def.~subst.}\displaybreak[0]\\[2pt]
\Leftrightarrow & \wp{\Calln{\PName}{n}{\decl}}\!(f) \preceq g &
                                                                  \by{I.H.} 
\end{array}
\]

\medskip
\noindent \textbf{Soundness of rule \lrule{wp-rec$_\omega$}.}  We prove that the
rule's premises entail
$l_n \preceq \wp{\Calln{\PName}{n}{\decl}}(f) \preceq u_n$ for all
$n \in \Nats$. The conclusion of the rule then follows immediately by taking the
supremum over $n$ on the three sides of the equation. We proceed by induction on
$n$. The base case is trivial since by definition,
$\wp{\Calln{\PName}{0}{\decl}}(f) = \wp{\Abort}(f) = \CteFun{0}$ and by the
rule's premise, $l_0 = u_0 = \CteFun{0}$. For the inductive case we reason as
follows:
\begin{align*}
\begin{array}{c@{\:\:} l@{} }
& l_{n+1} \preceq \wp{\Calln{\PName}{n+1}{\decl}}\!(f) \preceq u_{n+1} \displaybreak[0]\\[2pt]
\Leftrightarrow & \qquad \by{def. $n$-inl.}\displaybreak[0]\\[2pt]
& l_{n+1} \preceq  \wp{\decl(\PName)
 \subst{\Call{\PName}}{\Calln{\PName}{n}{\decl}}}\!(f) \preceq u_{n+1}\displaybreak[0]\\[2pt]
\Leftrightarrow & \qquad  \by{\autoref{thm:subst-env}}\displaybreak[0]\\[2pt]
& l_{n+1} \preceq \wpd{\decl(\PName)}
  {\PName \, \triangleright \, \Calln{\PName}{n}{\decl}}
  \!(f) \preceq u_{n+1}\displaybreak[0]\\[2pt]
\Leftrightarrow & \qquad \by{rule prem, \autoref{fact:deriv-elim}}\displaybreak[0]\\[2pt]
& l_{n} \preceq \wpd{\Call{\PName}}
   {\PName \, \triangleright \, \Calln{\PName}{n}{\decl}}
   \!(f) \preceq u_{n}\displaybreak[0]\\[2pt]
 \Leftarrow & \qquad \by{\autoref{thm:subst-env}}\displaybreak[0]\\[2pt]
& l_{n} \preceq  \wp{ \Call{\PName}
  \subst{\Call{\PName}}{\Calln{\PName}{n}{\decl}}}\!(f) \preceq u_{n}\displaybreak[0]\\[2pt]
 \Leftrightarrow & \qquad \by{def.~subst.}\displaybreak[0]\\[2pt]
& l_{n} \preceq  \wp{\Calln{\PName}{n}{\decl}}\!(f) \preceq u_{n}\displaybreak[0]\\[2pt]
\Leftrightarrow & \qquad \by{I.H.}\displaybreak[0]\\[2pt]
& \true
\end{array} \\[-\normalbaselineskip]\tag*{\qedhere}
\end{align*}

\subsection{Substitution of Procedure Calls}
\label{sec:subst}

\begin{figure}[h]
\scalebox{0.87}{
$
\begin{array}{ll}
\specialrule{0.8pt}{0pt}{2pt}
\boldsymbol{c} & \boldsymbol{c \subst{\Call{\PName}}{c'}}\\
\specialrule{0.8pt}{2pt}{6pt}
%
% No operation
\Skip   &  \Skip \\[1.5pt]
% Assignment
\Ass{x}{E}  & \Ass{x}{E}  \\[1.5pt]
% Abortion
\Abort & \Abort \\[1.5pt]
% Procedure call
\Call{\PName} & c' \\[3pt]
% Guarded command 
\Cond{G}{c_1}{c_2}  & \Cond{G}{c_1 \subst{\Call{\PName}}{c'}}{c_2 \subst{\Call{\PName}}{c'}} \\[3pt]
% Probabilistic choice
\PChoice{c_1}{p}{c_2} &  \PChoice{c_1 \subst{\Call{\PName}}{c'}}{p}{c_2 \subst{\Call{\PName}}{c'}} \\[3pt]
% Sequential composition 
c_1;c_2 & c_1\subst{\Call{\PName}}{c'};\; c_2\subst{\Call{\PName}}{c'} 
\end{array} 
$}
\caption{Syntactic replacement of procedure calls.}
\label{fig:command-subst}
\end{figure}

\begin{lemma} 
\label{thm:subst-env}
  For every command $c$ and closed command $c'$,
\[
\wp{c \subst{\Call{\PName}}{c'}} \:=\: \wpd{c}{\PName \triangleright c'}~.
\]
\end{lemma}
\begin{proof}
  By induction on the structure of $c$. Except for procedure calls, the proof
  for all other program constructs follows from de definition of
  $\wpsymbol$ and some simple calculations (and the inductive hypotheses in the case
  of compound instructions). For the case of procedure calls, the proof relies
  on the fact that as $c'$ is a closed command,
  $\Calln{\PName}{n}{\PName \, \triangleright \, c'} = c'$ for all $n \geq
  1$. Concretely, we reason as follows:
\begin{align*}
\begin{array}{c@{\:\:} l@{} }
& \wp{\Call{\PName}\, \subst{\Call{\PName}}{c'}}\!(f) \displaybreak[0]\\[2pt]
= & \qquad \by{def.~subst.}\displaybreak[0]\\[2pt]
& \wp{c'}\!(f) \displaybreak[0]\\[2pt]
= & \qquad \by{sup.\ of a constant sequence}\displaybreak[0]\\[2pt]
& \sup_{n} \wp{c'}\!(f)  \displaybreak[0]\\[2pt]
= & \qquad \by{observation above}\displaybreak[0]\\[2pt]
& \sup_{n} \wp{\Calln{\PName}{n+1}{\PName \, \triangleright \, c'}} \!(f)  \\[6pt]
= & \qquad \by{$\wp{\Calln{\PName}{0}{\PName \, \triangleright \, c'}} \!(f) = \CteFun{0}$}\\[6pt]
& \sup_{n} \wp{\Calln{\PName}{n}{\PName \, \triangleright \, c'}}  \!(f)  \\[6pt]
= & \qquad \by{def.~\wp{\cdot}}\displaybreak[0]\\[2pt]
& \wpd{\Call{\PName}}{\PName \triangleright c'}
\end{array} \\[-\normalbaselineskip]\tag*{\qedhere}
\end{align*}

\end{proof}

\subsection{Continuity of Transformer \ewllp{\cdot}{\theta}}

\begin{figure}[h]
\scalebox{0.9}{
$
\begin{array}{ll}
\specialrule{0.8pt}{0pt}{2pt}
\boldsymbol{c} & \boldsymbol{\ewp{c}{\theta}(f)}\\
\specialrule{0.8pt}{2pt}{6pt}
%
% No operation
\Skip   & f \\[1.5pt]
% Assignment
\Ass{x}{E}  & f\!\subst{x}{E} \\[1.5pt]
% Abortion
\Abort & \CteFun{0} \\[1.5pt]
% Guarded command 
\Cond{G}{c_1}{c_2}  & 
    \ToExp{G} \cdot \ewp{c_1}{\theta} (f) + \ToExp{\lnot G} \cdot \ewp{c_2}{\theta}
    (f) \\[3pt]
% Probabilistic choice
\PChoice{c_1}{p}{c_2} &   p \cdot \ewp{c_1}{\theta} (f)  + (1{-}p) \cdot
\ewp{c_2}{\theta} (f) \\[3pt]
% Procedure call
\Call{\PName} &  \theta(f) \\[3pt]
% Sequential composition 
c_1;c_2 & \ewp{c_1}{\theta} \bigl(\ewp{c_2}{\theta} (f)\bigr)
\\[10pt]
\specialrule{0.8pt}{0pt}{2pt}
\boldsymbol{c} & \boldsymbol{\ewlp{c}{\theta}(f)}\\
\specialrule{0.8pt}{2pt}{6pt}
% Abortion
\Abort & \CteFun{1} \\[1.5pt]
\end{array} 
$}
\caption{Expectation transformer \ewllp{\cdot}{\theta}. Transformer
  $\ewlp{\cdot}{\theta}$ differs from $\ewp{\cdot}{\theta}$ only in $\Abort$
  instructions.}
\label{fig:ewp-sem}
\end{figure}

As a preliminary step to discuss the continuity of \ewllp{\cdot}{(\cdot)} we
observe that order relation ``$\sqsubseteq$'' (see paragraph below
\autoref{thm:fp-rec}) endows the set of environments \SEnv with the structure of
an upper $\omega$--cpo with botom element
$\bot_\SEnv = \lambda f \colon \UEX \mydot \CteFun{0}$, where the supremum of an
increasing $\omega$--chain \chain{\theta}{\sqsubseteq} is given pointwise, \ie
$(\sup_n \theta_i) (f) = \sup_n \theta_i(f)$. Likewise, ``$\sqsubseteq$' endows
the set of liberal environments \LSEnv with the structure of a lower
$\omega$--cpo with top element
$\top_\LSEnv = \lambda f \colon \BEX \mydot \CteFun{1}$, where the infimum of a
decreasing $\omega$--chain \chain{\theta}{\sqsupseteq} is given pointwise, \ie
$(\inf_n \theta_i) (f) = \inf_n \theta_i(f)$.

We will discuss two kind of continuity results for
\ewllp{\cdot}{(\cdot)}. First, we show that for every environment $\theta$,
expectation transformer $\ewllp{\cdot}{\theta}$ is continuous, or equivalently,
that
\begin{align*}
  \ewp{c}{(\cdot)} &\colon (\SEnv, \sqsubseteq) \To
  (\SEnv, \sqsubseteq) \\
  \ewlp{c}{(\cdot)} &\colon (\LSEnv, \sqsubseteq) \To
  (\LSEnv, \sqsubseteq)
\end{align*}
This result will be established in \autoref{thm:ewp-cont-exp}. Second, we show
that the above environment transformers are themselves continuous, \ie that
\begin{align*}
  \ewp{c}{(\cdot)} &\colon \ucont{(\SEnv, \sqsubseteq)}{(\SEnv, \sqsubseteq)}\\
  \ewlp{c}{(\cdot)} &\colon \lcont{(\LSEnv, \sqsubseteq)}{(\LSEnv, \sqsubseteq)}
\end{align*}
This result will be established in \autoref{thm:ewp-cont-env}.

\begin{lemma}
\label{thm:ewp-cont-exp}
  Let $\theta  \in \SEnv$ and \chain{f}{\preceq} be an ascending $\omega$--chain
  of expectations in \UEX. Then for every command $c$,
  \[
  \ewp{c}{\theta} (\sup\nolimits_n f_n) \:=\: \sup\nolimits_n \,
  \ewp{c}{\theta} \!(f_n)~.
  \]
  Analogously, if \chain{f}{\succeq} is a descending $\omega$--chain
  of expectations in \BEX,
  \[
  \ewlp{c}{\theta} (\inf\nolimits_n f_n) \:=\: \inf\nolimits_n \,
  \ewlp{c}{\theta} \!(f_n)~.
  \]
\end{lemma}
\begin{proof}
By induction on the structure of $c$. Except for
procedure calls, all program constructs use the same proof argument as for the
continuity of plain transformer $\wllp{\cdot}$, which has already been dealt with in \eg
\cite{DBLP:journals/pe/GretzKM14}. For procedure calls we reason as follows.
\[
\begin{array}{c@{\:\:} l@{} }
& \ewp{\Call{\PName}}{\theta} (\sup\nolimits_n  f_n) \\[3pt]
= & \qquad \by{def.~\ewp{\cdot}{\theta}}\displaybreak[0]\\[2pt]
& \theta \, (\sup\nolimits_n  f_n)\displaybreak[0]\\[2pt]
= & \qquad \by{$\theta$ is continuous by hypothesis}\displaybreak[0]\\[2pt]
& \sup\nolimits_n \, \theta (f_n) \displaybreak[0]\\[2pt]
= & \qquad \by{def.~\ewp{\cdot}{\theta}}\displaybreak[0]\\[2pt]
& \sup\nolimits_n \, \ewp{\Call{\PName}}{\theta} \!(f_n)~.
\end{array}
\]
The reasoning to show that
\begin{align*}
	\ewlp{\Call{\PName}}{\theta} (\inf\nolimits_n  f_n) ~=~ \inf\nolimits_n \ewlp{\Call{\PName}}{\theta} (f_n)
\end{align*}
is analogous. 
\end{proof}

\begin{lemma}
\label{thm:ewp-cont-env}
  Let \chain{\theta}{\sqsubseteq} be an ascending $\omega$--chain in \SEnv. Then
  for every command $c$,
  \[
  \ewp{c}{\sup\nolimits_n \theta_n} \:=\: \sup\nolimits_n
  \ewp{c}{\theta_n}~.
  \]
  Analogously, if \chain{\theta}{\sqsupseteq} is a descending $\omega$--chain
  in \LSEnv,
  \[
   \ewlp{c}{\inf\nolimits_n \theta_n} \:=\: \inf\nolimits_n
  \ewlp{c}{\theta_n}~.
  \]
  
\end{lemma}
\begin{proof}
  By induction on the structure of $c$. We consider only the case of
  $\ewp{c}{\theta}$; the case of $\ewlp{c}{\theta}$ is analogous. For the three
  basic instructions $c=\Skip$, $c=\Ass{x}{E}$ and $c=\Abort$ the proof is
  straightforward since the action of transformer $\ewp{\cdot}{(\cdot)}$ on
  these instructions is independent of the semantic environment at stake (\ie
  constant functions are always continuous). For the remaining program
  constructs we reason as follows:

\medskip
\noindent \emph{Procedure Call:}
\[
\begin{array}{c@{\:\:} l@{} }
& \ewp{\Call{\PName}}{\sup\nolimits_n \theta_n} \!(f) \\[3pt]
= & \qquad \by{def.~\ewp{\cdot}{\theta}}\displaybreak[0]\\[2pt]
& (\sup\nolimits_n \theta_n) (f)\displaybreak[0]\\[2pt]
= & \qquad \by{def.~$\sup\nolimits_n \theta_n$}\displaybreak[0]\\[2pt]
& \sup\nolimits_n \theta_n(f) \displaybreak[0]\\[2pt]
= & \qquad \by{def.~\ewp{\cdot}{\theta}}\displaybreak[0]\\[2pt]
& \sup\nolimits_n \ewp{\Call{\PName}}{\theta_n} \!(f)~.
\end{array}
\]

\medskip
\noindent \emph{Sequential Composition:}
\[
\begin{array}{c@{\:\:} l@{} }
& \ewp{c_1;\, c_2}{\sup\nolimits_n \theta_n} \!(f)\displaybreak[0]\\[2pt]
= & \qquad \by{def.~\ewp{\cdot}{\theta}}\displaybreak[0]\\[2pt]
& \ewp{c_1}{\sup\nolimits_m \theta_m} \!\bigl(\ewp{c_2}{\sup\nolimits_n \theta_n}\!(f)
   \bigr)\displaybreak[0]\\[2pt]
= & \qquad \by{I.H.~on $c_2$}\displaybreak[0]\\[2pt]
& \ewp{c_1}{\sup\nolimits_m \theta_m} \!\bigl(\sup\nolimits_n \ewp{c_2}{\theta_n}\!(f)
   \bigr)\displaybreak[0]\\[2pt]
= & \qquad \by{\autoref{thm:ewp-cont-exp}}\displaybreak[0]\\[2pt]
& \sup\nolimits_n \ewp{c_1}{\sup\nolimits_m \theta_m} \!\bigl(\ewp{c_2}{\theta_n}\!(f)
   \bigr)\displaybreak[0]\\[2pt]
= & \qquad \by{I.H.~on $c_1$}\displaybreak[0]\\[2pt]
& \sup\nolimits_n \sup\nolimits_m \ewp{c_1}{\theta_m} \!\bigl(\ewp{c_2}{\theta_n}\!(f)
   \bigr)\displaybreak[0]\\[2pt]
\overset{\star}{=} & \qquad \by{\autoref{thm:diag-sup-cpo}}\displaybreak[0]\\[2pt]
& \sup\nolimits_i \ewp{c_1}{\theta_i} \!\bigl(\ewp{c_2}{\theta_i}\!(f)
   \bigr)\displaybreak[0]\\[2pt]
= & \qquad \by{def.~\ewp{\cdot}{\theta}}\displaybreak[0]\\[2pt]
& \sup\nolimits_i \ewp{c_1;\, c_2}{\theta_i} \!(f)
\end{array}
\]
For applying \autoref{thm:diag-sup-cpo} in step (*) we have to show that 
\[
\ewp{c_1}{\theta_m} \!\bigl(\ewp{c_2}{\theta_n}\!(f)
   \bigr) \preceq \ewp{c_1}{\theta_{m'}} \!\bigl(\ewp{c_2}{\theta_{n'}}\!(f)
   \bigr)
\] 
whenever $n \leq n'$  and $m \leq m'$. To this end, we use a transitivity
argument and show that
\begin{align}
\ewp{c_1}{\theta_m} \!\bigl(\ewp{c_2}{\theta_n}\!(f)
   \bigr) \preceq \ewp{c_1}{\theta_{m}} \!\bigl(\ewp{c_2}{\theta_{n'}}\!(f)
   \bigr) \label{eq:1}\\
\ewp{c_1}{\theta_m} \!\bigl(\ewp{c_2}{\theta_{n'}}\!(f)
   \bigr) \preceq \ewp{c_1}{\theta_{m'}} \!\bigl(\ewp{c_2}{\theta_{n'}}\!(f)
   \bigr) \label{eq:2}
\end{align}
To prove Equation~\eqref{eq:1} we first apply the I.H.\ on $c_2$. Since continuity
entails monotonicity, we obtain $\ewp{c_2}{\theta_n} \sqsubseteq 
\ewp{c_2}{\theta_{n'}}$, which itself gives $\ewp{c_2}{\theta_n}\!(f) \preceq
\ewp{c_2}{\theta_{n'}}\!(f)$. We are left to show that
$\ewp{c_1}{\theta_m}(\cdot)$ is monotonic, which follows by its continuity
guaranteed by \autoref{thm:ewp-cont-exp}. To prove Equation~\eqref{eq:2}, we
apply the I.H.\ on $c_1$.  Again, since the continuity of $\ewp{c_1}{(\cdot)}$ implies
its monotonicity, we obtain $\ewp{c_1}{\theta_m} \sqsubseteq \ewp{c_1}{\theta_{m'}}$,
which establishes Equation~\eqref{eq:2}.

\medskip
\noindent \emph{Conditional Branching:}
\[
\begin{array}{c@{\:\:} l@{} }
& \ewp{\Call{\Cond{G}{c_1}{c_2}}}{\sup\nolimits_n \theta_n} \!(f) \\[3pt]
= & \qquad \by{def.~\ewp{\cdot}{\theta}}\displaybreak[0]\\[2pt]
& 
\ToExp{G} \cdot \ewp{c_1}{\sup\nolimits_n \theta_n} (f) + \ToExp{\lnot G} \cdot
  \ewp{c_2}{\sup\nolimits_n \theta_n}(f) \displaybreak[0]\\[2pt]
= & \qquad \by{I.H.~on $c_1$,$c_2$}\displaybreak[0]\\[2pt]
& 
\ToExp{G} \cdot \sup\nolimits_n \ewp{c_1}{\theta_n} (f) + \ToExp{\lnot G} \cdot
  \sup\nolimits_n \ewp{c_2}{\theta_n}(f) \displaybreak[0]\\[2pt]
\overset{(*)}{=} & \qquad \by{\autoref{thm:MST}}\displaybreak[0]\\[2pt]
& \ToExp{G} \cdot \lim\limits_{n \To \infty} \ewp{c_1}{\theta_n} (f) + \ToExp{\lnot G} \cdot
  \lim\limits_{n \To \infty} \ewp{c_2}{\theta_n}(f) \\[5pt]
= & \qquad \by{algebra of limits}\displaybreak[0]\\[2pt]
& \lim\limits_{n \To \infty} \bigl( \ToExp{G} \cdot \ewp{c_1}{\theta_n} (f) +
  \ToExp{\lnot G} \cdot \ewp{c_2}{\theta_n}(f) \bigr) \displaybreak[0]\\[2pt]
\overset{(**)}{=} & \qquad \by{\autoref{thm:MST}}\displaybreak[0]\\[2pt]
& \sup_n \bigl( \ToExp{G} \cdot \ewp{c_1}{\theta_n} (f) +
  \ToExp{\lnot G} \cdot \ewp{c_2}{\theta_n}(f) \bigr) \displaybreak[0]\\[2pt]
= & \qquad \by{def.~\ewp{\cdot}{\theta}}\displaybreak[0]\\[2pt]
& \sup\nolimits_n \ewp{\Call{\Cond{G}{c_1}{c_2}}}{\theta_n}
\end{array}
\]
To apply \autoref{thm:MST} in steps (*) and (**) we have to show that sequences $\langle
\ewp{c_1}{\theta_n} (f) \rangle$ and $\langle
\ewp{c_2}{\theta_n} (f) \rangle$ are increasing. This follows by I.H. on $c_1$
and $c_2$ since continuity entails monotonicity.

\medskip
\noindent \emph{Probabilistic Choice:} follows the same argument as conditional
branching. 

\end{proof}

\subsection{Basic Properties of Transformer \boldeetsymbol}
\label{sec:eet-basic}

We begin by presenting some preliminary results that will be necessary for
establishing the main results about the \eetsymbol\ transformer. 

\begin{fact}[$(\RtEnv,\sqsubseteq)$ is an $\omega$--cpo]
\label{fact:RtEnv-cpo}
Let ``$\sqsubseteq$" denotes the pointwise order between runtime environments,
\ie for $\eta_1, \eta_2 \in \RtEnv$, $\eta_1 \sqsubseteq \eta_2$ iff
$\eta_1(\rt) \preceq \eta_2(\rt)$ for every $\rt \in \Runtimes$. Relation
``$\sqsubseteq$'' endows the set of runtime environments \RtEnv with the
structure of an upper $\omega$--cpo with botom element
$\bot_\RtEnv = \lambda \rt \colon \Runtimes \mydot \CteFun{0}$, where the
supremum of an increasing $\omega$--chain \chain{\eta}{\sqsubseteq} is given
pointwise, \ie $(\sup_n \eta_i) (\rt) = \sup_n \eta_i(\rt)$.
\end{fact}

\begin{lemma}[Continuity of $\eeet{\cdot}{\eta}$ \wrt $\eta$]
\label{thm:eeet-cont-env}
  Let \chain{\eta}{\sqsubseteq} be an ascending $\omega$--chain in \RtEnv. Then
  for every command $c$,
  \[
  \eeet{c}{\sup\nolimits_n \eta_n} \:=\: \sup\nolimits_n
  \eeet{c}{\eta_n}~.
  \]
\end{lemma}
\begin{proof}
  The proof follows the same argument as that for
    establishing the continuity of transformer \wpsymbol (see
    \autoref{thm:ewp-cont-env}).
\end{proof}

\begin{lemma}[$\eeet{c}{(\cdot)}$ preserves continuity]
\label{thm:eeet-cont-pres}
For every command $c$ and every (upper continuous) runtime environment
$\eta \in \RtEnv$, $\eeet{c}{\eta}$ is a continuous runtime transformer in
$\ucont{\Runtimes}{\Runtimes}$.
\end{lemma}
\begin{proof}
  By induction on the program structure. For every program constructs different
  from a procedure call, the reasoning is similar to that used in
  \autoref{thm:ewp-cont-exp} to prove the same property for transformer
  $\ewp{\cdot}{}$. For a procedure call the statement follows immediately since
  $\eta$ is continuous by hypothesis.
\end{proof}

\begin{lemma}[Alternative characterization of $\eetd{\Call{\PName}}{\decl}$]
\label{thm:eet-proc-call-iter}
Let $F(\eta) = \ctertenv{1} \oplus \eeet{\decl(\PName)}{\eta}$. Then
\[
\eetd{\Call{\PName}}{\decl} \:=\: \sup\nolimits_n F^n(\bot_\RtEnv)~,
\]
where $\bot_\RtEnv = \lambda \rt \colon \Runtimes \mydot \CteFun{0}$ and
$F^n(\bot_\RtEnv)$ denotes the repeated
application of $F$ from $\bot_\RtEnv$ $n$ times (\ie
$F^0(\bot_\RtEnv) = \mathit{id}$ and
$F^{n+1}(\bot_\RtEnv) = F (F^n(\bot_\RtEnv))$).
\end{lemma}
\begin{proof}
  Using \autoref{thm:eeet-cont-env} one can show that $F$ is an (upper)
  continuous runtime transformer. The result then follows from a direct
  application of Kleene's Fixed Point Theorem and \autoref{fact:RtEnv-cpo}.
\end{proof}

To present the following lemma we use the notion of \emph{expanding}
runtime environments. Given $\eta_0, \eta_1 \in \RtEnv$, $\theta \in \SEnv$ and
$k, \Delta \in \PosReals$ we say that $\langle \eta_1,\eta_0,\theta\rangle$ are
$\langle k, \Delta \rangle$--\emph{expanding} iff 
\[
\rt_1 - \rt_0 \succeq k \cdot (\CteFun{1}
{-} f) + \CteFun{\Delta}
\]
implies
\[
\eta_1(\rt_1) - \eta_0(\rt_0) \succeq  k \cdot \bigl(\CteFun{1}
- \theta(f)\bigr) + \CteFun{\Delta}
\]
for all $\rt_0,\rt_1 \in \Runtimes$ and $f \in \BEX$.

\begin{lemma}\label{thm:expanding-env}
 Let $\langle \eta_1,\eta_0,\theta\rangle$ be
$\langle k, \Delta \rangle$--expanding environments\footnote{See paragraph above.} and $c$ be an
$\Abort$--free command. Then
\[
\rt_1 - \rt_0 \succeq k \cdot (\CteFun{1}
{-} f) + \CteFun{\Delta}
\]
implies
\[
\eeet{c}{\eta_1} \!(\rt_1) - \eeet{c}{\eta_0} \!(\rt_0) \succeq  k \cdot \bigl(\CteFun{1}
- \ewp{c}{\theta} \!(f) \bigr) + \CteFun{\Delta}
\]
for all $\rt_0,\rt_1 \in \Runtimes$ and $f \in \BEX$.
\end{lemma}
\
\begin{proof}
By induction on the structure of $c$.

\medskip
\noindent \emph{No--op:}
\[
\begin{array}{c@{\:\:} l@{}}
& \eeet{\Skip}{\eta_1}\!(\rt_1) - \eeet{\Skip}{\eta_0}\!(\rt_0)\displaybreak[0]\\[2pt]
& \succeq k \cdot \bigl(\CteFun{1}
- \ewp{\Skip}{\theta} \!(f) \bigr) + \CteFun{\Delta} \displaybreak[0]\\[2pt]
\Leftrightarrow & \qquad \by{def.~of $\eeet{\cdot}{\eta}$, $\ewp{\cdot}{\theta}$}\displaybreak[0]\\[2pt]
& (\CteFun{1} + \rt_1) - (\CteFun{1} + \rt_0) \succeq k \cdot (\CteFun{1}
{-} f) + \CteFun{\Delta} \displaybreak[0]\\[2pt]
\Leftarrow & \qquad \by{hypothesis}\displaybreak[0]\\[2pt]
& \true
\end{array}
\]

\medskip
\noindent \emph{Assignment:}
\[
\begin{array}{c@{\:\:} l@{} }
& \eeet{\Ass{x}{E}}{\eta_1}\!(\rt_1) - \eeet{\Ass{x}{E}}{\eta_0}\!(\rt_0)\displaybreak[0]\\[2pt]
& \succeq k \cdot \bigl(\CteFun{1}
- \ewp{\Ass{x}{E}}{\theta} \!(f) \bigr) + \CteFun{\Delta} \displaybreak[0]\\[2pt]
\Leftrightarrow & \qquad \by{def.~of $\eeet{\cdot}{\eta}$, $\ewp{\cdot}{\theta}$}\displaybreak[0]\\[2pt]
& (\CteFun{1} + \rt_1)\subst{x}{E} - (\CteFun{1} + \rt_0)\subst{x}{E} \succeq k \cdot \bigl(\CteFun{1}
- f\subst{x}{E} \bigr) + \CteFun{\Delta} \displaybreak[0]\\[2pt]
\Leftrightarrow & \qquad \by{algebra}\displaybreak[0]\\[2pt]
& (\rt_1 - \rt_0)\subst{x}{E} \succeq \bigl( k \cdot (\CteFun{1}
{-} f) + \CteFun{\Delta}\bigr) \subst{x}{E} \displaybreak[0]\\[2pt]
\Leftarrow & \qquad \by{hypothesis}\displaybreak[0]\\[2pt]
& \true
\end{array}
\]

\medskip
\noindent \emph{Procedure Call:}
\[
\begin{array}{c@{\:\:} l@{} }
& \eeet{\Call{\PName}}{\eta_1}\!(\rt_1) - \eeet{\Call{\PName}}{\eta_0}\!(\rt_0)\displaybreak[0]\\[2pt]
& \succeq k \cdot \bigl(\CteFun{1}
- \ewp{\Call{\PName}}{\theta} \!(f) \bigr) + \CteFun{\Delta} \displaybreak[0]\\[2pt]
\Leftrightarrow & \qquad \by{def.~of $\eeet{\cdot}{\eta}$, $\ewp{\cdot}{\theta}$}\displaybreak[0]\\[2pt]
& \eta_1(\rt_1) - \eta_0(\rt_0) \succeq k \cdot \bigl(\CteFun{1}
- \theta (f) \bigr) + \CteFun{\Delta} \displaybreak[0]\\[2pt]
\Leftarrow & \qquad \by{$\langle \eta_1,\eta_0,\theta\rangle$ is
$\langle k, \CteFun{\Delta} \rangle$--\emph{expanding}}\displaybreak[0]\\[2pt]
& \rt_1 - \rt_0 \succeq k \cdot (\CteFun{1}
{-} f) + \CteFun{\Delta} \displaybreak[0]\\[2pt]
\Leftarrow & \qquad \by{hypothesis}\displaybreak[0]\\[2pt]
& \true
\end{array}
\]

\medskip
\noindent \emph{Probabilistic Choice:} 
\[
\begin{array}{c@{\:\:} l@{} }
& \eeet{\PChoice{c_1}{p}{c_2}}{\eta_1}\!(\rt_1) - \eeet{\PChoice{c_1}{p}{c_2}}{\eta_0}\!(\rt_0)\displaybreak[0]\\[2pt]
& \succeq k \cdot \bigl(\CteFun{1}
- \ewp{\PChoice{c_1}{p}{c_2}}{\theta} \!(f) \bigr) + \CteFun{\Delta} \displaybreak[0]\\[2pt]
\Leftrightarrow & \qquad \by{def.~of $\eeet{\cdot}{\eta}$, $\ewp{\cdot}{\theta}$}\displaybreak[0]\\[2pt]
& p \cdot  \bigl(  \eeet{c_1}{\eta_1}\!(\rt_1)  - \eeet{c_1}{\eta_0}\!(\rt_0)
  \bigr) \displaybreak[0]\\[2pt]
&+ \:(1{-}p) \cdot  \bigl(  \eeet{c_2}{\eta_1}\!(\rt_1)  - \eeet{c_2}{\eta_0}\!(\rt_0)
  \bigr) \displaybreak[0]\\[2pt]
& \succeq k \cdot \bigl(\CteFun{1}
- \bigl( p \cdot \ewp{c_1}{\theta}\!(f) + (1{-}p) \cdot \ewp{c_2}{\theta}\!(f)
  \bigr) \bigr) + \CteFun{\Delta} \displaybreak[0]\\[2pt]
\Leftarrow & \qquad \by{IH on $c_1$, $c_2$}\displaybreak[0]\\[2pt]
& p \cdot  \bigl( k \cdot \bigl(\CteFun{1}
- \ewp{c_1}{\theta} \!(f) \bigr) + \CteFun{\Delta}
  \bigr) \displaybreak[0]\\[2pt]
&+ \:(1{-}p) \cdot  \bigl(  k \cdot \bigl(\CteFun{1}
- \ewp{c_2}{\theta} \!(f) \bigr) + \CteFun{\Delta}
  \bigr) \displaybreak[0]\\[2pt]
& \succeq k \cdot \bigl(\CteFun{1}
- \bigl( p \cdot \ewp{c_1}{\theta}\!(f) + (1{-}p) \cdot \ewp{c_1}{\theta}\!(f)
  \bigr) \bigr) + \CteFun{\Delta} \displaybreak[0]\\[2pt]
\Leftarrow & \qquad \by{algebra (equality holds)}\displaybreak[0]\\[2pt]
& \true
\end{array}
\]

\medskip
\noindent \emph{Conditional Branching:} analogous to the case of probabilistic choice.

\medskip
\noindent \emph{Sequential Composition:}
\begin{align*}
\begin{array}{c@{\:\:} l@{} }
& \eeet{c_1;c_2}{\eta_1}\!(\rt_1) - \eeet{c_1;c_2}{\eta_0}\!(\rt_0)\displaybreak[0]\\[2pt]
& \succeq k \cdot \bigl(\CteFun{1}
- \ewp{c_1;c_2}{\theta} \!(f) \bigr) + \CteFun{\Delta} \displaybreak[0]\\[2pt]
\Leftrightarrow & \qquad \by{def.~of $\eeet{\cdot}{\eta}$, $\ewp{\cdot}{\theta}$}\displaybreak[0]\\[2pt]
& \eeet{c_1}{\eta_1}\!\bigl(\eeet{c_2}{\eta_1} \!(\rt_1)  -
  \eeet{c_1}{\eta_1}\!\bigl(\eeet{c_2}{\eta_0} \!(\rt_0)\displaybreak[0]\\[2pt]
& \succeq k \cdot \bigl(\CteFun{1}
- \ewp{c_1}{\theta} \bigl( \ewp{c_2}{\theta}\!(f) \bigr) \bigr) + \CteFun{\Delta} \displaybreak[0]\\[2pt]
\Leftarrow & \qquad \by{IH on $c_1$}\displaybreak[0]\\[2pt]
& \eeet{c_2}{\eta_1} \!(\rt_1)  -
  \eeet{c_2}{\eta_0} \!(\rt_0)  \succeq k \cdot \bigl(\CteFun{1}
- \ewp{c_2}{\theta}\!(f) \bigr) + \CteFun{\Delta} \displaybreak[0]\\[2pt]
\Leftarrow & \qquad \by{IH on $c_2$}\displaybreak[0]\\[2pt]
& \rt_1 - \rt_0 \succeq k \cdot (\CteFun{1}
{-} f) + \CteFun{\Delta} \displaybreak[0]\\[2pt]
\Leftarrow & \qquad \by{hypothesis}\displaybreak[0]\\[2pt]
& \true
\end{array}\\[-\normalbaselineskip]\tag*{\qedhere}
\end{align*}

\end{proof}

\begin{lemma} \label{thm:ert-call-div}
 Let \PName be an $\Abort$--free procedure with declaration $\decl$. Then for
 every runtime $\rt$,
\[
\eet{\Call{\PName}}\!(\rt) \: \succeq\: \sup\nolimits_{n} \: n \,{\cdot}\, \bigl(\CteFun{1}
- \wpd{\Call{\PName}}{\decl}\!(\CteFun{1}) \bigr)~.
\]
\end{lemma}
\begin{proof}
  Let $F(\eta) = \ctertenv{1} \oplus \eeet{\decl(\PName)}{\eta}$.  Since by
  \autoref{thm:eet-proc-call-iter},
  $\eetd{\Call{\PName}}{\decl} = \sup\nolimits_n F^n(\bot)$, the result follows
  from showing that for all $n\geq 0$,
  \[ 
   F^{n+1}(\bot) (\rt) \succeq (n+1) \,{\cdot}\, \bigl(\CteFun{1}
- \wpd{\Call{\PName}}{\decl}\!(\CteFun{1}) \bigr)~.
\]  
To establish this, we first prove by induction on $i$ that whenever $\rt_1 -
\rt_0 \succeq 0$,
\[
F^{i+1}(\bot) (\rt_1) -
F^{i}(\bot) (\rt_0) \succeq \CteFun{1} - \wpd{\Call{\PName}}{\decl}\!(\CteFun{1})~,
\]
and then conclude using a telescopic sum argument as follows:
\begin{align*}
F^{n+1}(\bot) (\rt) \:=\;& F^{0}(\bot) (\rt) +  \sum\nolimits_{i=0}^{n} F^{i+1}(\bot) (\rt) -
  F^{i}(\bot) (\rt) \\
\succeq\:& \sum\nolimits_{i=0}^{n} F^{i+1}(\bot) (\rt) -
  F^{i}(\bot) (\rt) \\
\succeq\:& \: (n+1) \,{\cdot}\, \bigl(\CteFun{1}
- \wpd{\Call{\PName}}{\decl}\!(\CteFun{1}) \bigr)~.
\end{align*}
For the inductive proof we reason as follows. For the base case we have
\[
\begin{array}{c@{\:\:} l@{} }
& F^{1}(\bot) (\rt_1) - F^{0}(\bot)(\rt_0) \succeq \CteFun{1} - \wpd{\Call{\PName}}{\decl}\!(\CteFun{1}) \displaybreak[0]\\[2pt]
\Leftrightarrow & \qquad \by{def.~of $F^n$, $\bot$}\displaybreak[0]\\[2pt]
&  \CteFun{1} + \eeet{\decl(\PName)}{\bot}\!(\rt_1)  - \bot(t_0) \succeq \CteFun{1} - \wpd{\Call{\PName}}{\decl}\!(\CteFun{1}) \displaybreak[0]\\[2pt]
\Leftrightarrow & \qquad \by{def.~of $\bot$}\displaybreak[0]\\[2pt]
&  \CteFun{1} + \eeet{\decl(\PName)}{\bot}\!(\rt_1) \succeq \CteFun{1} - \wpd{\Call{\PName}}{\decl}\!(\CteFun{1}) \displaybreak[0]\\[2pt]
\Leftarrow & \qquad \by{$\wpd{\Call{\PName}}{\decl}\!(\CteFun{1}) \succeq
                  \CteFun{0}$}\displaybreak[0]\\[2pt]
& \true
\end{array}
\]
while for the inductive case we have, 
\begin{align*}
\begin{array}{c@{\:\:} l@{} }
& F^{i+2}(\bot) (\rt_1) - F^{i+1}(\bot)(\rt_0) \succeq \CteFun{1} - \wpd{\Call{\PName}}{\decl}\!(\CteFun{1}) \displaybreak[0]\\[2pt]
\Leftrightarrow & \qquad \by{def.~of $F^n$}\displaybreak[0]\\[2pt]
& \bigl( \CteFun{1} + \eeet{\decl(\PName)}{ F^{i+1}(\bot)}\!(\rt_1) \bigr) - 
\bigl( \CteFun{1} + \eeet{\decl(\PName)}{ F^{i}(\bot)}\! (\rt_0) \bigr) \displaybreak[0]\\[2pt]
& \succeq \CteFun{1} - \wpd{\Call{\PName}}{\decl}\!(\CteFun{1}) \displaybreak[0]\\[2pt]
\Leftrightarrow & \qquad \by{algebra}\displaybreak[0]\\[2pt]
&\eeet{\decl(\PName)}{ F^{i+1}(\bot)}\!(\rt_1) - 
 \eeet{\decl(\PName)}{ F^{i}(\bot)}\!(\rt_0)  \displaybreak[0]\\[2pt]
& \succeq 1 \cdot \bigl(\CteFun{1} - \wpd{\Call{\PName}}{\decl}\!(\CteFun{1}) \bigr) + 0 \displaybreak[0]\\[2pt]
\Leftrightarrow & \qquad \by{$ \wpd{\Call{\PName}}{\decl} =
                  \ewp{\decl(\PName)}{\wpd{\Call{\PName}}{\decl}}$ by \autoref{thm:fp-rec}}\displaybreak[0]\\[2pt]
&\eeet{\decl(\PName)}{ F^{i+1}(\bot)}\!(\rt_1) - 
 \eeet{\decl(\PName)}{ F^{i}(\bot)}\!(\rt_0)  \displaybreak[0]\\[2pt]
& \succeq 1 \cdot \bigl(\CteFun{1} - \ewp{\decl(\PName)}{\wpd{\Call{\PName}}{\decl}}\!(\CteFun{1}) \bigr) + 0 \displaybreak[0]\\[2pt]
\Leftarrow & \qquad \by{\autoref{thm:expanding-env}}\displaybreak[0]\\[2pt]
& \rt_1 - \rt_0 \succeq 1 \cdot (\CteFun{1} - \CteFun{1}) + \CteFun{0} \text{ and}\displaybreak[0]\\[2pt]
&\text{$\bigl\langle F^{i+1}(\bot),F^{i}(\bot),\wpd{\Call{\PName}}{\decl}\bigr\rangle$ are
$\langle 1, 0 \rangle$--expanding}\displaybreak[0]\\[2pt]
\Leftarrow & \qquad \by{hypothesis}\displaybreak[0]\\[2pt]
&\text{$\bigl\langle F^{i+1}(\bot),F^{i}(\bot),\wpd{\Call{\PName}}{\decl}\bigr\rangle$ are
$\langle 1, 0 \rangle$--expanding}\displaybreak[0]\\[2pt]
\Leftarrow & \qquad \by{IH}\displaybreak[0]\\[2pt]
& \true
\end{array}\\[-\normalbaselineskip]\tag*{\qedhere}
\end{align*}
\end{proof}

\begin{lemma}\label{thm:ert-of-constant}
For every constant $k \in \PosReals$ and $\Abort$--free program
$\prog{c}{\decl}$, 
\[
\eetd{c}{\decl} \!(\ctert{k}) \succeq \ctert{k}~.
\]
\end{lemma}
\begin{proof}
By induction on the structure of $c$. Except for the case of procedure calls,
all other program constructs pose no difficulty. For the case of a procedure
call, we make a case distinction on the termination behaviour of the procedure. 
If from state $s$ the procedure terminates almost surely, \ie 
$\wpd{\Call{\PName}}{\decl}(\ctert{1})(s) = 1$, the result follows from
\autoref{thm:eet-wp} and the linearity of $\wp{\cdot}$ (see
\autoref{thm:wp-basic-prop}) since
\begin{align*}
\MoveEqLeft[2] \eetd{\Call{\PName}}{\decl}\!(\ctert{k})(s)\\
&\:=\: \eetd{\Call{\PName}}{\decl}\!(\ctert{0}) (s) +  \wpd{\Call{\PName}}{\decl}\!(\ctert{k}) (s)\\
&\:\geq\: \wpd{\Call{\PName}}{\decl}\!(\ctert{k}) (s)\\
&\:=\: k \cdot \wpd{\Call{\PName}}{\decl}\!(1) (s) \:=\: k
\end{align*}
If, on the contrary, the procedure terminates with probability strictly less
than $1$ from state $s$, we conclude applying \autoref{thm:ert-call-div} since
\begin{align*}
\MoveEqLeft[2] \eetd{\Call{\PName}}{\decl}\!(\ctert{k})(s) \\
&\:\geq\: \sup\nolimits_{n} \: n \,{\cdot}\, \bigl(\underbrace{\CteFun{1}
- \wpd{\Call{\PName}}{\decl}\!(\CteFun{1})(s)}_{>0} \bigr)\\[-4pt]
&\:=\: \infty \:\geq\: k~.\qedhere
\end{align*}

\end{proof}

For stating the following lemma we use the notion of ``constant separable''
runtime environment. We say that $\eta \in \RtEnv$ is
\emph{constant separable into} $\upsilon \in \RtEnv$ iff for all $k \in
\PosReals$ and $\rt \in \Runtimes$, $\eta(\ctert{k}+t) = \ctert{k} + \upsilon(t)$. 

\begin{lemma}\label{thm:cte-sep-rtenv}
Let $\eta$ be a runtime environment constant separable\footnote{See paragraph
  above.} into $\upsilon$. Then for all command $c$,
\[
\eeet{c}{\eta}\!(\ctert{k}+t) \:=\: \ctert{k} + \eeet{c}{\upsilon}\!(t)~.
\]
\end{lemma}
\begin{proof}
By induction on the structure of $c$. 

\medskip
\noindent \emph{No--op:}
\[
\begin{array}{c@{\:\:} l@{}}
& \eeet{\Skip}{\eta}\!(\ctert{k} + \rt) \displaybreak[0]\\[2pt]
= & \qquad \by{def.~of $\eeet{\cdot}{\eta}$}\displaybreak[0]\\[2pt]
& \CteFun{1} + \ctert{k} + t \displaybreak[0]\\[2pt]
= & \qquad \by{def.~of $\eeet{\cdot}{\upsilon}$}\displaybreak[0]\\[2pt]
& \ctert{k} + \eeet{\Skip}{\upsilon}\!(\rt) \displaybreak[0]\\[2pt]
\end{array}
\]

\medskip
\noindent \emph{Assignment:}
\[
\begin{array}{c@{\:\:} l@{} }
& \eeet{\Ass{x}{E}}{\eta}\!(\ctert{k} + \rt) \displaybreak[0]\\[2pt]
= & \qquad \by{def.~of $\eeet{\cdot}{\eta}$}\displaybreak[0]\\[2pt]
& (\ctert{k} + \rt)\subst{x}{E} \displaybreak[0]\\[2pt]
= & \qquad \by{$\ctert{k}\subst{x}{E} = \ctert{k}$}\displaybreak[0]\\[2pt]
& \ctert{k} + \rt\subst{x}{E} \displaybreak[0]\\[2pt]
= & \qquad \by{def.~of $\eeet{\cdot}{\upsilon}$}\displaybreak[0]\\[2pt]
& \ctert{k} + \eeet{\Ass{x}{E}}{\upsilon}\!(\rt) 
\end{array}
\]

\medskip
\noindent \emph{Procedure Call:}
\[
\begin{array}{c@{\:\:} l@{} }
& \eeet{\Call{\PName}}{\eta_1}\!(\ctert{k} + \rt) \displaybreak[0]\\[2pt]
= & \qquad \by{def.~of $\eeet{\cdot}{\eta}$}\displaybreak[0]\\[2pt]
& \eta(\ctert{k} + \rt) \displaybreak[0]\\[2pt]
= & \qquad \by{$\eta$ constant separable into $\upsilon$}\displaybreak[0]\\[2pt]
& \ctert{k} + \upsilon(\rt) \displaybreak[0]\\[2pt]
= & \qquad \by{def.~of $\eeet{\cdot}{\upsilon}$}\displaybreak[0]\\[2pt]
& \ctert{k} + \eeet{\Call{\PName}}{\upsilon}\!(\rt) 
\end{array}
\]

\medskip
\noindent \emph{Probabilistic Choice:} 
\[
\begin{array}{c@{\:\:} l@{} }
& \eeet{\PChoice{c_1}{p}{c_2}}{\eta}\!(\ctert{k} + \rt) \displaybreak[0]\\[2pt]
= & \qquad \by{def.~of $\eeet{\cdot}{\eta}$}\displaybreak[0]\\[2pt]
& p \cdot  \eeet{c_1}{\eta}\!(\ctert{k} + \rt) + (1{-}p) \cdot
  \eeet{c_2}{\eta}\!(\ctert{k} + \rt) \displaybreak[0]\\[2pt]
= & \qquad \by{I.H.~on $c_1$, $c_2$}\displaybreak[0]\\[2pt]
& p \cdot  \bigl( \ctert{k} + \eeet{c_1}{\upsilon}\!(\rt) \bigr) + (1{-}p) \cdot
 \bigl(\ctert{k} + \eeet{c_2}{\upsilon}\!(\rt) \bigr) \displaybreak[0]\\[2pt] 
= & \qquad \by{algebra}\displaybreak[0]\\[2pt]
& \ctert{k} +  p \cdot  \eeet{c_1}{\upsilon}\!(\rt) + (1{-}p) \cdot
 \eeet{c_2}{\upsilon}\!(\rt) \displaybreak[0]\\[2pt] 
= & \qquad \by{def.~of $\eeet{\cdot}{\upsilon}$}\displaybreak[0]\\[2pt]
& \ctert{k} + \eeet{\PChoice{c_1}{p}{c_2}}{\upsilon}\!(\rt) 
\end{array}
\]

\medskip
\noindent \emph{Conditional Branching:} analogous to the case of probabilistic choice.

\medskip
\noindent \emph{Sequential Composition:}
\begin{align*}
\begin{array}{c@{\:\:} l@{} }
& \eeet{c_1;c_2}{\eta_1}\!(\ctert{k} + \rt) \displaybreak[0]\\[2pt]
= & \qquad \by{def.~of $\eeet{\cdot}{\eta}$}\displaybreak[0]\\[2pt]
& \eeet{c_1}{\eta}\!\bigl(\eeet{c_2}{\eta} \!(\ctert{k} + \rt)\bigr)  \displaybreak[0]\\[2pt]
= & \qquad \by{I.H.~on $c_2$}\displaybreak[0]\\[2pt]
& \eeet{c_1}{\eta}\!\bigl(\ctert{k} + \eeet{c_2}{\upsilon}\!(\rt) \bigr)
  \displaybreak[0]\\[2pt]
= & \qquad \by{I.H.~on $c_1$}\displaybreak[0]\\[2pt]
& \ctert{k} + \eeet{c_1}{\upsilon}\!\bigl(\eeet{c_2}{\upsilon}\!(\rt) \bigr)
  \displaybreak[0]\\[2pt]
= & \qquad \by{def.~of $\eeet{\cdot}{\upsilon}$}\displaybreak[0]\\[2pt]
& \ctert{k} + \eeet{c_1; c_2}{\upsilon}\!(\rt) 
\end{array}\\[-\normalbaselineskip]\tag*{\qedhere}
\end{align*}
  
\end{proof}

\begin{lemma}\label{thm:fix-point-compos}
  Let $(\mathcal{D}_1, \leq_1)$,  $(\mathcal{D}_2, \leq_2)$ and  $(\mathcal{D},
  \leq)$ be upper $\omega$--cpos with botom elements $\bot_1$, $\bot_2$ and $\bot$,
  respectively. Moreover let $F_1 \colon \mathcal{D}_1 \To \mathcal{D}_1$, $F_2
  \colon \mathcal{D}_2 \To \mathcal{D}_2$, $f_1 \colon \mathcal{D}_1 \To
  \mathcal{D}$, $f_2 \colon \mathcal{D}_2 \To \mathcal{D}$ be upper continuous
  and $h_1, h_2 \colon  \mathcal{D} \To \mathcal{D}$. If
  \begin{enumerate}
  \item \label{it-1} $\forall d_1 \mydot f_1(F_1(d_1)) \leq h_1(f_1(d_1))$ and  $\forall d_2
    \mydot f_2(F_2(d_2)) \leq h_2(f_2(d_2))$,
  \item \label{it-2} $f_1(\bot_1) \leq f_2 (\lfp{\!}{F_2})$ and $f_2(\bot_2) \leq f_1
    (\lfp{\!}{F_1})$, and
  \item \label{it-3} $h_1(f_2(\lfp{\!}{F_2})) \leq f_2(\lfp{\!}{F_2})$ and $h_2(f_1(\lfp{\!}{F_1})) \leq f_1(\lfp{\!}{F_1})$,
  \end{enumerate}
then
\[
f_1(\lfp{\!}{F_1}) \:=\: f_2(\lfp{\!}{F_2})~.
\]
\end{lemma}
\begin{proof}

\begin{align*}
\begin{array}{c@{\:\:} l@{} }
& f_1(\lfp{\!}{F_1}) = f_2(\lfp{\!}{F_2}) \displaybreak[0]\\[2pt]
\Leftrightarrow & \qquad \by{"$\leq$" is a partial order over $\mathcal{D}$}\displaybreak[0]\\[2pt]
& f_1(\lfp{\!}{F_1}) \leq f_2(\lfp{\!}{F_2})  \land f_2(\lfp{\!}{F_2})  \leq f_1(\lfp{\!}{F_1}) \displaybreak[0]\\[2pt]
\Leftrightarrow & \qquad \by{Kleene's Fixed Point Theorem, $F_1, F_2$ continuous}\displaybreak[0]\\[2pt]
& f_1\bigl(\sup_n F_1^n(\bot_1)\bigr) \leq f_2(\lfp{\!}{F_2})\\
& \land\: f_2\bigl(\sup_n F_2^n(\bot_2)\bigr)  \leq f_1(\lfp{\!}{F_1}) \displaybreak[0]\\[2pt]
\Leftrightarrow & \qquad \by{$f_1, f_2$ continuous}\displaybreak[0]\\[2pt]
& \sup_n f_1 \bigl( F_1^n(\bot_1)\bigr) \leq f_2(\lfp{\!}{F_2})\\
& \land\: \sup_n f_2 \bigl(F_2^n(\bot_2)\bigr)  \leq f_1(\lfp{\!}{F_1}) \displaybreak[0]\\[2pt]
\Leftarrow & \qquad \by{$\forall n\mydot a_n \leq S \implies \sup_n a_n \leq S$}\displaybreak[0]\\[2pt]
& \forall n\mydot f_1 \bigl( F_1^n(\bot_1)\bigr) \leq f_2(\lfp{\!}{F_2})\\
& \land\: \forall n\mydot f_2 \bigl(F_2^n(\bot_2)\bigr)  \leq f_1(\lfp{\!}{F_1}) \displaybreak[0]\\[2pt]
\end{array}
\end{align*}
We prove the above pair of inequalities by induction on $n$. We exhibit the
details only for the first one; the second one follows a similar argument. The
base case  $f_1 \bigl( F_1^0(\bot_1)\bigr) \leq f_2(\lfp{\!}{F_2})$ follows from
hypothesis \ref{it-2}. For the inductive case $ f_1 \bigl(
F_1^{n+1}(\bot_1)\bigr) \leq f_2(\lfp{\!}{F_2})$ we reason as follows:
\begin{align*}
\begin{array}{c@{\:\:} l@{} }
& f_1 \bigl(F_1^{n+1}(\bot_1)\bigr) \displaybreak[0]\\[2pt]
= & \qquad \by{def.~of $F^{n+1}$}\displaybreak[0]\\[2pt]
& f_1 \bigl(F \bigl(F_1^n(\bot_1)\bigr)\bigr) \displaybreak[0]\\[2pt]
\leq & \qquad \by{hyp.~\ref{it-1}}\displaybreak[0]\\[2pt]
& h_1 \bigl(f_1 \bigl(F_1^n(\bot_1)\bigr)\bigr) \displaybreak[0]\\[2pt]
\leq & \qquad \by{IH, monot.~of $h_1$}\displaybreak[0]\\[2pt]
& h_1 (f_2(\lfp{\!}{F_2})) \displaybreak[0]\\[2pt]
\leq & \qquad \by{hyp.~\ref{it-3}}\displaybreak[0]\\[2pt]
& f_2(\lfp{\!}{F_2})
\end{array} \\[-\normalbaselineskip]\tag*{\qedhere}
\end{align*}
  
\end{proof}

\paragraph{Proof of \autoref{thm:eet-prop}.} 
The proof of all properties proceeds by induction on the program
structure. Except for the case of probabilistic choice and procedure call, all
other programs constructs have already been dealt with in
\cite{Kaminski:arXiv:2016,Kaminski:ETAPS:2016}. For probabilistic choice we
follow the same reasoning as for conditional branches. We are left to analyze
then only the case of procedure calls. For each of the properties we reason as
follows:

\medskip
\noindent \emph{Continuity.} Let
$F(\eta) = \ctertenv{1} \oplus \eeet{\decl(\PName)}{\eta}$. 

\[
\begin{array}{c@{\:\:} l@{} }
& \eetd{\Call{\PName}}{\decl}\!(\sup_n \rt_n) \displaybreak[0]\\[2pt]
= & \qquad \by{\autoref{thm:eet-proc-call-iter}}\displaybreak[0]\\[2pt]
& \sup_m F^m(\bot_\RtEnv) (\sup_n \rt_n) \displaybreak[0]\\[2pt]
= & \qquad \by{$F^m(\bot_\RtEnv)$ continuous; see below}\displaybreak[0]\\[2pt]
& \sup_m \sup_n F^m(\bot_\RtEnv)(\rt_n) \displaybreak[0]\\[2pt]
= & \qquad \by{\autoref{thm:diag-sup-cpo}}\displaybreak[0]\\[2pt]
& \sup_n \sup_m F^m(\bot_\RtEnv)(\rt_n) \displaybreak[0]\\[2pt]
= & \qquad \by{\autoref{thm:eet-proc-call-iter}}\displaybreak[0]\\[2pt]
& \sup_n \eetd{\Call{\PName}}{\decl}\!(\rt_n) \displaybreak[0]\\[2pt]
\end{array}
\]
We are only left to prove that $F^m(\bot_\RtEnv)$ is continuous for all
$m \in \Nats$. We prove this by induction on $m$. The base case is immediate
since $F^0(\bot_\RtEnv) = \bot_\RtEnv$ and $\bot_\RtEnv$ is continuous. For the
inductive case we have $F^{m+1}(\bot_\RtEnv) = F (F^m(\bot_\RtEnv))$. The
continuity of $F^{m+1}(\bot_\RtEnv)$ follows from the I.H. and the fact that 
$F$ preserves continuity, \ie $\eta$ continuous implies $F(\eta)$
continuous (see \autoref{thm:eeet-cont-pres}).

\medskip
\noindent \emph{Propagation of constants.} By letting $F(\eta) = \ctertenv{1} \oplus
\eeet{\decl(\PName)}{\eta}$ we can recast the property as $\lfp{\!}{F}
(\ctert{k} + \rt) = \ctert{k} + \lfp{\!}{F}(\rt)$, or equivalently, as 
\[
\bigl( \lambda \eta^\star \!\mydot \lambda \rt^\star \!\mydot \eta^\star (\ctert{k} +
\rt^\star)  \bigr) (\lfp{\!}{F}) \:=\: \bigl( \lambda \eta^\star \!\mydot \lambda
\rt^\star \!\mydot \ctert{k} + \eta^\star (\rt^\star)  \bigr) (\lfp{\!}{F})~.
\] 
To prove this equation, we apply \autoref{thm:fix-point-compos} with
instantiations
\begin{align*}
  F_1 &= F_2 = F\\
 f_1 &= \lambda \eta^\star \!\mydot \lambda \rt^\star \!\mydot \eta^\star (\ctert{k} +
\rt^\star) \\
f_2 &=  \lambda \eta^\star \!\mydot \lambda
\rt^\star \!\mydot \ctert{k} + \eta^\star (\rt^\star) \\
h_1 &= \lambda \eta^\star \!\mydot \lambda \rt^\star \!\mydot \ctert{1} +
      \eeet{\decl(\PName)}{\lambda t' \mydot \eta^\star(t' - \ctert{k})}
      (\ctert{k} + \rt^\star)\\
h_2 &= \lambda \eta^\star \!\mydot \lambda \rt^\star \!\mydot \ctert{k} + \ctert{1} +
      \eeet{\decl(\PName)}{\lambda t' \mydot \eta^\star(t') - \ctert{k}}
      (\rt^\star )
\end{align*}
and underlying $\omega$-cpos
$(\mathcal{D}_1,\leq_1) = (\mathcal{D}_2,\leq_2) = (\mathcal{D},\leq) =
(\RtEnv, \sqsubseteq)$
and botom elements $\bot_1 = \bot_2 = \bot = \bot_\RtEnv$. The application of
\autoref{thm:fix-point-compos} requires the continuity of $F$ which follows from
\autoref{thm:eeet-cont-env}, the continuity of $f_1$ and $f2$, which holds
because runtime environments are continuous by definition, and finally the
monotonicity of $h_1$ and $h_2$. This latter fact, together with the fact that
$h_1$ and $h_2$ are effectively well--defined (\ie have type
$\RtEnv \To \RtEnv$) can be proved with an inductive argument (on the structure
of $\decl(\PName)$).

We are left to discharge hypotheses \ref{it-1}--\ref{it-3} of
\autoref{thm:fix-point-compos}. A simple unfolding of the involved functions
yields $f_1(F(\eta)) \sqsubseteq  h_1(f_1(\eta))$ and  $f_2(F(\eta)) \sqsubseteq
h_2(f_2(\eta))$ for all $\eta \in \RtEnv$; this establishes hypothesis
\ref{it-1}. As for hypothesis \ref{it-2}, $f_1(\bot_\RtEnv) \sqsubseteq f_2
(\lfp{\!}{F})$ holds because $f_1(\bot_\RtEnv) = \bot_\RtEnv$ and
$f_2(\bot_\RtEnv) \sqsubseteq f_1(\lfp{\!}{F})$ reduces to $\ctert{k} \preceq
\eetd{\Call{\PName}}{\decl}\!(\ctert{k} + t)$, which holds in view of the
monotonicity of transformer \eetsymbol\ and
\autoref{thm:ert-of-constant}. Finally, to discharge hypothesis \ref{it-3} we
reason as follows:

\begin{align*}
\begin{array}{c@{\:\:} l@{} }
& h_1(f_2(\lfp{\!}{F}))(t) \preceq f_2(\lfp{\!}{F})(t) \displaybreak[0]\\[2pt]
\Leftrightarrow & \qquad \by{def.~of $h_1,f_2, F$; let $\eta (t') = \ctert{k} +
    \eetd{\Call{\PName}}{\decl}\!(t' {-} \ctert{k})$}\displaybreak[0]\\[2pt]
& \ctert{1} + \eeet{\decl(\PName)}{\eta}(\ctert{k} + \rt) \preceq \ctert{k} +
  \eetd{\Call{\PName}}{\decl}\!(\rt)\displaybreak[0]\\[2pt]
\Leftrightarrow & \qquad \by{$\eta$ is constant separable into
                  $\eetd{\Call{\PName}}{\decl}$; \autoref{thm:cte-sep-rtenv}
             }\displaybreak[0]\\[2pt]
& \ctert{1} + \ctert{k} + \eeet{\decl(\PName)}{\eetd{\Call{\PName}}{\decl}} (t)
  \preceq \ctert{k} + \eetd{\Call{\PName}}{\decl}\!(\rt)\displaybreak[0]\\[2pt]
\Leftrightarrow & \qquad \by{def.~of $F$}\displaybreak[0]\\[2pt]
& \ctert{k} + F (\eetd{\Call{\PName}}{\decl}) (t)
  \preceq \ctert{k} + \eetd{\Call{\PName}}{\decl}\!(\rt)\displaybreak[0]\\[2pt]
\Leftrightarrow & \qquad \by{def.~of $\eetsymbol$}\displaybreak[0]\\[2pt]
& \ctert{k} + F (\lfp{\!}{F}) (t)
  \preceq \ctert{k} + \lfp{\!}{F}(\rt)\displaybreak[0]\\[2pt]
\Leftrightarrow & \qquad \by{def.~of $\lfpsymbol$}\displaybreak[0]\\[2pt]
& \ctert{k} + \lfp{\!}{F} (t)
  \preceq \ctert{k} + \lfp{\!}{F}(\rt)\displaybreak[0]\\[2pt]
\Leftarrow & \qquad \by{"$\preceq$" is a partial order}\displaybreak[0]\\[2pt]
& \true
\end{array} 
\end{align*}

\begin{align*}
\begin{array}{c@{\:\:} l@{} }
& h_2(f_1(\lfp{\!}{F}))(t) \preceq f_1(\lfp{\!}{F})(t) \displaybreak[0]\\[2pt]
\Leftrightarrow & \qquad \by{def.~of $h_2,f_1, F$; let $\upsilon (t') =
                  \eetd{\Call{\PName}}{\decl}\!(t' {+} \ctert{k}) - \ctert{k}$}\displaybreak[0]\\[2pt]
& \ctert{k} + \ctert{1} + \eeet{\decl(\PName)}{\upsilon}(\rt) \preceq 
  \eetd{\Call{\PName}}{\decl}\!(\ctert{k} + \rt)\displaybreak[0]\\[2pt]
\Leftrightarrow & \qquad \by{$\eetd{\Call{\PName}}{\decl}$ is constant separable into
                $\upsilon$; \autoref{thm:cte-sep-rtenv}
             }\displaybreak[0]\\[2pt]
& \ctert{k} + \ctert{1} + \bigl(
  \eeet{\decl(\PName)}{\eetd{\Call{\PName}}{\decl}} (\ctert{k} {+} \rt) -
  \ctert{k} \bigr)\displaybreak[0]\\[2pt]
& \preceq \eetd{\Call{\PName}}{\decl}\!(\ctert{k} + \rt)\displaybreak[0]\\[2pt]
\Leftrightarrow & \qquad \by{algebra; def.~of $F$}\displaybreak[0]\\[2pt]
& F (\eetd{\Call{\PName}}{\decl}) (\ctert{k} + \rt)
  \preceq \eetd{\Call{\PName}}{\decl}\!(\ctert{k} + \rt)\displaybreak[0]\\[2pt]
\Leftrightarrow & \qquad \by{def.~of $\eetsymbol$}\displaybreak[0]\\[2pt]
& F (\lfp{\!}{F}) (\ctert{k} + \rt)
  \preceq \lfp{\!}{F}(\ctert{k} +\rt)\displaybreak[0]\\[2pt]
\Leftrightarrow & \qquad \by{def.~of $\lfpsymbol$}\displaybreak[0]\\[2pt]
& \lfp{\!}{F} (\ctert{k} + \rt)
  \preceq \lfp{\!}{F}(\ctert{k} +\rt)\displaybreak[0]\\[2pt]
\Leftarrow & \qquad \by{"$\preceq$" is a partial order}\displaybreak[0]\\[2pt]
& \true
\end{array} 
\end{align*}

\medskip
\noindent \emph{Preservation of infinity.} By the monotonicity of
$\eetd{c}{\decl}$ and \autoref{thm:ert-of-constant}, we have 
\[
\eetd{c}{\decl}\!(\CteFun{\infty}) \succeq \ctert{k} \quad \forall k \in
\PosReals~,
\]
which itself entails $\eetd{c}{\decl}\!(\CteFun{\infty}) = \CteFun{\infty}$. 
% The property follows from a more
% general result that says that for all \Abort--free program $\prog{c}{\decl}$ and constant
% $k \in [0, \infty]$,
% %
% \[
% \rt \succeq \ctert{k} \implies \eetd{c}{\decl}\!(\rt) \succeq \ctert{k}~.
% \]
% The proof of this property also proceeds by induction on the program structure. 

\subsection{Relation between Transformers \boldeetsymbol\ and \boldwpsymbol}
\label{sec:eet-wp}
To establish \autoref{thm:eet-wp} we make use of a subsidiary result. This
result relies on the notion of \emph{separable} runtime environment. We
say that a runtime environment $\eta$ is \emph{separable} into runtimes
environments $\eta_1$ and $\eta_2$ iff we have $\eta(\rt_1 + \rt_2) =
\eta_1(\rt_1) + \eta_2(\rt_2)$ for every any two runtimes $\rt_1$ and $\rt_2$.

\begin{lemma}
\label{thm:eeet-ewp-sep}
 For every command $c$ and runtime environment $\eta$ separable into $\eta_1$ and
 $\eta_2$,
\[
\eeet{c}{\eta} \!(\rt_1 + \rt_2) \:=\: \eeet{c}{\eta_1} \!(\rt_1) +
\ewp{c}{\eta_2} \!(\rt_2)~.
\]
\end{lemma}
\begin{proof}
 For the basic
instructions ($\Skip$, $\Abort$ and assignment), the statement follows
immediately from the definitions of $\eetsymbol$ and $\wpsymbol$.  For the
remaining program constructs we reason as follows:

\medskip
\noindent \emph{Conditional Branching:}
\[
\begin{array}{c@{\:\:} l@{} }
& \eeet{\Cond{G}{c_1}{c_2}}{\eta}\!(\rt_1 + \rt_2) \displaybreak[0]\\[2pt]
= & \qquad \by{def.~of $\eeet{\cdot}{\eta}$}\displaybreak[0]\\[2pt]
& \ctert{1} + 
    \ToExp{G} \cdot \eeet{c_1}{\eta} \!(\rt_1 + \rt_2) + \ToExp{\lnot G} \cdot \eeet{c_2}{\eta}\!(\rt_1 + \rt_2) \displaybreak[0]\\[2pt]
= & \qquad \by{I.H. on $c_1$,$c_1$}\displaybreak[0]\\[2pt]
& \ctert{1} + 
    \ToExp{G} \cdot \bigl(\eeet{c_1}{\eta_1} \!(\rt_1) +
\ewp{c_1}{\eta_2} \!(\rt_2) \bigr) \displaybreak[0]\\[2pt]
& \quad\! + \ToExp{\lnot G} \cdot \bigl(\eeet{c_2}{\eta_1} \!(\rt_1) +
\ewp{c_2}{\eta_2} \!(\rt_2) \bigr)\displaybreak[0]\\[2pt]
= & \qquad \by{algebra}\displaybreak[0]\\[2pt]
& \ctert{1} + 
    \ToExp{G} \cdot \eeet{c_1}{\eta_1} \!(\rt_1) + 
 \ToExp{\lnot G} \cdot \eeet{c_2}{\eta_1} \!(\rt_1) \displaybreak[0]\\[2pt]
& \quad\! + \ToExp{G} \cdot \ewp{c_1}{\eta_2} \!(\rt_2) +
\ToExp{\lnot G} \cdot \ewp{c_2}{\eta_2} \!(\rt_2) \displaybreak[0]\\[2pt]
= & \qquad \by{def.~of $\eeet{\cdot}{\eta}, \ewp{\cdot}{\eta}$}\displaybreak[0]\\[2pt]
& \eeet{\Cond{G}{c_1}{c_2}}{\eta_1}\!(\rt_1) \displaybreak[0]\\[2pt]
& \quad\! + \ewp{\Cond{G}{c_1}{c_2}}{\eta_2}\!(\rt_2)
\end{array}
\]

\medskip
\noindent \emph{Probabilistic Choice:} analogous to the conditional branching case.

\medskip
\noindent \emph{Sequential Composition:}
\[
\begin{array}{c@{\:\:} l@{} }
& \eeet{c_1;c_2}{\eta}\!(\rt_1 + \rt_2) \displaybreak[0]\\[2pt]
= & \qquad \by{def.~of $\eeet{\cdot}{\eta}$}\displaybreak[0]\\[2pt]
& \eeet{c_1}{\eta}\!\bigl(\eeet{c_2}{\eta}(\rt_1 + \rt_2)\bigr) \displaybreak[0]\\[2pt]
= & \qquad \by{I.H. on $c_2$}\displaybreak[0]\\[2pt]
& \eeet{c_1}{\eta}\!\bigl(\eeet{c_2}{\eta_1} \!(\rt_1) +
\ewp{c_2}{\eta_2} \!(\rt_2)  \bigr)
  \displaybreak[0]\\[2pt]
= & \qquad \by{I.H. on $c_1$}\displaybreak[0]\\[2pt]
& \eeet{c_1}{\eta_1}\!\bigl(\eeet{c_2}{\eta_1} \!(\rt_1) \bigr)  +  \ewp{c_1}{\eta_2}\!\bigl(\ewp{c_2}{\eta_2} \!(\rt_2)  \bigr)
  \displaybreak[0]\\[2pt]
= & \qquad \by{def.~of $\eeet{\cdot}{\eta}, \ewp{\cdot}{\eta}$}\displaybreak[0]\\[2pt]
& \eeet{c_1; c_2}{\eta}\!(\rt_1) + \ewp{c_1; c_2}{\eta}\!(\rt_2)
\end{array}
\]

\medskip
\noindent \emph{Procedure Call:}
\begin{align*}
\begin{array}{c@{\:\:} l@{} }
& \eeet{\Call{\PName}}{\eta}\!(\rt_1 + \rt_2) \displaybreak[0]\\[2pt]
= & \qquad \by{def.~of $\eeet{\cdot}{\eta}$}\displaybreak[0]\\[2pt]
& \eta(\rt_1 + \rt_2)\displaybreak[0]\\[2pt]
= & \qquad \by{$\eta$ sep.~into $\eta_1, \eta_2$}\displaybreak[0]\\[2pt]
& \eta_1(\rt_1) + \eta_2(\rt_2)\displaybreak[0]\\[2pt]
= & \qquad \by{def.~of $\eeet{\cdot}{\eta}, \ewp{\cdot}{\eta}$}\displaybreak[0]\\[2pt]
& \eeet{\Call{\PName}}{\eta_1}\!(\rt_1) + \ewp{\Call{\PName}}{\eta_2}\!(\rt_2)
  \hfill 
\end{array} \\[-\normalbaselineskip]\tag*{\qedhere}
\end{align*}
\end{proof}

\paragraph{Proof of \autoref{thm:eet-wp}.} 
The proof proceeds by induction on the program structure, but for the inductive
reasoning to work we need to consider a stronger statement, namely
\begin{equation} 
\label{eq:ert-wp}
\eetd{c}{\decl} \!(\rt_1 + \rt_2) \:=\: \eetd{c}{\decl} \!(\rt_1) +
\wpd{c}{\decl} \!(\rt_2)~.
\end{equation}
(We recover the original statement by taking $\rt_1 = \ctert{0}$). For all
program constructs $c$ different from a procedure call, establishing
\autoref{eq:ert-wp} follows exactly the same argument as that used in
\autoref{thm:eeet-ewp-sep} for establishing
\[
\eeet{c}{\eta} \!(\rt_1 + \rt_2) \:=\: \eeet{c}{\eta_1} \!(\rt_1) +
\ewp{c}{\eta_2} \!(\rt_2)
\]
since $\eeet{\cdot}{\eta}$ and $\eet{\cdot}$ obey the same definition rule for
such program constructs. 

For the case of a procedure call we have to prove that
\[
\eetd{\Call{\PName}}{\decl}(\rt_1 + \rt_2) \:=\: \eetd{\Call{\PName}}{\decl}(\rt_1) +  \wpd{\Call{\PName}}{\decl}(\rt_2)~.
\]
Since
\begin{align*}
  \eetd{\Call{\PName}}{\decl} &= \lfp{}{F} \;\text{where }
F(\eta) = \ctertenv{1} \oplus \eeet{\decl(\PName)}{\eta}\\
  \wpd{\Call{\PName}}{\decl} &=  \lfp{}{G} \;\text{where }
G(\theta) = \ewp{\decl(\PName)}{\theta}~,
\end{align*}
and both $F$ and $G$ are continuous (see \autoref{thm:ewp-cont-env} and
\autoref{thm:eeet-cont-env}), by Kleene's Fixed Point Theorem our statement can
be recast as
\
\begin{multline*}
\sup\nolimits_n F^n(\bot_\RtEnv) (\rt_1 + \rt_2) \:=\\
 \sup\nolimits_n F^n(\bot_\RtEnv)(\rt_1) + \sup\nolimits_n G^n(\bot_\SEnv)(\rt_2)~,
\end{multline*}
where $\bot_\SEnv = \lambda f\!:\!  \UEX\mydot \ctert{0}$,
$\bot_\RtEnv = \lambda \rt\!:\!  \Runtimes \mydot \ctert{0}$,
$F^n(\bot_\RtEnv) = F(\ldots F(F(\bot_\RtEnv)) \ldots)$ denotes the repeated
application of $F$ from $\bot_\RtEnv$ $n$ times and likewise for
$G^n(\bot_\SEnv)$. Since a standard property of complete partial orders ensures
that $F^n(\bot_\RtEnv)$ and $G^n(\bot_\SEnv)$ are monotonic \wrt $n$, we can use
the Monotone Sequence Theorem (\autoref{thm:MST}) to replace $\sup_n$ with
$\lim_{n \To \infty}$ in the above equation and this way ``merge'' the two
limits in the RHS into a single limit. The above equation is then entailed by
formula
\[
\forall n\mydot F^n(\bot_\RtEnv) (\rt_1 {+} \rt_2) = F^n(\bot_\RtEnv)(\rt_1) +
G^n(\bot_\SEnv)(\rt_2)~,
\]
% Formally, 
% %
% \[
% \begin{array}{c@{\:\:} l@{} }
% & \sup\nolimits_n F^n(\bot_\RtEnv) (\rt_1 + \rt_2) \:=\displaybreak[0]\\[2pt]
% & \qquad \sup\nolimits_n F^n(\bot_\RtEnv)(\rt_1) + \sup\nolimits_n G^n(\bot_\SEnv)(\rt_2)\displaybreak[0]\\[2pt]
% %
% \Leftarrow & \qquad \by{\autoref{thm:MST}}\displaybreak[0]\\[2pt]
% %
% & \lim\limits_{n \To \infty} F^n(\bot_\RtEnv) (\rt_1 + \rt_2) \:=\displaybreak[0]\\[2pt]
% & \qquad \lim\limits_{n \To \infty} F^n(\bot_\RtEnv)(\rt_1) + \lim\limits_{n \To \infty} G^n(\bot_\SEnv)(\rt_2)\\[5pt]
% %
% \Leftarrow & \qquad \by{algebra of limits}\displaybreak[0]\\[2pt]
% %
% & \lim\limits_{n \To \infty} F^n(\bot_\RtEnv) (\rt_1 + \rt_2) \:=\displaybreak[0]\\[2pt]
% & \qquad \lim\limits_{n \To \infty} F^n(\bot_\RtEnv)(\rt_1) + G^n(\bot_\SEnv)(\rt_2)\displaybreak[0]\\[2pt]
% %
% \Leftarrow & \displaybreak[0]\\[2pt]
% %
% & \forall n\mydot F^n(\bot_\RtEnv) (\rt_1 + \rt_2) \:= \displaybreak[0]\\[2pt]
% & \qquad  F^n(\bot_\RtEnv)(\rt_1) + G^n(\bot_\SEnv)(\rt_2)
% \end{array}
% \]
% %
which we prove by induction on $n$. The base case is immediate
since for every runtime $\rt$, $F^0(\bot_\RtEnv)(\rt) = G^0(\bot_\SEnv)(\rt) =
\ctert{0}$. For the inductive case we reason as follows:

\begin{align*}
\begin{array}{c@{\:\:} l@{} }
& F^{n+1}(\bot_\RtEnv) (\rt_1 + \rt_2) \:=\: \displaybreak[0]\\[2pt]
& \quad F^{n+1}(\bot_\RtEnv)(\rt_1) + G^{n+1}(\bot_\SEnv)(\rt_2) \displaybreak[0]\\[2pt]
\Leftrightarrow & \qquad \by{def.~$F^{n+1}, G^{n+1}$}\displaybreak[0]\\[2pt]
&  \ctert{1} + \eeet{\decl(\PName)}{F^{n}(\bot_\RtEnv)} (\rt_1 + \rt_2) \:=\displaybreak[0]\\[2pt]
& \quad \ctert{1} + \eeet{\decl(\PName)}{F^{n}(\bot_\RtEnv)}
  (\rt_1)  + \ewp{\decl(\PName)}{G^n(\bot_\SEnv)} (\rt_2) \displaybreak[0]\\[2pt]
\Leftrightarrow & \qquad \by{algebra}\displaybreak[0]\\[2pt]
&  \eeet{\decl(\PName)}{F^{n}(\bot_\RtEnv)} (\rt_1 + \rt_2) \:=\displaybreak[0]\\[2pt]
& \quad \eeet{\decl(\PName)}{F^{n}(\bot_\RtEnv)}
  (\rt_1)  + \ewp{\decl(\PName)}{G^n(\bot_\SEnv)} (\rt_2) \\[5pt]
\Leftarrow & \qquad \by{\autoref{thm:eeet-ewp-sep}, I.H.}\displaybreak[0]\\[2pt]
& \quad \true 
\end{array} \\[-\normalbaselineskip]\tag*{\qedsymbol}
\end{align*}

\subsection{Soundness of Proof Rules for \boldeetsymbol}
\label{sec:eet-rules-sound}

To establish the soundness of rules \lrule{eet-rec} and
\lrule{eet-rec$_\omega$} we make use of the following result.

\begin{fact}
\label{fact:deriv-eeet}
The derivability assertion 
\[
\deriv{\eet{\Call{\PName}}\!(\rt_1) \preceq u_1}{\eet{c}\!(\rt_2) \preceq u_2}
\]
implies that for every runtime environment $\eta$,
\[
\eta(\rt_1) \preceq u_1  \implies  
\eeet{c}{\eta}\!(\rt_2) \preceq u_2~.
\]
The result remain valid if we reverse all inequalities. 
\end{fact}
\noindent We have already used a similar result for establishing the soundness of rules
\lrule{wp-rec} and \lrule{wp-rec$_\omega$} (even though in that case the
conclusion was stated using $\wp{\cdot}$ instead of $\ewp{\cdot}{\theta}$).

\medskip
\noindent \textbf{Soundness of rule \lrule{eet-rec}.} Let runtime environment
$\eta^\star$ map $\rt$ to $u$ and all other runtimes to (the constant runtime)
$\boldsymbol{\infty}$. The validity of the rule follows from the following
reasoning:

\begin{align*}
\begin{array}{c@{\:\:} l@{} }
& \eetd{\Call{\PName}}{\decl}\!(\rt)
  \preceq \ctert{1} + u\displaybreak[0]\\[2pt]
\Leftrightarrow & \qquad \by{def. $\eetsymbol$ (\autoref{fig:eet})}\displaybreak[0]\\[2pt]
& \lfpsymbol_\sqsubseteq \bigl( \lambda \eta\!:\!\RtEnv \mydot
\ctertenv{1} \oplus \eeet{\decl(\PName)}{\eta} \bigr) (\rt) \:\preceq\: \ctert{1} + u\displaybreak[0]\\[2pt]
\Leftrightarrow & \qquad \by{def. $\eta^\star$,$\sqsubseteq$}\displaybreak[0]\\[2pt]
& \lfpsymbol_\sqsubseteq \bigl( \lambda \eta\!:\!\RtEnv \mydot
\ctertenv{1} \oplus \eeet{\decl(\PName)}{\eta} \bigr) \:\sqsubseteq\: \ctertenv{\ctert{1}} \oplus \eta^\star\displaybreak[0]\\[2pt]
\Leftarrow & \qquad \by{Park's Lemma\footnotemark, \autoref{fact:RtEnv-cpo}, \autoref{thm:eeet-cont-env}}\displaybreak[0]\\[2pt]
& \ctertenv{1} \oplus \eeet{\decl(\PName)}{\ctertenv{\ctert{1}} \oplus \eta^\star} \:\sqsubseteq\: \ctertenv{\ctert{1}} \oplus \eta^\star\displaybreak[0]\\[2pt]
\Leftrightarrow & \qquad \by{def. $\eta^\star$,$\sqsubseteq$}\displaybreak[0]\\[2pt]
& \ctert{1} + \eeet{\decl(\PName)}{\ctertenv{\ctert{1}} \oplus \eta^\star} \! (\rt)
  \:\preceq\: \ctert{1} + u\displaybreak[0]\\[2pt]
\Leftrightarrow & \qquad \by{algebra}\displaybreak[0]\\[2pt]
& \eeet{\decl(\PName)}{\ctertenv{\ctert{1}} \oplus \eta^\star} \!(\rt)
  \:\preceq\: u\displaybreak[0]\\[2pt]
\Leftarrow & \qquad \by{\autoref{fact:deriv-eeet}, rule premise}\displaybreak[0]\\[2pt]
& (\ctertenv{\ctert{1}} \oplus \eta^\star) (\rt) \preceq \ctert{1} + u\displaybreak[0]\\[2pt]
\Leftrightarrow & \qquad \by{def. $\eta^\star$}\displaybreak[0]\\[2pt]
& \true
\end{array} \\[-\normalbaselineskip]\tag*{\qedhere}
\end{align*}
\footnotetext{If $H\colon \mathcal{D}
  \To \mathcal{D}$ is an upper continuous function over an upper $\omega$--cpo
  $(\mathcal{D},\sqsubseteq)$ with bottom element, then $H(d) \sqsubseteq d$
  implies $\lfp{\sqsubseteq}{H} \sqsubseteq d$ for every $d \in
  \mathcal{D}$~\cite{Wechler:MTCS:92}.}

\medskip
\noindent \textbf{Soundness of rule
  \lrule{eet-rec$_\omega$}.} For simplicity, we consider the one--side version
of the rule for obtaining lower bound only:   
\[
\begin{array}{c}
\infrule{ l_0=\CteFun{0} \\
 \ctert{1} + l_n \preceq \eet{\Call{\PName}}\!(\rt) \derivsymbol l_{n+1} \preceq \eet{\decl(\PName)}\!(\rt)} 
{~\ctert{1} {+} \sup_n l_n \preceq \eetd{\Call{\PName}}{\decl}\!(\rt) }\\
\end{array}
\]
\smallskip

\noindent The reasoning for the orignal---two--side rule---is analogous. The validity of
the above rule follows from the following reasoning:
\begin{align*}
\begin{array}{c@{\:\:} l@{} }
& \ctert{1} + \sup_n l_n \preceq \eetd{\Call{\PName}}{\decl}\!(\rt)\displaybreak[0]\\[2pt]
\Leftrightarrow & \qquad \by{def. $\eetsymbol$ (\autoref{fig:eet}), $F(\eta) =\ctertenv{1} \oplus \eeet{\decl(\PName)}{\eta}$}\displaybreak[0]\\[2pt]
&  \ctert{1} + \sup_n l_n \preceq \lfp{\sqsubseteq}{F} (\rt)\displaybreak[0]\\[2pt]
\Leftrightarrow & \qquad \by{Kleene's Fixed Point Thm, \autoref{thm:eeet-cont-env}}\displaybreak[0]\\[2pt]
& \ctert{1} + \sup_n l_n \preceq \sup_n F^n(\bot_\RtEnv) (t)
\end{array} 
\end{align*}

\smallskip
 
\noindent Since $ F^n(\bot_\RtEnv)$ is monotonic \wrt $n$, $\sup_n
F^n(\bot_\RtEnv) = \sup_n F^{n+1}(\bot_\RtEnv)$ and the reasoning continues as follows:
\begin{align*}
\hspace*{-4.5em} \begin{array}{c@{\:\:} l@{} }
\Leftrightarrow & \displaybreak[0]\\[2pt]
&  \ctert{1} + \sup_n l_n \preceq \sup_n F^{n+1}(\bot_\RtEnv) (t)\displaybreak[0]\\[2pt]
\Leftrightarrow & \qquad \by{$k + \sup_n a_n = \sup_n k + a_n$}\displaybreak[0]\\[2pt]
&  \sup_n \ctert{1}  + l_n \preceq \sup_n F^{n+1}(\bot_\RtEnv) (t)\displaybreak[0]\\[2pt]
\Leftarrow & \displaybreak[0]\\[2pt]
& \forall n \mydot \ctert{1}  + l_n \preceq F^{n+1}(\bot_\RtEnv)(\rt)
\end{array} 
\end{align*}

\smallskip
\noindent  We prove the above statement by induction on $n$. For the base case we have
\begin{align*}
\begin{array}{c@{\:\:} l@{} }
&  \ctert{1}  + l_0 \preceq F^1(\bot_\RtEnv)(\rt)\displaybreak[0]\\[2pt]
\Leftrightarrow & \qquad \by{rule premise, def $F^1(\bot_\RtEnv)$}\displaybreak[0]\\[2pt]
& \ctert{1}  \preceq \ctert{1} + \eeet{\decl(\PName)}{\bot_\RtEnv}\!(\rt)\displaybreak[0]\\[2pt]
\Leftarrow & \qquad \by{$\eeet{\decl(\PName)}{\bot_\RtEnv}\!(\rt) \succeq \ctert{0}$} \displaybreak[0]\\[2pt]
& \true
\end{array} 
\end{align*}
For the inductive case we have
\begin{align*}
\begin{array}{c@{\:\:} l@{} }
&  \ctert{1}  + l_{n+1} \preceq F^{n+2}(\bot_\RtEnv)(\rt)\displaybreak[0]\\[2pt]
\Leftrightarrow & \qquad \by{def $F^{n+2}(\bot_\RtEnv)$}\displaybreak[0]\\[2pt]
& \ctert{1} + l_{n+1} \preceq \ctert{1} + \eeet{\decl(\PName)}{F^{n+1}(\bot_\RtEnv)}\!(\rt)\displaybreak[0]\\[2pt]
\Leftrightarrow & \qquad \by{algebra} \displaybreak[0]\\[2pt]
& l_{n+1} \preceq \eeet{\decl(\PName)}{F^{n+1}(\bot_\RtEnv)}\!(\rt)\\[4pt]
\Leftarrow & \qquad \by{\autoref{fact:deriv-eeet}, rule premise}\displaybreak[0]\\[2pt]
& \ctert{1} + l_n \preceq  F^{n+1}(\bot_\RtEnv) (\rt) \displaybreak[0]\\[2pt]
\Leftrightarrow & \qquad \by{$I.H.$} \displaybreak[0]\\[2pt]
& \true
\end{array} 
\end{align*}

\subsection{Operational Model of \pGCL}
\label{sec:operational-model}

\begin{definition}[Pushdown Markov Chains with Rewards]
	A \textbf{pushdown Markov chain with rewards (PRMC)} is a tuple $\boldsymbol{\Automaton} = (Q,\, q_\mathit{init},\, \Gamma,\, \gamma_0,\, \Delta,\, \rewardsymbol)$, where
	\begin{itemize}
		%\item $\Sigmaact$ is a finite set of actions,
		\item $Q$ is a countable set of control states,
		\item $q_{\mathit{init}} \in Q$ is the initial control state,
		\item $\Gamma$ is a finite stack alphabet,
		\item $\gamma_0 \in \Gamma$ is a special bottom--of--stack symbol,
		\item $\Delta \colon Q \times \Gamma \dashrightarrow \SD{Q} \times \big(\Gamma \setminus \{\gamma_0\}\big)^*$ (where $\SD{Q}$ denotes the set of probability distributions over $Q$) is a probabilistic transition relation,
		\item $\rewardsymbol\colon Q \rightarrow \PosReals$ is a reward function.
	\end{itemize}
	A \textbf{path of $\boldsymbol{\Automaton}$} is a finite sequence $\boldsymbol{\rho} = (q_0,\, \beta_0) \stackrel{a_1}{\longrightarrow} \cdots \stackrel{a_k}{\longrightarrow} (q_k,\, \beta_k)$, where $q_0 = q_\mathit{init}$, $\beta_0 = \gamma_0$, and for all $1 \leq i \leq k$ holds $\beta_i \in \gamma_0 \cdot \big(\Gamma\setminus\{\gamma_0\}\big)^*$ and $\exists\, \mu \in \SD{Q}$ and $\exists\,  \gamma_1 \in \Gamma$ and $\exists\, \gamma_2 \in \Gamma \setminus \{\gamma_0\} \cup \{\varepsilon\}$, such that $\Delta (q_{i-1},\, \gamma_1) = (\mu,\, \gamma_2)$ and $\beta_{i-1} = w \cdot \gamma_1$ and $\beta_i = w \cdot \gamma_2$ and $\mu(q_i) = a_i > 0$.
	The \textbf{set of paths in $\boldsymbol{\Automaton}$} is denoted by $\boldsymbol{\PathsP}$.
	In the following let $\rho = (q_0,\, \beta_0) \stackrel{a_1}{\longrightarrow} \cdots \stackrel{a_k}{\longrightarrow} (q_k,\, \beta_k)$.
	The \textbf{probability of $\boldsymbol\rho$} is given by $\boldsymbol{\Prob{\Automaton}{\rho}} = \prod_{i=1}^k a_i$ be a path.
	The \textbf{reward of a path $\boldsymbol\rho$} is given by $\boldsymbol{\rew{\rho}} = \Prob{\Automaton}{\rho} \cdot \sum_{i=0}^k \rew{q_i}$.
	The \textbf{expected reward for reaching a set of target states} $T \subseteq Q$ is given by $\boldsymbol{\ExpRew{\Automaton}{T}} = \sum_{\rho' \in P}\rew{\rho'}$ where $P = \{\rho' \in \PathsP ~|~ \rho' = (q_0,\, \beta_0) \stackrel{a_1}{\longrightarrow} \cdots \stackrel{a_j}{\longrightarrow} (q_j,\, \beta_j),\, q_j \in T,\, \forall \, 0 \leq \ell < j\colon q_\ell \not\in T\}$.
	We stick to the convention that an empty sum yields value zero, i.e.\ in particular $\sum_{\rho' \in \emptyset} \rew{\rho'} = 0$.
\end{definition}
We assume a given labeling for each program $c \in \Cmd$ that specifies the control flow of $c$ as illustrated in \autoref{sec:operational-model}. Let $\LabUsed$ denote the finite set of labels used in a given program $\Cmd$. We assume a special symbol $\Term$ to denote successful termination of a program. Furthermore, we make use of the following operations between statements and labels.
\begin{itemize}
 \item $\Init\colon \Cmd \to \LabUsed$ gives the label corresponding to the beginning of a given program.
 \item $\StmtOfLabelSymbol\colon \LabUsed \to (\Cmd \cup \{ \Term \})$ gives the statement associated to a label used in a program,
 \item $\SuccOneSymbol, \SuccTwoSymbol\colon \LabUsed \to \big(\LabUsed \cup \{ \Term \}\big)$ give the first and second successor label of a given program label. 
       In case $\ell \in \LabUsed$ has no such successor, we define $\SuccOne{\ell} = \Term$ and $\SuccTwo{\ell} = \Term$, respectively.
\end{itemize}
\begin{definition}[Operational PRMCs]
 Let $\sigma_0 \in \State$ and $f \in \UEX$. 
 The \emph{operational PRMC} of program $\prog{c}{\decl}$ starting in initial state $\sigma_0$ with respect to post--expectation $f$ is given by $\OPRMC{c,\decl}{\sigma_0}{f} = (Q,\, q_\mathit{init},\, \Gamma,\allowbreak\, \gamma_0,\allowbreak\, \Delta,\allowbreak\, \rewardsymbol)$ where 
 \begin{itemize}
  \item $Q ~=~ \big\{ (\ell,\sigma) ~|~ \ell \in \LabUsed \cup \{ \Term,\, \Sink \},\, \sigma \in \State \big\}$,
  \item $q_\mathit{init} ~=~ \OpState{\Init(c)}{\sigma_0}$,
  \item $\Gamma ~=~ \LabUsed \cup \{\gamma_0\}$,
  \item $\Delta$ is given by the least partial function satisfying the rules provided in \autoref{fig:operational},
  \item $\rew{\OpState{\Sink}{\sigma}} = f(\sigma)$ for each $\sigma \in \State$  and $\rew{q} = 0$, if $q$ is not of the form $\OpState{\Sink}{\sigma}$.
 \end{itemize}
\end{definition}

\subsection{Soundness of Transformer \boldwpsymbol}
\label{sec:eet-soundness}

\paragraph{Proof of \autoref{thm:correspondance}.}
For simplicity in the remainder we will assume the program declaration $\decl$
fixed and therefore, omit it. Consider first an automaton $\leftsuper{n}{\OPRMC{c}{\sigma}{f}}$ that behaves exactly the same as $\OPRMC{c}{\sigma}{f}$, but counts the number of symbols that currently lie on top of $\gamma_0$ on the stack and which self--loops if that number is exactly $n$ and $\OPRMC{c}{\sigma}{f}$ would perform another push onto the stack.
It is evident that 
 \begin{align*}
    \ExpRew{\OPRMC{c}{\sigma}{f}}{\mathcal{T}} ~~=~~ \sup_{n \in \Nats}~\ExpRew{\leftsuper{n}{\OPRMC{c}{\sigma}{f}}}{\mathcal{T}}~,
 \end{align*}
since $\leftsuper{n}{\OPRMC{c}{\sigma}{f}}$ exhibits a partial behavior of $\OPRMC{c}{\sigma}{f}$ in the sense that every path of $\leftsuper{n}{\OPRMC{c}{\sigma}{f}}$ that reaches $\mathcal T$ is (up to renaming) also a path of $\OPRMC{c}{\sigma}{f}$.
In the other direction, every path $\pi$ of $\OPRMC{c}{\sigma}{f}$ that reaches $\mathcal{T}$ can be implemented with finite stack size.
Therefore, there exists an $n_0 \in \Nats$ such that for all $n \geq n_0$ the path $\pi$ is also a path of $\leftsuper{n}{\OPRMC{c}{\sigma}{f}}$.

Consider now that by \autoref{thm:fp-rec} and its proof we can conclude that
\begin{align*}
\sup_{n \in \Nats} ~ \ewp{c}{\wp{\Calln{\PName}{n}{\decl}}}(f) ~~=~~ \wpd{c}{\decl}(f)~.
\end{align*}
It is therefore only left to show that the missing link
\begin{align*}
\lambda \sigma \mydot \ExpRew{\leftsuper{n}{\OPRMC{c}{\sigma}{f}}}{\mathcal{T}} ~=~ \ewp{c}{\wp{\Calln{\PName}{n}{\decl}}}(f)
\end{align*}
holds for all $n \in \Nats$.
The proof of this equality proceeds by induction on $n$:

\paragraph{The base case $\boldsymbol{n = 0}$:}
We have to show that
\begin{align*}
\lambda \sigma \mydot \ExpRew{\leftsuper{0}{\OPRMC{c}{\sigma}{f}}}{\mathcal{T}} ~=~ \ewp{c}{\wp{\Calln{\PName}{0}{\decl}}}(f)
\end{align*}
holds.
Whenever the automaton $\OPRMC{c}{\sigma}{f}$ would perform the push action associated with a procedure call, the automaton $\leftsuper{0}{\OPRMC{c}{\sigma}{f}}$ immediately self--loops as \emph{no} push to the stack whatsoever is allowed in this restricted automaton.
Therefore, we can syntactically replace every call in $c$ by an $\Abort$ and still obtain the same behavior for the corresponding restricted automaton. Formally,
\begin{align*}
\ExpRew{\leftsuper{0}{\OPRMC{c}{\sigma}{f}}}{\mathcal{T}} ~=~ \ExpRew{\leftsuper{0}{\OPRMC{c \,\subst{\Call{\PName}}{\Abort}}{\sigma}{f}}}{\mathcal{T}}~.
\end{align*}
Now, since syntactically $\Calln{\PName}{0}{\decl} = \Abort$ we have
\begin{align*}
\ewp{c}{\wp{\Calln{\PName}{0}{\decl}}}(f) ~=~ \ewp{c}{\Abort}(f)
\end{align*}
and therefore, it is left to show that
\begin{align*}
\lambda \sigma \mydot \ExpRew{\leftsuper{0}{\OPRMC{c \,\subst{\Call{\PName}}{\Abort}}{\sigma}{f}}}{\mathcal{T}} ~=~ \ewp{c}{\Abort}(f)
\end{align*}
holds.
The proof of this equality proceeds by structural induction on $c$:
For the base cases we have:

\subparagraph{The effectless program $\Skip$:}

On the denotational side, we have
\begin{align*}
\ewp{\Skip}{\Abort}(f)(\sigma) ~=~ f(\sigma)~.
\end{align*}

On the operational side we have $\Skip\subst{\Call{\PName}}{\Abort} = \Skip$. 
Let $\Init(\Skip) = \ell$, $\StmtOfLabel{\ell} = \Skip$, and $\SuccOne{\ell} = \Term$.
%The pushdown system $\OPRMC{\Skip}{\sigma}{f}$ is then given by:
%\todo[inline]{PICTURE GOES HERE.}
The only path of $\leftsuper{0}{\OPRMC{\Skip}{\sigma}{f}}$ reaching $\mathcal{T}$ is 
\begin{align*}
	\rho = \big(\OpState{\ell}{\sigma},\, \gamma_0\big) \stackrel{1}{\longrightarrow} \big(\OpState{\Term}{\sigma},\, \gamma_0\big) \stackrel{1}{\longrightarrow} \big(\OpState{\Sink}{\sigma},\, \gamma_0\big)
\end{align*} and its reward is 
\begin{align*}
1 \cdot 1 \cdot \big(0 + 0 + f(\sigma)\big) ~=~ f(\sigma)~.
\end{align*}
As $\rho$ is the only path reaching $\mathcal{T}$, we have 
\begin{align*}
    \ExpRew{\leftsuper{0}{\OPRMC{\Skip}{\sigma}{f}}}{ \mathcal{T} } ~=~ f(\sigma) ~=~ \ewp{\Skip}{\Abort}(f)(\sigma)~.
\end{align*}

\subparagraph{The diverging program $\Abort$:}

On the denotational side, we have 
\begin{align*}
	\ewp{\Abort}{\Abort}(f)(\sigma) ~=~ \CteFun{0}(\sigma) ~=~ 0~.
\end{align*}

On the operational side we have $\Abort\subst{\Call{\PName}}{\Abort} = \Abort$. 
Let $\Init(\Abort) = \ell$, $\StmtOfLabel{\ell} = \Abort$, and $\SuccOne{\ell} \allowbreak= \Term$.
%The pushdown system $\OPRMC{\Abort}{\sigma}{f}$ is then given by:
%\todo[inline]{PICTURE GOES HERE.}
The paths of $\leftsuper{0}{\OPRMC{\Abort}{\sigma}{f}}$ are all of the form 
\begin{align*}
	\big(\OpState{\ell}{\sigma},\, \gamma_0\big) \stackrel{1}{\longrightarrow}\big(\OpState{\ell}{\sigma},\, \gamma_0\big) \stackrel{1}{\longrightarrow}\big(\OpState{\ell}{\sigma},\, \gamma_0\big) \stackrel{1}{\longrightarrow} \cdots
\end{align*}
and none of them ever reaches $\mathcal{T}$.
Thus the expected reward is an empty sum and we therefore have 
\begin{align*}
    \ExpRew{\leftsuper{0}{\OPRMC{\Abort}{\sigma}{f}}}{ \mathcal{T} } ~=~ 0 ~=~ \ewp{\Abort}{\Abort}(f)(\sigma)~.
\end{align*}

\subparagraph{The assignment $\Ass{x}{E}$:}

On the denotational side, we have 
\begin{align*}
	\ewp{\Ass{x}{E}}{\Abort}(f)(\sigma) ~=~ &f\subst{x}{E}(\sigma)\\
	 ~=~ &f\left(\sigma\big[x \mapsto \sigma(E)\big]\right)~.
\end{align*}

On the operational side we have $\Ass x E\subst{\Call{\PName}}{\Abort}\allowbreak = \Ass x E$. 
Let $\Init(\Ass x E) = \ell$, $\StmtOfLabel{\ell} = \Ass x E$, and $\SuccOne{\ell} = \Term$.
%The pushdown system $\OPRMC{\Ass x E}{\sigma}{f}$ is then given by:
%\todo[inline]{PICTURE GOES HERE.}
The only path of $\leftsuper{0}{\OPRMC{\Ass x E}{\sigma}{f}}$ reaching $\mathcal{T}$ is 
\begin{align*}
	\rho ~=~ 	&\big(\OpState{\ell}{\sigma},\, \gamma_0\big) \stackrel{1}{\longrightarrow} \big(\OpState{\Term}{\sigma\big[x \mapsto \sigma(E)\big]},\, \gamma_0\big)\\
				& ~ \stackrel{1}{\longrightarrow} \big(\OpState{\Sink}{\sigma\big[x \mapsto \sigma(E)\big]},\, \gamma_0\big)
\end{align*}
 and its reward is 
\begin{align*}
	1 \cdot 1 \cdot \Big(0 + 0 + f\left(\sigma\big[x \mapsto \sigma(E)\big]\right)\Big) = f\left(\sigma\big[x \mapsto \sigma(E)\big]\right)\;.
\end{align*}
As $\rho$ is the only path reaching $\mathcal{T}$, we have 
\begin{align*}
	\ExpRew{\leftsuper{0}{\OPRMC{\Ass x E}{\sigma}{f}}}{\mathcal{T}} ~=~ &f\left(\sigma\big[x \mapsto \sigma(E)\big]\right)\\
	 ~=~ &\ewp{\Ass x E}{\Abort}(f)(\sigma)~.
\end{align*}

\subparagraph{The call $\Call{\PName}$:}

On the denotational side, we have 
\begin{align*}
	\ewp{\Call{\PName}}{\Abort}(f)(\sigma) &~=~ \ewp{\Abort}{\Abort}(f)(\sigma)
\end{align*}

On the operational side we have $\Call{\PName}\subst{\Call{\PName}}{\Abort} = \Abort$. 
Therefore, we can fall back to the base case $\Abort$.

\subparagraph{The inductive hypothesis on $c_1$ and $c_2$:} We now assume that for arbitrary but fixed programs $c_i$, with $i \in \{1,\, 2\}$, holds
\begin{align*}
\lambda \sigma \mydot \ExpRew{\leftsuper{0}{\OPRMC{c_i \,\subst{\Call{\PName}}{\Abort}}{\sigma}{f}}}{\mathcal{T}} ~=~ \ewp{c_i}{\Abort}(f)~.
\end{align*}
We can then proceed with the inductive steps:

\subparagraph{The sequential composition $c_1;c_2$:}

On the denotational side, we have 
\begin{align*}
	\ewp{c_1;c_2}{\Abort}(f)(\sigma) ~=~ \ewp{c_1}{\Abort}\left(\ewp{c_2}{\Abort}(f)\right)(\sigma)~.
\end{align*}

Operationally, we have 
\begin{align*}
	(c_1;c_2)\subst{\Call{\PName}}{\Abort} = c_1\subst{\Call{\PName}}{\Abort};c_2\subst{\Call{\PName}}{\Abort}~.
\end{align*}
%the pushdown system $\OPRMC{c_1;c_2}{\sigma}{f}$ is given by:
%\todo[inline]{PICTURE GOES HERE.}
We furthermore observe that any path of the automaton 
\begin{align*}
	\leftsuper{0}{\OPRMC{c_1\subst{\Call{\PName}}{\Abort};c_2\subst{\Call{\PName}}{\Abort}}{\sigma}{f}}
\end{align*}
reaching $\mathcal{T}$ is of the form
\begin{align*}
	\rho ~=~ &\big(\OpState{\Init(c_1\subst{\Call{\PName}}{\Abort})}{\sigma},\, \gamma_0\big) \stackrel{a_1}{\longrightarrow} \cdots \\
	& \stackrel{a_k}{\longrightarrow} \big(\OpState{\Term}{\sigma'},\, \gamma_0\big) \\
	& \stackrel{1}{\longrightarrow} \big(\OpState{\Init(c_2\subst{\Call{\PName}}{\Abort})}{\sigma'},\, \gamma_0\big) \xrightarrow{a_{k+2}} \cdots\\
	& \stackrel{a_{k'}}{\longrightarrow} \big(\OpState{\Term}{\sigma''},\, \gamma_0\big)\\
	& \stackrel{1}{\longrightarrow} \big(\OpState{\Sink}{\sigma''},\, \gamma_0\big)
\end{align*}
and any such a path's reward is given by 
\begin{align*}
	&\left.\prod_{i=1}^{k} \middle(a_i\right) \cdot \left(0 + \cdots + 0\vphantom{\left.\prod_{i=k+2}^{k'} \middle(a_i\right)}\right.\\
	&\left.~ + \left.\prod_{i=k+2}^{k'} \middle(a_i\right) \cdot \big(0 + \cdots + 0 + f(\sigma'')\big)\right) \\
	& ~=~ \left.\prod_{i=1}^{k} \middle(a_i\right) \cdot \left.\prod_{i=k+2}^{k'} \middle(a_i\right) \cdot f(\sigma'')
\end{align*}
Next, we observe that for any such path $\rho$ a suffix of it, namely 
\begin{align*}
	&\big(\OpState{\Init(c_2\subst{\Call{\PName}}{\Abort})}{\sigma'},\, \gamma_0\big) \stackrel{a_{k+2}}{\longrightarrow} \cdots \\
	& \stackrel{a_{k'}}{\longrightarrow} \big(\OpState{\Term}{\sigma''},\, \gamma_0\big) \stackrel{1}{\longrightarrow} \big(\OpState{\Sink}{\sigma''},\, \gamma_0\big)~,
\end{align*}
is a path of $\leftsuper{0}{\OPRMC{c_2\subst{\Call{\PName}}{\Abort}}{\sigma'}{f}}$ reaching $\mathcal{T}$ with reward 
\begin{align*}
	&\left.\prod_{i=k+2}^{k'} \middle(a_i\right) \cdot \big(0 + \cdots + 0 + f(\sigma'')\big)\\
	&  ~=~ \left.\prod_{i=k+2}^{k'} \middle(a_i\right) \cdot  f(\sigma'')~.
\end{align*}
Moreover, we can think of the expected reward of 
\begin{align*}
	\leftsuper{0}{\OPRMC{c_2\subst{\Call{\PName}}{\Abort}}{\sigma'}{f}}
\end{align*}
as an expectation
\begin{align*}
    \lambda \sigma'\mydot \ExpRew{\OPRMC{c_2\,\subst{\Call{\PName}}{\Abort}}{\sigma'}{f}}{ \mathcal{T} }~,
\end{align*}
which by the inductive hypothesis on $c_2$ is equal to 
\begin{align*}
	\ewp{c_2}{\Abort}(f)~.
\end{align*} 
Therefore, $\leftsuper{0}{\OPRMC{c_1\subst{\Call{\PName}}{\Abort}}{\sigma}{\ewp{c_2\,\subst{\Call{\PName}}{\Abort}}{\Abort}(f)}}$ and $\leftsuper{0}{\OPRMC{c_1\subst{\Call{\PName}}{\Abort};c_2\subst{\Call{\PName}}{\Abort}}{\sigma}{f}}$ have the same expected reward, as in the former all paths reaching $\mathcal{T}$ have the form
\begin{align*}
	&\big(\OpState{\Init(c_1\subst{\Call{\PName}}{\Abort})}{\sigma},\, \gamma_0\big) \stackrel{a_1}{\longrightarrow} \cdots \\
	& \stackrel{a_k}{\longrightarrow} \big(\OpState{\Term}{\sigma'},\, \gamma_0\big) \stackrel{1}{\longrightarrow} \big(\OpState{\Sink}{\sigma'},\, \gamma_0\big)
\end{align*}
and reward 
\begin{align*}
	&\left.\prod_{i=1}^{k} \middle(a_i\right) \cdot \big(0 {+ \cdots +} 0 + \ewp{c_2\subst{\Call{\PName}}{\Abort}}{\Abort}(f)(\sigma')\big)\\
	&~=~ \ewp{c_2\subst{\Call{\PName}}{\Abort}}{\Abort}(f)(\sigma')~.
\end{align*}
Keeping that in mind and applying the inductive hypothesis to $c_1$ now yields the desired statement:
\begin{align*}
    &\ExpRew{\leftsuper{0}{\OPRMC{c_1\,\subst{\Call{\PName}}{\Abort};c_2\,\subst{\Call{\PName}}{\Abort}}{\sigma}{f}}}{ \mathcal{T} } 	\\
    &~=~ \ExpRew{\OPRMC{c_1\,\subst{\Call{\PName}}{\Abort}}{\sigma}{\wp{c_2\subst{\Call{\PName}}{\Abort}}(f)}}{ \mathcal{T} } \\
    &~=~\ewp{c_1}{\Abort}\left(\ewp{c_2}{\Abort}(f)\right)(\sigma) \tag{I.H.\ on $c_1$}	 \\
	&~=~ \ewp{c_1;c_2}{\Abort}(f)(\sigma)
\end{align*}

\subparagraph{The conditional choice $\Cond{G}{c_1}{c_2}$:}

We distinguish two cases:

In Case 1 we have $\sigma \models G$.
Then on the denotational side, we have 
\begin{align*}
	& \ewp{\Cond{G}{c_1}{c_2}}{\Abort}(f)(\sigma) \\
	& ~=~\big(\ToExp{G} \cdot \ewp{c_1}{\Abort}(f) + \ToExp{\neg G} \cdot \ewp{c_2}{\Abort}(f)\big)(\sigma) \\ 
	& ~=~\ewp{c_1}{\Abort}(f)(\sigma) \tag{$\ToExp{G}(\sigma) = 1$ and $\ToExp{\neg G}(\sigma) = 0$}
\end{align*}

On the operational side we have 
\begin{align*}
	&\big(\Cond{G}{c_1}{c_2}\big)\subst{\Call{\PName}}{\Abort}\\
	& ~=~ \Cond{G}{c_1\subst{\Call{\PName}}{\Abort}}{c_2\subst{\Call{\PName}}{\Abort}}~.
\end{align*}
Regarding the control flow, let the following hold:\\
 $\Init(\Cond{G}{c_1\subst{\Call{\PName}}{\Abort}}{c_2\subst{\Call{\PName}}{\Abort}})\allowbreak = \ell$,\\
 $\StmtOfLabel{\ell} = \Cond{G}{c_1\subst{\Call{\PName}}{\Abort}}{c_2\subst{\Call{\PName}}{\Abort}}$,\\
 $\SuccOne{\ell} = \Init(c_1\subst{\Call{\PName}}{\Abort})$, and finally\\
 $\SuccTwo{\ell} = \Init(c_2\subst{\Call{\PName}}{\Abort})$.
We observe that any path of $\leftsuper{0}{\OPRMC{\Cond{G}{c_1\subst{\Call{\PName}}{\Abort}}{c_2\subst{\Call{\PName}}{\Abort}}}{\sigma}{f}}$ finally reaching $\mathcal{T}$ is of the form
\begin{align*}
	\rho ~=~ &\big(\OpState{\ell}{\sigma},\, \gamma_0\big)\\
	 & \stackrel{1}{\longrightarrow} \big(\OpState{\Init(c_1\subst{\Call{\PName}}{\Abort})}{\sigma},\, \gamma_0\big) \stackrel{a_{2}}{\longrightarrow} \cdots \\
	& \stackrel{a_{k}}{\longrightarrow} \big(\OpState{\Term}{\sigma'},\, \gamma_0\big) \stackrel{1}{\longrightarrow} \big(\OpState{\Sink}{\sigma'},\, \gamma_0\big)
\end{align*}
and it's reward is given by
\begin{align*}
	&1\cdot\left.\prod_{i=2}^{k} \middle(a_i\right) \cdot \left(0 + 0 + \cdots + 0 + f(\sigma')\right)\\
	& ~=~ \left.\prod_{i=2}^{k} \middle(a_i\right) \cdot f(\sigma')~.
\end{align*}
Next, observe that removing from any such path $\rho$ the initial segment, i.e.\ removing $\big(\OpState{\ell}{\sigma},\, \gamma_0\big) \stackrel{1}{\longrightarrow}{}$, gives a path of the form
\begin{align*}
	&\big(\OpState{\Init(c_1\subst{\Call{\PName}}{\Abort})}{\sigma},\, \gamma_0\big) \stackrel{a_{2}}{\longrightarrow} \cdots\\
	& \stackrel{a_{k}}{\longrightarrow} \big(\OpState{\Term}{\sigma'},\, \gamma_0\big)  \stackrel{1}{\longrightarrow} \big(\OpState{\Sink}{\sigma'},\, \gamma_0\big)~,
\end{align*}
which is a path of $\leftsuper{0}{\OPRMC{c_1\subst{\Call{\PName}}{\Abort}}{\sigma}{f}}$ reaching $\mathcal{T}$ with reward 
\begin{align*}
	\left.\prod_{i=2}^{k} \middle(a_i\right) \cdot \big(0 + \cdots + 0 + f(\sigma')\big) ~=~ \left.\prod_{i=2}^{k} \middle(a_i\right) \cdot  f(\sigma')~.
\end{align*}
Notice that if we remove the initial segments from every path in $\Paths{\leftsuper{0}{\OPRMC{\Cond{G}{c_1\subst{\Call{\PName}}{\Abort}}{c_2\subst{\Call{\PName}}{\Abort}}}{\sigma}{f}}}$ we obtain exactly the set $\Paths{\leftsuper{0}{\OPRMC{c_1\subst{\Call{\PName}}{\Abort}}{\sigma}{f}}}$.
Thus
\begin{align*}
	\leftsuper{0}{\OPRMC{\Cond{G}{c_1\subst{\Call{\PName}}{\Abort}}{c_2\subst{\Call{\PName}}{\Abort}}}{\sigma}{f}}
\end{align*}  as well as $\leftsuper{0}{\OPRMC{c_1\subst{\Call{\PName}}{\Abort}}{\sigma}{f}}$ have the same expected reward.
This immediately yields the desired statement:
\begin{align*}
    &\ExpRew{\leftsuper{0}{\OPRMC{\big(\Cond{G}{c_1}{c_2}\big)\subst{\Call{\PName}}{\Abort}}{\sigma}{f}}}{ \mathcal{T} }\\
	& {=}\, \ExpRew{\leftsuper{0}{\OPRMC{\Cond{G}{c_1\subst{\Call{\PName}}{\Abort}}{c_2\subst{\Call{\PName}}{\Abort}}}{\sigma}{f}}}{ \mathcal{T} }\\
	& {=}\, \ExpRew{\leftsuper{0}{\OPRMC{c_1\subst{\Call{\PName}}{\Abort}}{\sigma}{f}}}{ \mathcal{T} }\\
	& {=}\, \ewp{c_1}{\Abort}(f)(\sigma)\tag{I.H.\ on $c_1$}\\
	& {=}\, \ewp{\Cond{G}{c_1}{c_2}}{\Abort}(f)(\sigma)
\end{align*}
The reasoning for Case 2, i.e. $\sigma \not\models G$, is completely analogous using the inductive hypothesis on $c_2$

\subparagraph{The probabilistic choice $\PChoice{c_1}{p}{c_2}$:} 

On the denotational side, we have 
\begin{align*}
	& \ewp{\PChoice{c_1}{p}{c_2}}{\Abort}(f)(\sigma) \\
	& ~=~\big(p \cdot \ewp{c_1}{\Abort}(f) + (1-p) \cdot \ewp{c_2}{\Abort}(f)\big)(\sigma)\\
	& ~=~p \cdot \ewp{c_1}{\Abort}(f)(\sigma) + (1-p) \cdot \ewp{c_2}{\Abort}(f)(\sigma)
\end{align*}

On the operational side we have
\begin{align*}
	&\big(\PChoice{c_1}{p}{c_2}\big)\subst{\Call{\PName}}{\Abort} \\
	&~=~ \PChoice{c_1\subst{\Call{\PName}}{\Abort}}{p}{c_2\subst{\Call{\PName}}{\Abort}}
\end{align*}
Let $\Init(\PChoice{c_1\subst{\Call{\PName}}{\Abort}}{p}{c_2\subst{\Call{\PName}}{\Abort}})\allowbreak = \ell$, $\StmtOfLabel{\ell} = \PChoice{c_1\subst{\Call{\PName}}{\Abort}}{p}{c_2\subst{\Call{\PName}}{\Abort}}$, let\\
$\SuccOne{\ell} = \Init(c_1\subst{\Call{\PName}}{\Abort})$, and let\\
$\SuccTwo{\ell} = \Init(c_2\subst{\Call{\PName}}{\Abort})$.
We observe that any path of $\leftsuper{0}{\OPRMC{\PChoice{c_1\subst{\Call{\PName}}{\Abort}}{p}{c_2\subst{\Call{\PName}}{\Abort}}}{\sigma}{f}}$ reaching $\mathcal{T}$ is either of the form
\begin{align*}
	\rho_1 ~=~ &\big(\OpState{\ell}{\sigma},\, \gamma_0\big) \\
	& \stackrel{p}{\longrightarrow} \big(\OpState{\Init(c_1\subst{\Call{\PName}}{\Abort})}{\sigma},\, \gamma_0\big) \stackrel{a_{2}}{\longrightarrow} \cdots \\
	& \stackrel{a_{k}}{\longrightarrow} \big(\OpState{\Term}{\sigma'},\, \gamma_0\big) \stackrel{1}{\longrightarrow} \big(\OpState{\Sink}{\sigma'},\, \gamma_0\big)
\end{align*}
and it's reward is given by
\begin{align*}
	&p \cdot \left(0 + \left.\prod_{i=2}^{k} \middle(a_i\right) \cdot \big(0 + \cdots + 0 + f(\sigma')\big)\right)\\
	& ~=~ p \cdot \left.\prod_{i=2}^{k} \middle(a_i\right) \cdot  f(\sigma')~,
\end{align*}
or it is of the form
\begin{align*}
	\rho_2 \,=\, &\big(\OpState{\ell}{\sigma},\, \gamma_0\big) \xrightarrow{1-p}\\
	& \big(\OpState{\Init(c_2\subst{\Call{\PName}}{\Abort})}{\sigma},\, \gamma_0\big) \stackrel{a_{2}'}{\longrightarrow} \cdots \\
	& \stackrel{a_{k'}'}{\longrightarrow} \big(\OpState{\Term}{\sigma''},\, \gamma_0\big) \stackrel{1}{\longrightarrow} \big(\OpState{\Sink}{\sigma''},\, \gamma_0\big)
\end{align*}
and it's reward is given by
\begin{align*}
	&(1-p) \cdot \left(0 + \left.\prod_{i=2}^{k'} \middle(a_i'\right) \cdot \big(0 + \cdots + 0 + f(\sigma'')\big)\right)\\
	& ~=~ (1-p) \cdot \left.\prod_{i=2}^{k'} \middle(a_i'\right) \cdot  f(\sigma'')~.
\end{align*}
Notice that there is a possibility to partition the set 
\begin{align*}
	\Paths{\leftsuper{0}{\OPRMC{\PChoice{c_1\subst{\Call{\PName}}{\Abort}}{p}{c_2\subst{\Call{\PName}}{\Abort}}}{\sigma}{f}}}
\end{align*} 
into two sets $P_p$ containing those paths starting with\\
$\big(\OpState{\ell}{\sigma},\, \gamma_0\big) \stackrel{p}{\longrightarrow} \big(\OpState{\Init(c_1\subst{\Call{\PName}}{\Abort})}{\sigma},\, \gamma_0\big)$, and a set $P_{1-p}$ containing those paths starting with $\big(\OpState{\ell}{\sigma},\, \gamma_0\big) \stackrel{1-p}{\longrightarrow} \big(\OpState{\Init(c_2\subst{\Call{\PName}}{\Abort})}{\sigma},\, \gamma_0\big)$.

Next, observe that removing from any path in $P_p$ the initial segment, i.e.\ removing $\big(\OpState{\ell}{\sigma},\, \gamma_0\big) \stackrel{p}{\longrightarrow}{}$, gives exactly the set $\Paths{\leftsuper{0}{\OPRMC{c_1\subst{\Call{\PName}}{\Abort}}{\sigma}{f}}}$.
The paths of $\leftsuper{0}{\OPRMC{c_1\subst{\Call{\PName}}{\Abort}}{\sigma}{f}}$ reaching $\mathcal{T}$ are of the form
\begin{align*}
	&\big(\OpState{\Init(c_1\subst{\Call{\PName}}{\Abort})}{\sigma},\, \gamma_0\big) \stackrel{a_{2}}{\longrightarrow} \cdots \\
	&\stackrel{a_{k}}{\longrightarrow} \big(\OpState{\Term}{\sigma'},\, \gamma_0\big) \stackrel{1}{\longrightarrow} \big(\OpState{\Sink}{\sigma'},\, \gamma_0\big)~,
\end{align*}
and have reward 
\begin{align*}
	&\left.\prod_{i=2}^{k} \middle(a_i\right) \cdot \big(0 + \cdots + 0 + f(\sigma')\big) ~=~ \left.\prod_{i=2}^{k} \middle(a_i\right) \cdot  f(\sigma').
\end{align*}

Dually, removing from any path in $P_{1-p}$ the initial segment, i.e.\ removing $\big(\OpState{\ell}{\sigma},\, \gamma_0\big) \xrightarrow{1-p}{}$, gives exactly the set 
\begin{align*}
\Paths{\leftsuper{0}{\OPRMC{c_2\,\subst{\Call{\PName}}{\Abort}}{\sigma}{f}}}~.
\end{align*}
The paths of $\OPRMC{c_2\subst{\Call{\PName}}{\Abort}}{\sigma}{f}$ reaching $\mathcal{T}$ are of the form
\begin{align*}
	&\big(\OpState{\Init(c_2\subst{\Call{\PName}}{\Abort})}{\sigma},\, \gamma_0\big) \stackrel{a_{2}'}{\longrightarrow} \cdots\\
	& \stackrel{a_{k'}'}{\longrightarrow} \big(\OpState{\Term}{\sigma''},\, \gamma_0\big) \stackrel{1}{\longrightarrow} \big(\OpState{\Sink}{\sigma''},\, \gamma_0\big)~,
\end{align*}
and have reward 
\begin{align*}
	&\left.\prod_{i=2}^{k'} \middle(a_i'\right) \cdot \big(0 + \cdots + 0 + f(\sigma')\big) ~=~ \left.\prod_{i=2}^{k'} \middle(a_i'\right) \cdot  f(\sigma'').
\end{align*}

Since $P_p$ and $P_{1-p}$ was a partition of the path set 
\begin{align*}
\Paths{\leftsuper{0}{\OPRMC{\PChoice{c_1\subst{\Call{\PName}}{\Abort}}{p}{c_2\subst{\Call{\PName}}{\Abort}}}{\sigma}{f}}}~,
\end{align*} we can conclude:
\begin{align*}
	&\ExpRew{\leftsuper{0}{\OPRMC{\PChoice{c_1\subst{\Call{\PName}}{\Abort}}{p}{c_2\subst{\Call{\PName}}{\Abort}}}{\sigma}{f}}}{\mathcal{T}}\\
	& ~=~ p \cdot \ExpRew{\leftsuper{0}{\OPRMC{c_1\subst{\Call{\PName}}{\Abort}}{\sigma}{f}}}{\mathcal{T}}\\
	&\qquad\qquad + (1-p) \cdot \ExpRew{\leftsuper{0}{\OPRMC{c_2\subst{\Call{\PName}}{\Abort}}{\sigma}{f}}}{\mathcal{T}}\\
	& ~=~ p \cdot \ewp{c_1}{\Abort}(f)(\sigma) + (1-p) \cdot \ewp{c_1}{\Abort}(f)(\sigma)\tag{I.H. on $c_1$ and $c_2$}\\
	& ~=~ \ewp{\PChoice{c_1}{p}{c_2}}{\Abort}(f)(\sigma)
\end{align*}

This ends the proof for the base case of the induction on $n$ and we can now state the inductive hypothesis:

\paragraph{Inductive hypothesis on $\boldsymbol n$:}
We assume that for an arbitrary but fixed $n \in \Nats$ holds
\begin{align*}
\lambda \sigma \mydot \ExpRew{\leftsuper{n}{\OPRMC{c}{\sigma}{f}}}{\mathcal{T}} ~=~ \ewp{c}{\wp{\Calln{\PName}{n}{\decl}}}(f)
\end{align*}
 \emph{for all programs} $c$.
 We can then proceed with the inductive step:

\paragraph{Inductive step $\boldsymbol{n \rightarrow n + 1}$:}
We now have to show that 
\begin{align*}
\lambda \sigma \mydot \ExpRew{\leftsuper{n+1}{\OPRMC{c}{\sigma}{f}}}{\mathcal{T}} ~=~ \ewp{c}{\wp{\Calln{\PName}{n+1}{\decl}}}(f)
\end{align*}
holds assuming the inductive hypothesis on $n$.
The proof of this equality proceeds quite analogously, again by structural induction on $c$:

\subparagraph{The base cases $\Skip$, $\Abort$, $\Ass x E$:}
The proofs for these base cases are completely analogous to the proofs conducted in the base case $n=0$.

\subparagraph{The procedure call $\Call{\PName}$:}
The procedure call is technically a base case in the structural induction on $c$ as it is an atomic statement.
It does, however, require using the inductive hypothesis on $n$.
The proof goes as follows:
By an argument on the transition relation $\Delta$ of $\leftsuper{n+1}{\OPRMC{\Call{\PName}}{\sigma}{f}}$ we see that 
\begin{align*}
	\ExpRew{\leftsuper{n+1}{\OPRMC{\Call{\PName}}{\sigma}{f}}}{\mathcal{T}} ~=~ \ExpRew{\leftsuper{n}{\OPRMC{\decl(\PName)}{\sigma}{f}}}{\mathcal{T}}~.
\end{align*}
To the right hand side, we can apply the inductive hypothesis on $n$ and then obtain the desired result:
\begin{align*}
	&\lambda \sigma \mydot \ExpRew{\leftsuper{n+1}{\OPRMC{\Call{\PName}}{\sigma}{f}}}{\mathcal{T}} \\
	&~=~  \lambda \sigma \mydot \ExpRew{\leftsuper{n}{\OPRMC{\decl(\PName)}{\sigma}{f}}}{\mathcal{T}}\\
	&~=~ \ewp{\decl(\PName)}{\Calln{\PName}{n}{\decl}}(f) \tag{I.H.\ on $n$}\\
	&~=~ \ewp{\Call{\PName}}{\Calln{\PName}{n+1}{\decl}}(f) 
\end{align*}

\subparagraph{Inductive hypothesis and all inductive steps:}
The inductive hypothesis and the proofs for the inductive steps are completely analogous to the inductive hypothesis and the proofs conducted in the base case $n=0$.
Exemplarily, we shall sketch the proof for the sequential composition: 
By a lengthy argument and application of the inductive hypothesis on $c_2$ (completely analog to the base case for $n=0$) one arrives at
\begin{align*}
&\ExpRew{\leftsuper{n+1}{\OPRMC{c_1;c_2}{\sigma}{f}}}{\mathcal{T}} \\
&~=~ \ExpRew{\leftsuper{n+1}{\OPRMC{c_1}{\sigma}{\ewp{c_2}{\wp{\Calln{n+1}{\PName}{\decl}}}}}}{\mathcal{T}}~.
\end{align*}
Applying the inductive hypothesis on $c_1$ then yields the desired result:
\belowdisplayskip=-1\baselineskip
\begin{align*}
&\lambda \sigma \mydot \ExpRew{\leftsuper{n+1}{\OPRMC{c_1}{\sigma}{\ewp{c_2}{\wp{\Calln{n+1}{\PName}{\decl}}}}}}{\mathcal{T}}\\
&~=~ \ewp{c_1}{\wp{\Calln{n+1}{\PName}{\decl}}}\left( \ewp{c_2}{\wp{\Calln{n+1}{\PName}{\decl}}}(f) \right)\\
&~=~ \ewp{c_1; c_2}{\wp{\Calln{n+1}{\PName}{\decl}}}(f)
\end{align*}

\subsection{Case Study}
\label{app:casestudy}
The omitted details for proving the second partial correctness property are provided in \autoref{fig:binarysearch:partial:not-exists}.
\begin{figure}[t]
\begin{align*}
    &\hphantom{\boldsymbol{\colon}}\;\; \text{\footnotesize $\textcolor{gray}{\frac{[\mathit{left} {<} \mathit{right}]}{\mathit{right} {-} \mathit{left} {+} 1} \sum_{i = \mathit{left}}^{\mathit{right}}\left(\!\!\!\!\begin{array}{l}\big[a[\mathit{i}] {<} \mathit{val}\big] {\cdot} g[\mathit{left}/\textsf{min}(i+1, \mathit{right})]\\+ \big[a[\mathit{i}] {>} \mathit{val}\big] {\cdot} g[\mathit{right}/\textsf{max}(i-1, \mathit{left})]\end{array}\!\!\!\!\right)}$}\\[-2pt]
    &\hphantom{\boldsymbol{\colon}}\;\; \text{\footnotesize $\textcolor{gray}{~ + [\mathit{left} = \mathit{right}] \cdot \big[a[\mathit{left}] \neq \mathit{val}\big]}$}\\[-2pt]
\text{\scriptsize 1} &\boldsymbol{\colon}\;\; \boldsymbol{\Ass{\mathit{mid}}{\textbf{\textsf{uniform}}(\mathit{left},\, \mathit{right})};}\\[-2pt]
    &\hphantom{\boldsymbol{\colon}}\;\; \text{\footnotesize $\textcolor{gray}{[\mathit{left} < \mathit{right}]\cdot\Big(\big[a[\mathit{mid}] < \mathit{val}\big] \cdot g[\mathit{left}/\cdots]}$}\\[-2pt]
 &\hphantom{\boldsymbol{\colon}}\;\; \text{\footnotesize $\textcolor{gray}{\qquad + \big[a[\mathit{mid}] > \mathit{val}\big] \cdot g[\mathit{right}/\cdots]}$}\\[-2pt]
    &\hphantom{\boldsymbol{\colon}}\;\; \text{\footnotesize $\textcolor{gray}{~ + [\mathit{left} \geq \mathit{right}] \cdot f}$}\\[-2pt]
\text{\scriptsize 2} &\boldsymbol{\colon}\;\; \boldsymbol{\textbf{\textsf{if}} ~ (\mathit{left} < \mathit{right})\{}\\[-2pt]
    &\hphantom{\boldsymbol{\colon}}\;\; \qquad\text{\footnotesize $\textcolor{gray}{\big[a[\mathit{mid}] {<} \mathit{val}\big] {\cdot} g[\mathit{left}/\cdots] + \big[a[\mathit{mid}] {>} \mathit{val}\big] {\cdot} g[\mathit{right}/\cdots]}$}\\[-2pt]
\text{\scriptsize 3} &\boldsymbol{\colon}\;\; \qquad\boldsymbol{\textbf{\textsf{if}} ~ (a[\mathit{mid}] < \mathit{val})\{}\\[-2pt]
    &\hphantom{\boldsymbol{\colon}}\;\; \qquad\qquad \text{\footnotesize $\textcolor{gray}{g[\mathit{left}/\textsf{min}(\mathit{mid} + 1,\, \mathit{right})]}$}\\[-2pt]
\text{\scriptsize 4} &\boldsymbol{\colon}\;\; \qquad\qquad\boldsymbol{\Ass{\mathit{left}}{\textbf{\textsf{min}}(\mathit{mid} + 1,\, \mathit{right})};}\\[-2pt]
    &\hphantom{\boldsymbol{\colon}}\;\; \qquad\qquad \text{\footnotesize $\textcolor{gray}{g}$}\\[-2pt]
\text{\scriptsize 5} &\boldsymbol{\colon}\;\; \qquad\qquad\boldsymbol{\Call B}\\[-2pt]
    &\hphantom{\boldsymbol{\colon}}\;\; \qquad\qquad \text{\footnotesize $\textcolor{gray}{f}$}\\[-2pt]
\text{\scriptsize 6} &\boldsymbol{\colon}\;\; \qquad\boldsymbol{\} ~ \textbf{\textsf{else}} ~ \{}\\[-2pt]
    &\hphantom{\boldsymbol{\colon}}\;\; \qquad\qquad\text{\footnotesize $\textcolor{gray}{\big[a[\mathit{mid}] > \mathit{val}\big] \cdot g[\mathit{right}/\cdots] + \big[a[\mathit{mid}] < \mathit{val}\big]}$}\\[-2pt]
\text{\scriptsize 7} &\boldsymbol{\colon}\;\; \qquad\qquad\boldsymbol{\textbf{\textsf{if}} ~ (a[\mathit{mid}] > \mathit{val})\{}\\[-2pt]
    &\hphantom{\boldsymbol{\colon}}\;\; \qquad\qquad\qquad \text{\footnotesize $\textcolor{gray}{g[\mathit{right}/\textsf{max}(\mathit{mid} - 1,\, \mathit{left})]}$}\\[-2pt]
\text{\scriptsize 8} &\boldsymbol{\colon}\;\; \qquad\qquad\qquad\boldsymbol{\Ass{\mathit{right}}{\textbf{\textsf{max}}(\mathit{mid} - 1,\, \mathit{left})};}\\[-2pt]
    &\hphantom{\boldsymbol{\colon}}\;\; \qquad\qquad\qquad \text{\footnotesize $\textcolor{gray}{g}$}\\[-2pt]
\text{\scriptsize 9} &\boldsymbol{\colon}\;\; \qquad\qquad\qquad\boldsymbol{\Call B}\\[-2pt]
    &\hphantom{\boldsymbol{\colon}}\;\; \qquad\qquad\qquad \text{\footnotesize $\textcolor{gray}{f}$}\\[-2pt]
\text{\scriptsize 10} &\boldsymbol{\colon}\;\; \qquad\qquad\boldsymbol{\} ~ \textbf{\textsf{else}} ~ \{}~ \text{\footnotesize $\textcolor{gray}{f}$} ~ \boldsymbol{\Skip}~ \text{\footnotesize $\textcolor{gray}{f}$} ~\boldsymbol{\}} ~ \text{\footnotesize $\textcolor{gray}{f}$} \\[-2pt]
\text{\scriptsize 11} &\boldsymbol{\colon}\;\; \qquad\boldsymbol{\}} ~ \text{\footnotesize $\textcolor{gray}{f}$} \\[-2pt]
\text{\scriptsize 12} &\boldsymbol{\colon}\;\; \boldsymbol{\} ~ \textbf{\textsf{else}} ~ \{}~ \text{\footnotesize $\textcolor{gray}{f}$} ~ \boldsymbol{\Skip}~ \text{\footnotesize $\textcolor{gray}{f}$} ~\boldsymbol{\}}~ \text{\footnotesize $\textcolor{gray}{f}$} 
\end{align*}
\caption{Proof that $\Call B$ finds an index at which the value at this position is unequal to $\mathit{val}$ when started in a sorted array $a[\mathit{left}\,..\,\mathit{right}]$ in which the value $\mathit{val}$ does not exist. We write $\text{\footnotesize \textcolor{gray}{$j$}} ~ \boldsymbol{C} ~ \text{\footnotesize \textcolor{gray}{$h$}}$ for $i \preceq \wp{C}(h)$. Recall that $g = [\mathit{left} \leq \mathit{right}] \cdot \big[\textsf{sorted}(\mathit{left},\, \mathit{right})\big] \cdot \big[\forall x \in [\mathit{left},\, \mathit{right}]\colon a[x] \neq \mathit{val}\big]$ and $f = \big[a[\mathit{mid}] \neq \mathit{val}\big]$, and that we assume $g \preceq \wlp{\Call B}(f)$.}
\label{fig:binarysearch:partial:not-exists}
\end{figure}
%
%

% % The bibliography should be embedded for final submission.
% \begin{thebibliography}{}
% \softraggedright
% \bibitem[Smith et~al.(2009)Smith, Jones]{smith02}
% P. Q. Smith, and X. Y. Jones. ...reference text...
% \end{thebibliography}

\end{document}